\documentclass[orivec]{llncs}
\usepackage[margin=1in]{geometry}
\usepackage{amssymb}
\usepackage{epstopdf}
\usepackage{xytree}
\usepackage[all]{xy}
\usepackage{amsmath}
\usepackage{bussproofs}
\usepackage{comment}
\usepackage{url}

\newcommand{\limp}[0]{\ensuremath{\rightarrow}}

\newcommand{\mimp}[0]{\ensuremath{{-\!*\;}}}

\newcommand{\simp}[0]{\ensuremath{\triangleright}}

\newcommand{\mand}[0]{\ensuremath{*}}

\newcommand{\lsbbi}[0]{\ensuremath{\mathit{LS}_{\mathit{BBI}}}}
\newcommand{\fvlsbbi}[0]{\ensuremath{\mathit{FVLS}_{\mathit{BBI}}}}
\newcommand{\ilsbbi}[0]{\ensuremath{\mathit{LS}^e_{\mathit{BBI}}}}
\newcommand{\iilsbbi}[0]{\ensuremath{\mathit{LS}^{\mathit{sf}}_{\mathit{BBI}}}}

\def\Ccal{\mathcal{C}}
\def\Gcal{\mathcal{G}}
\def\Scal{\mathcal{S}}
\def\subst{\mathit{subst}}

\def\Rel{\mathit{Rel}}
\def\Cbb{\mathbb{C}}

\def\fvu{\mathbf{u}}
\def\fvv{\mathbf{v}}
\def\fvw{\mathbf{w}}
\def\fvx{\mathbf{x}}
\def\fvy{\mathbf{y}}
\def\fvz{\mathbf{z}}

\def\cfr{\mathfrak{c}}

\begin{document}

\title{A Labelled Sequent Calculus for BBI: Proof Theory and Proof Search}

\author{Zh\'e H\'ou, Alwen Tiu and Rajeev Gor\'e}
\institute{
Research School of Computer Science\\ The Australian National University\\ Canberra, ACT 0200, Australia}

\maketitle

\begin{abstract}
We present a labelled sequent calculus for Boolean BI (BBI), a classical variant of the logic of Bunched Implication. The calculus is simple, sound, complete, and enjoys cut-elimination. We show that all the structural rules in the calculus, i.e., those rules that manipulate labels and ternary relations, can be localised around applications of certain logical rules, thereby localising the handling of these rules in proof search. Based on this, we demonstrate a free variable calculus that deals with the structural rules lazily in a constraint system. We propose a heuristic method to quickly solve certain constraints, and show some experimental results to confirm that our approach is feasible for proof search. Additionally, we show that different semantics for BBI and some axioms in concrete models can be captured by adding extra structural rules.
\end{abstract}

\section{Introduction}
\label{sec:intro}
The logic of bunched implications (BI) 
was introduced 
to reason about
resources
using
additive connectives $\land$, $\lor$, $\limp$, $\top$, $\bot$, and
multiplicative connectives $\top^*$, $\mand$, $\mimp$~\cite{Ohearn1999}. Both parts are
intuitionistic so BI is also Intuitionistic logic (IL) plus Lambek
multiplicative logic (LM). Changing the additive part to classical
logic gives Boolean BI (BBI). Replacing LM by multiplicative classical
linear logic gives Classical BI (CBI).  BI logics are closely related
to separation logic~\cite{Reynolds2002LICS}, a logic for proving
properties of programs. Thus, the semantics and proof theory of BI-logics,
particularly for proof search, are important in computer
science.

The ternary relational Kripke semantics of 
BBI-logics come in at least
three different flavours: non-deterministic (ND), partial
deterministic (PD),
and total deterministic (TD)~\cite{Wendling2010}. These semantics give
different logics w.r.t.\  validity, i.e., $BBI_{ND},BBI_{PD},BBI_{TD}$ respectively, and all are
undecidable~\cite{BrotherstonK10,Wendling2010}.
The purely syntactic proof theory of BBI also comes in three flavours:
Hilbert calculi~\cite{Pym2002,Galmiche2006}, display
calculi~\cite{Brotherston2010} and nested sequent
calculi~\cite{Park2013}. All are sound and complete w.r.t.\ the
ND-semantics~\cite{Galmiche2006,Brotherston2010,Park2013}.

In between the relational semantics and the purely syntactic proof
theory are the labelled tableaux of
Larchey-Wendling and Galmiche which are sound and complete
w.r.t. the 
PD-semantics~\cite{Wendling2009,Wendling2012}.
They remark that ``the adaptation of this tableaux system to
$BBI_{TD}$ should be straightforward (contrary to
$BBI_{ND}$)''~\cite{Wendling2013}.
We return to
these issues in Section~\ref{sec:conc}.

The 
structural rules of display calculi, especially the contraction rule
on structures, are impractical for backward proof search. Nested
sequents also face similar problems, and although Park et al. showed
the admissibility of contraction in an improved nested sequent
calculus, it contains other rules that explicitly contract structures.
Their iterative deepening automated theorem prover for BBI based on
nested sequents is terminating and incomplete for bounded depths, but
complete and potentially non-terminating for an unbounded
depth~\cite{Park2013}. The labelled tableaux of
Larchey-Wendling and Galmiche compile all structural rules into PD-monoidal
constraints, and are cut-fee complete for BBI$_{\rm
  PD}$ using a potentially infinite
counter-model construction~\cite{Wendling2012}. But effective proof search is only a ``perspective'' and is
left as further work~\cite[page~2]{Wendling2012}.

Surprisingly, many applications of BBI do not directly correspond to
it's widely used non-deterministic semantics. For example, separation
logic models are instances of partial deterministic
models~\cite{Wendling2010} while ``memory models'' for BBI are
restricted to have \textit{indivisible units}: ``the empty memory
cannot be split into non-empty pieces''~\cite{BrotherstonK10}. Our
goal is to give a labelled proof system for BBI based upon the
ND-semantics which easily extend to the PD- and TD-semantics, and also these other, more ``practical'', semantics.

Our labelled sequent calculus $\lsbbi$ for BBI adopts some features
from existing labelled tableaux for BBI~\cite{Wendling2009} and
existing labelled sequent calculi for modal logics~\cite{Negri2005}. 
Unlike these calculi, some $\lsbbi$-rules contain substitutions on labels.
From a proof-search perspective, labelled calculi
are no better
than display calculi since they
require extra-logical rules to explicitly encode the frame conditions
of the underlying (Kripke) semantics.
Such rules, which we refer
to simply as structural rules, are just as bad as display postulates
for
proof search 
since we may be forced to explore all potential models.  
As a step towards our goal,
we show that 
the applications of these structural
rules 
can be localised 
around logical rules. 
Thus 
these structural rules are only triggered by 
applications of logical
rules, leading to
a purely syntax-driven proof search procedure for $\lsbbi.$

Our work is novel from two perspectives. Compared to the labelled
tableaux of Larchey-Wendling and Galmiche,
we deal with the non-deterministic semantics of BBI, which
they have flagged as a difficulty,
and obtain a constructive cut-elimination procedure.
Compared to the nested sequent calculus of Park et al., our
calculus is much simpler, and generally gives much shorter
proofs. Note that Park et al. actually gave a labelled variant of
their nested sequent calculus, with the same logical rules as
ours. However, their structural rules are still just notational
variants of the original ones, which are lengthy and do not use
ternary relations. 
We also show that adding certain structural rules
to $\lsbbi$ allows us to 
obtain cut-free labelled calculi for all the other semantics mentioned above.

The rest of the paper is organized as follows. In Section~\ref{sec:bbi}, we present the semantics
of BBI, following~\cite{Wendling2010}, and our labelled sequent
calculus $\lsbbi$. We show that $\lsbbi$ is sound with respect to the semantics, and it is complete indirectly
via a Hilbert system for BBI~\cite{Galmiche2006} which is already shown complete.
In Section~\ref{sec:cutelim}, we prove some important proof theoretic properties of $\lsbbi$: invertibility
of inference rules, admissibility of contraction, and more importantly, cut-elimination. 
In Section~\ref{sec:local-struct-rules}, we discuss a permutation result for inference rules, allowing us to
isolate applications of structural rules during proof
search. 
In Section~\ref{sec:mapping-proof-search}, we describe how to reduce proof
search to constraint solving in a free-variable sequent calculus.
We give a heuristic method
for solving the resulting constraint problem in Section~\ref{sec:heur-proof-search} and report on experimental
results in Section~\ref{sec:experiment}.
Section~\ref{sec:conc} concludes the paper.  
Detailed proofs are available in the appendices.


\section{The Labelled Sequent Calculus for BBI}
\label{sec:bbi}
The semantics
of BBI is in Section~\ref{sec:sem_bbi}, then we present
our labelled calculus in Section~\ref{sec:lsbbi}. The soundness proof is outlined in Section~\ref{sec:sound}, followed by the completeness proof in Section~\ref{sec:comp}, in which the Hilbert system of BBI is used.

\subsection{Syntax and Semantics of BBI}
\label{sec:sem_bbi}
BBI formulae are defined inductively as follows, where $p$ is an
atomic proposition, $\top^*,\mand,\mimp$ are the multiplicative unit,
conjunction, and implication respectively:
\begin{center}
$A ::= p\mid \top\mid \bot\mid\lnot A\mid A\lor A\mid A\land A\mid A\limp A\mid \top^*\mid A\mand A\mid A\mimp A$
\end{center}

The labelled sequent calculus for BBI employs a ternary relation of
worlds that is based on a non-deterministic monoid structure, a l\`a Galmiche et
al.~\cite{Galmiche2006}.

A non-deterministic monoid is a triple ($\mathcal{M},\circ,\epsilon$) where
$\mathcal{M}$ is a non-empty set, $\epsilon \in \mathcal{M}$ and
$\circ : \mathcal{M} \times \mathcal{M} \rightarrow
\mathcal{P}(\mathcal{M})$. The extension of  $\circ$ to 
$\mathcal{P}(\mathcal{M})$ uses $X \circ Y = \bigcup \{x \circ y : x \in
X,y \in Y\}$. The following conditions hold in this monoid:
\begin{itemize}
\item{\em Identity:} $\forall a \in \mathcal{M}. \epsilon \circ a = \{a\}$ 
\item{\em Commutativity:} $\forall a,b \in \mathcal{M}. a \circ b = b \circ a$ 
\item{\em Associativity:} $\forall a,b,c \in \mathcal{M}. a \circ (b \circ c) = (a \circ b) \circ c$.
\end{itemize} 

The ternary relation over worlds is defined by $\simp \subseteq
\mathcal{M} \times \mathcal{M} \times \mathcal{M}$ such that $\simp(a,
b, c)$ if and only if $c \in a \circ b$. Following Galmiche et al., we
write $a,b \simp c$ instead of $\simp(a, b, c)$. We therefore have the
following conditions for all $a,b,c,d \in \mathcal{M}$:
\begin{itemize}
\item {\em Identity:} $\epsilon, a \simp b$ iff $a = b$ 
\item {\em Commutativity:} $a,b \simp c$ iff $b,a \simp c$ 
\item {\em Associativity:} If there exists $k$ s.t. $(a,k \simp d)$ and  $(b,c \simp k)$ 
then there exists $l$ s.t. $(a,b \simp l)$ and $(l,c \simp d)$.
\end{itemize}

Intuitively, 
the relation $x,y \simp z$ means that $z$ can be partitioned into two
parts: $x$ and $y$. The identity condition can be read as every world
can be partitioned into an empty world and itself. Commutativity
captures that partitioning $z$ into $x$ and $y$ is the same as
partitioning $z$ into $y$ and $x$. Finally, associativity means that
if $z$ can be partitioned into $x$ and $y$, and $x$ can further be
partitioned into $u$ and $v$, then all together $z$ consists of
$u$, $v$ and $y$. Therefore there must exist an element $w$ which is
the combination of $v$ and $y$, such that $w$ and $u$ form $z$.
Note that since we do not restrict this monoid to be cancellable, $(x, y \simp x)$ does not imply $y = \epsilon$.

Let $(\mathcal{M}, \simp, \epsilon)$ be a relational frame and $v:Var
\rightarrow \mathcal{P}(\mathcal{M})$ be a \textit{valuation}. A
\textit{forcing relation} ``$\Vdash$'' between elements of
$\mathcal{M}$ and BBI-formulae is defined 
as follows~\cite{Galmiche2006}:

\begin{center}
\begin{tabular}{lll}
$m \Vdash \top^*$ iff $m = \epsilon$ & & $m \Vdash P$ iff $P \in Var$ and $m \in v(P)$\\
$m \Vdash \bot$ iff never & & $m \Vdash A\lor B$ iff $m \Vdash A$ or $m \Vdash B$\\
$m \Vdash \top$ iff always & & $m \Vdash A\land B$ iff $m \Vdash A$ and $m \Vdash B$\\
$m \Vdash \lnot A$ iff $m \not \Vdash A$ & & $m \Vdash A\limp B$ iff $m\not \Vdash A$ or $m\Vdash B$\\
\multicolumn{3}{l}{$m \Vdash A\mand B$ iff $\exists a,b.(a,b\simp m$ and $a\Vdash A$ and $b \Vdash B)$}\\
\multicolumn{3}{l}{$m \Vdash A \mimp B$ iff $\forall a,b.((m,a \simp b$ and $a \Vdash A)$ implies $b \Vdash B)$}\\
\end{tabular}
\end{center}
A formula $A$ is true at $m \in \mathcal{M}$ if $m\Vdash A$ and is
{\em valid} if $m \Vdash A$ for every $m\in \mathcal{M}$ in every
model $(\mathcal{M}, \simp, \epsilon, v)$.

\subsection{The Labelled Sequent Calculus}
\label{sec:lsbbi}
The inference rules of our labelled system $\lsbbi$ are shown in
Figure~\ref{fig:LS_BBI}, where we  use $A$ to denote a formula, $w,x,y,z$ are in the set $LVar$
of label variables, and $\epsilon$ is the label constant. We define a
mapping $\rho: \{\epsilon\}\cup LVar \rightarrow \mathcal{M}$ from
labels to worlds. Note that we overload the notation so that the empty
world in the semantics and the label constant are both referred to as
$\epsilon$, the ternary relation in the semantics and the calculus are
both referred to as $\simp$. Therefore we impose the following condition
on mappings from labels to worlds: 
$\forall \rho. \rho(\epsilon) = \epsilon$. We shall assume this condition implicitly
in what follows.

A labelled formula $w:A$ means formula $A$ is true in world
$\rho(w)$. A relational atom $(x,y\simp z)$ is interpreted as
$\rho(x),\rho(y) \simp \rho(z)$ in the semantics. That is, a labelled formula $w:A$ is
true iff $\rho(w) \Vdash A$, and a relational atom $(x,y\simp z)$ is true iff
$\rho(x),\rho(y)\simp \rho(z)$ holds.

A sequent is of the form $\Gamma \vdash \Delta$, where $\Gamma$ and
$\Delta$ are \textit{structures}, the empty structure
is $\emptyset_a$ and $\Gamma$ and $\Delta$ are multisets of
labelled formulae and relational atoms, defined formally via:
\begin{align*}
\Gamma &::= \emptyset_a|w:A|(x,y \simp z)|\Gamma;\Gamma\\
\Delta &::= w:A|\Delta;\Delta
\end{align*}

\begin{definition}[Sequent Validity]
A sequent $\Gamma \vdash \Delta$ in $\lsbbi$ is valid if for all
$(\mathcal{M}, \simp, \epsilon)$, $v$ and $\rho$, if every member of $\Gamma$ is true then so is at least one
member of $\Delta$.
\end{definition}
Note that BBI-validity of a formula $A$ corresponds to validity of 
the sequent $ \vdash x : A$, where $x$ is an arbitrary label. This notion
of validity is also adopted in other work for BBI\cite{Wendling2010,Park2013} and CBI~\cite{BrostherstonCal10}. but  
is stronger than that used for BI~\cite{Pym2002},
where a valid sequent is defined as one with a multiplicative unit 
on the left hand side.
For example, the formula $\top^*$ is valid in BI, but it is
not a valid in our definition because the  sequent $\vdash x :
\top^*$ is not provable (although the
sequent $\vdash \epsilon : \top^*$ is provable). 
Translated to our setting, this would correspond to defining a valid formula
as one which is true in the world $\epsilon$.

In our definition of sequents, the structural connective ``$;$'' in
the antecedent means (additive) ``and'' whereas in the succedent it
means (additive) ``or''. This is slightly different from the
traditional sequent notation where ``$,$'' is used as the structural
connective. Our notation is consistent with sequent systems for the
family of Bunched Implication (BI) logics, where ``$;$'' is the
additive structural connective, and ``$,$'' is used to denote the
multiplicative structural connective. The multiplicative structural
connective is not explicitly presented in our sequent notation, but
as we shall see later, it is encoded implicitly in the relational
atoms.


The formula introduced in the conclusion of each rule
is the \textit{principal formula}, and the relational atom 
introduced in the conclusion of each rule is the \textit{principal
  relational atom}.

The semantics of $\mand$ involves an existential condition, so rules
$\mand L$ and $\mand R$ incorporate existential and universal
quantifiers respectively. Similarly, $\mimp L$ and $\mimp R$
incorporate universal and existential quantifiers
respectively.
Therefore, rules $\mand L$ and $\mimp R$ create a premise containing new
relations, and the labels in the created relation must
be fresh (except for the label of the principal formula). Rules $\mand R$
and $\mimp L$ create a premise using already existing relations from
the conclusion. Further, in rules
$A$ and $A_C$, the label $w$ must be fresh in the premise, as it
represents a new partition of the original world.

In the rule $\top^* L$, there is an operation of global substitution
$[\epsilon/x]$ in the premise. A substitution $\Gamma[y/x]$ is defined
in the usual way: replace every occurrence of $x$ in $\Gamma$ by $y$. 

\begin{figure*}[!t]
\centering
\begin{tabular}{cc}
\multicolumn{2}{l}{\textbf{Identity and Cut:}}\\[10px]
\AxiomC{$$}
\RightLabel{\tiny $id$}
\UnaryInfC{$\Gamma;w:P \vdash w:P;\Delta$}
\DisplayProof
&
\AxiomC{$\Gamma \vdash x:A;\Delta$}
\AxiomC{$\Gamma';x:A \vdash \Delta'$}
\RightLabel{$cut$}
\BinaryInfC{$\Gamma;\Gamma' \vdash \Delta;\Delta'$}
\DisplayProof
\\[15px]
\multicolumn{2}{l}{\textbf{Logical Rules:}}\\[10px]
\AxiomC{$$}
\RightLabel{\tiny $\bot L$}
\UnaryInfC{$\Gamma; w:\bot \vdash \Delta$}
\DisplayProof
$\qquad$
\AxiomC{$\Gamma[\epsilon/w] \vdash \Delta[\epsilon/w]$}
\RightLabel{\tiny $\top^* L$}
\UnaryInfC{$\Gamma;w:\top^* \vdash \Delta$}
\DisplayProof
&
\AxiomC{$$}
\RightLabel{\tiny $\top R$}
\UnaryInfC{$\Gamma \vdash w:\top;\Delta$}
\DisplayProof
$\qquad$
\AxiomC{$$}
\RightLabel{\tiny $\top^* R$}
\UnaryInfC{$\Gamma \vdash \epsilon:\top^*;\Delta$}
\DisplayProof\\[15px]
\AxiomC{$\Gamma;w:A;w:B \vdash \Delta$}
\RightLabel{\tiny $\land L$}
\UnaryInfC{$\Gamma;w:A\land B \vdash \Delta$}
\DisplayProof
&
\AxiomC{$\Gamma \vdash w:A;\Delta$}
\AxiomC{$\Gamma \vdash w:B;\Delta$}
\RightLabel{\tiny $\land R$}
\BinaryInfC{$\Gamma \vdash w:A\land B;\Delta$}
\DisplayProof\\[15px]
\AxiomC{$\Gamma \vdash w:A;\Delta$}
\AxiomC{$\Gamma; w:B \vdash \Delta$}
\RightLabel{\tiny $\limp L$}
\BinaryInfC{$\Gamma;w:A\limp B \vdash \Delta$}
\DisplayProof
&
\AxiomC{$\Gamma;w:A \vdash w:B;\Delta$}
\RightLabel{\tiny $\limp R$}
\UnaryInfC{$\Gamma \vdash w:A\limp B; \Delta$}
\DisplayProof\\[15px]
\AxiomC{$(x,y \simp z);\Gamma;x:A;y:B \vdash \Delta$}
\RightLabel{\tiny $\mand L$}
\UnaryInfC{$\Gamma;z:A\mand B \vdash \Delta$}
\DisplayProof
&
\AxiomC{$(x,y \simp z);\Gamma;x:A \vdash z:B;\Delta$}
\RightLabel{\tiny $\mimp R$}
\UnaryInfC{$\Gamma \vdash y:A\mimp B;\Delta$}
\DisplayProof\\[15px]
\multicolumn{2}{c}{
\AxiomC{$(x,y \simp z);\Gamma \vdash x:A;z:A\mand B;\Delta$}
\AxiomC{$(x,y \simp z);\Gamma \vdash y:B;z:A\mand B;\Delta$}
\RightLabel{\tiny $\mand R$}
\BinaryInfC{$(x,y \simp z);\Gamma \vdash z:A\mand B;\Delta$}
\DisplayProof
}\\[15px]
\multicolumn{2}{c}{
\AxiomC{$(x,y \simp z);\Gamma;y:A\mimp B \vdash x:A;\Delta$}
\AxiomC{$(x,y \simp z);\Gamma;y:A\mimp B; z:B \vdash \Delta$}
\RightLabel{\tiny $\mimp L$}
\BinaryInfC{$(x,y \simp z);\Gamma;y:A\mimp B \vdash \Delta$}
\DisplayProof
}
\\[15px]
\multicolumn{2}{l}{\textbf{Structural Rules:}}\\[15px]
\AxiomC{$(y,x \simp z);(x,y \simp z); \Gamma \vdash \Delta$}
\RightLabel{\tiny $E$}
\UnaryInfC{$(x,y \simp z); \Gamma \vdash \Delta$}
\DisplayProof
&
\AxiomC{$(u,w \simp z);(y,v \simp w);(x,y \simp z);(u,v \simp x);\Gamma \vdash \Delta$}
\RightLabel{\tiny $A$}
\UnaryInfC{$(x,y \simp z);(u,v \simp x);\Gamma \vdash \Delta$}
\DisplayProof\\[15px]
\AxiomC{$(x,\epsilon \simp x); \Gamma \vdash \Delta$}
\RightLabel{\tiny $U$}
\UnaryInfC{$ \Gamma \vdash \Delta$}
\DisplayProof
&
\AxiomC{$(x,w\simp x);(y,y\simp w);(x,y\simp x);\Gamma \vdash \Delta$}
\RightLabel{\tiny $A_C$}
\UnaryInfC{$(x,y\simp x);\Gamma \vdash \Delta$}
\DisplayProof\\[15px]
\AxiomC{$(\epsilon,w'\simp w');\Gamma[w'/w] \vdash \Delta[w'/w]$}
\RightLabel{\tiny $Eq_1$}
\UnaryInfC{$(\epsilon,w\simp w');\Gamma \vdash \Delta$}
\DisplayProof
&
\AxiomC{$(\epsilon, w' \simp w');\Gamma[w'/w] \vdash \Delta[w'/w]$}
\RightLabel{\tiny $Eq_2$}
\UnaryInfC{$(\epsilon,w'\simp w);\Gamma \vdash \Delta$}
\DisplayProof\\ \\ \\
\multicolumn{2}{l}{\textbf{Side conditions:}} \\
\multicolumn{2}{l}{
In $\top^* L$, $Eq_1$ and $Eq_2$,  $w \not = \epsilon$. 
}\\
\multicolumn{2}{l}{
In $\mand L$ and $\mimp R$, the labels $x$ and $y$
do not occur in the conclusion.}\\
\multicolumn{2}{l}{
In $A$ and $A_C$, the label $w$ does not occur in the conclusion.}
\end{tabular}

\caption{The (cut-free) labelled sequent calculus $\lsbbi$ for Boolean BI.}
\label{fig:LS_BBI}
\end{figure*}


The \textit{additive rules} ($\bot L$, $\top R$, $\land L$, $\land R$,
$\limp L$, $\limp R$) and the \textit{multiplicative rules} ($\top^*
L$, $\top^* R$, $\mand L$, $\mand R$, $\mimp L$, $\mimp R$)
respectively deal with the additive/ multiplicative connectives.  The
\textit{zero-premise rules} are those with no premise ($id$, $\bot L$,
$\top R$, $\top^* R$). Figure~\ref{fig:ex} shows an example derivation
of $\lsbbi$.

\begin{figure*}[!t]
\centering
\AxiomC{}
\RightLabel{\tiny $\top^* R$}
\UnaryInfC{$(\epsilon,a\simp a);(a,\epsilon\simp a);a:A \vdash \epsilon:\top^* $}
\AxiomC{}
\RightLabel{\tiny $id$}
\UnaryInfC{$(\epsilon,a\simp a);(a,\epsilon\simp a);a:A \vdash a:A$}
\RightLabel{\tiny $\mand R$}
\BinaryInfC{$(\epsilon,a\simp a);(a,\epsilon\simp a);a:A \vdash a:\top^*  \mand   A$}
\RightLabel{\tiny $E$}
\UnaryInfC{$(a,\epsilon\simp a);a:A \vdash a:\top^*  \mand   A$}
\RightLabel{\tiny $U$}
\UnaryInfC{$a:A \vdash a:\top^*  \mand   A$}
\RightLabel{\tiny $\limp R$}
\UnaryInfC{$ \vdash a:A \limp   (\top^*  \mand   A)$}
\DisplayProof
\caption{An example derivation in $\lsbbi$.}
\label{fig:ex}
\end{figure*}

Note that we start (at the bottom) by labelling the formula with an
arbitrary world $a$. Since provability is preserved by 
substitutions of labels (Lemma~\ref{subs}), provability of $\vdash a : F$ implies 
provability of $\vdash w : F$, for any world $w.$ 
Thus, if a formula is provable, then it is true in every world. 

\subsection{Soundness}
\label{sec:sound}
The soundness proof reasons about the falsifiability of sequents, which is defined as follows.

\begin{definition}[Sequent Falsifiability]
A sequent $\Gamma \vdash \Delta$ in $\lsbbi$ is falsifiable if there
exist some $(\mathcal{M}, \simp, \epsilon)$, $v$ and $\rho$, such that
every relational atom and labelled formula in $\Gamma$ is true and
every labelled formula in $\Delta$ is false, where:
\begin{quote}
$w:A$ is true iff $\rho(w) \Vdash A$\\
$w:A$ is false iff $\rho(w) \not \Vdash A$\\
$(x,y\simp z)$ is true iff $\rho(x),\rho(y)\simp \rho(z)$ holds
\end{quote} 
\end{definition} 

\begin{theorem}[Soundness]
\label{thm:soundness}
The labelled sequent calculus $\lsbbi$ is sound
w.r.t. the Kripke semantics for BBI.
\end{theorem}

\begin{proof}
To prove the soundness of $\lsbbi$, we show that each rule preserves
falsifiability upwards, as this is a more natural direction in terms
of backward proof search. Therefore to prove that a rule is sound, we
need to show that if the conclusion is falsifiable, then at least one
of the premises is falsifiable (usually in the same choice of $v$,
$\rho$, and $\mathcal{M}$). As the rules in $\lsbbi$ are designed based on the semantics, this is easy to verify. The details are in Appendix~\ref{app:sound_lsbbi}.\qed
\end{proof}

\subsection{Completeness}
\label{sec:comp}
We prove the completeness of $\lsbbi$ by showing that every 
derivation of a 
formula 
in the Hilbert system for BBI~\cite{Galmiche2006} can be 
mimicked 
in
$\lsbbi$,
possibly using cuts. 

The Hilbert system for BBI consists of the axioms and rules for
classical propositional logic for the additive fragment and 
additional axioms and rules for the 
multiplicative fragment. For the latter, we use the axiomatisation
given in \cite{Galmiche2006}, and listed in Figure~\ref{fig:hilb}.
We omit the axioms for classical propositional logic as they are standard,
and can be found in, e.g., \cite{Troelstra96}.
\begin{figure}[t!]
\textbf{Axioms}
\begin{enumerate}
\item $A \limp (\top^* \mand A)$
\item $(\top^* \mand A) \limp A$
\item $(A \mand B) \limp (B \mand A)$
\item $(A \mand (B \mand C)) \limp ((A \mand B) \mand C)$
\end{enumerate}

\textbf{Deduction Rules}
\begin{center}
\begin{tabular}{cc}
\AxiomC{$\vdash A$}
\AxiomC{$\vdash A \limp B$}
\RightLabel{$MP$}
\BinaryInfC{$\vdash B$}
\DisplayProof
&
\AxiomC{$\vdash A \limp C$}
\AxiomC{$\vdash B \limp D$}
\RightLabel{$\mand$}
\BinaryInfC{$\vdash (A \mand B) \limp (C \mand D)$}
\DisplayProof\\[10px]
\AxiomC{$\vdash A \limp (B \mimp C)$}
\RightLabel{$\mimp 1$}
\UnaryInfC{$\vdash (A \mand B) \limp C$}
\DisplayProof
&
\AxiomC{$\vdash (A \mand B) \limp C$}
\RightLabel{$\mimp 2$}
\UnaryInfC{$\vdash A \limp (B \mimp C)$}
\DisplayProof
\end{tabular}
\end{center}
\caption{Some axioms and rules for the Hilbert system for BBI.}
\label{fig:hilb}
\end{figure}

\begin{theorem}[Completeness]
\label{thm:completeness}
The labelled sequent calculus $\lsbbi$ is complete
w.r.t. the Kripke semantics for BBI.
\end{theorem}

\begin{proof}
Given a derivation $\Pi$ of a formula $A$ in the Hilbert system for
BBI, we show that one can construct an $\lsbbi$ derivation $\Pi'$ of the sequent
$\emptyset_a \vdash w : A$, for any label $w \not = \epsilon$.  It is enough to show that each axiom and each rule of the
Hilbert system can be derived. The derivations of the axioms in
$\lsbbi$ are straightforward; we show here a non-trivial case in the
derivation of the rules of the Hilbert system. Consider the rule
$\mimp 1$: Suppose $\Pi$ is the derivation:
\begin{prooftree}
\AxiomC{$\Pi_1$}
\noLine
\UnaryInfC{$A \limp (B \mimp C)$}
\RightLabel{$\mimp 1$}
\UnaryInfC{$(A \mand B) \limp C$}
\end{prooftree}
The $\lsbbi$ derivation $\Pi'$ is in Figure~\ref{fig:completeness}, 
where $\Pi_1'$ comes from $\Pi_1$ via the induction hypothesis, $\Pi_2$ is the upper derivation in Figure~\ref{fig:completeness}, 
and $\Gamma = \{(w_1,w_2 \simp w) ; w_1 : A ; w_2 : B \}$).
\begin{figure*}
\begin{prooftree}
\AxiomC{$~$}
\RightLabel{$id$}
\UnaryInfC{$\Gamma \vdash w_1 : A$}
  \AxiomC{$~$}
  \RightLabel{$id$}
  \UnaryInfC{$\Gamma ; w_1 : B \mimp C \vdash w_2 : B$}
  \AxiomC{$~$}
  \RightLabel{$id$}
  \UnaryInfC{$\Gamma ; w_1 : B \mimp C ; w : C \vdash w : C$}
\RightLabel{$\mimp L$}
\BinaryInfC{$\Gamma ; w_1 : B \mimp C \vdash w : C$}
\RightLabel{$\limp L$}
\BinaryInfC{$\Gamma ; w_1 : A \limp (B \mimp C) \vdash w : C$}
\end{prooftree}

\begin{prooftree}
\AxiomC{$\Pi_1'$}
\noLine
\UnaryInfC{$\vdash w_1 : A \limp (B \mimp C)$}
\AxiomC{$\Pi_2$}
\noLine
\UnaryInfC{$(w_1,w_2 \simp w) ; w_1 : A \limp (B \mimp C) ; 
         w_1 : A ; w_2 : B \vdash w : C$}
\RightLabel{$cut$}
\BinaryInfC{$(w_1,w_2 \simp w) ; w_1 : A ; w_2 : B \vdash w : C$}
\UnaryInfC{$w : A \mand B \vdash w : C$}
\RightLabel{$\limp R$}
\UnaryInfC{$\vdash w : (A \mand B) \limp C$}
\end{prooftree}
\caption{A derivation of the rule $\mimp 1$ putting $\Gamma = \{(w_1,w_2 \simp w) ; w_1 : A ; w_2 : B \}$).}
\label{fig:completeness}
\end{figure*}
\qed
\end{proof}

\begin{corollary}[Formula validity]
A BBI formula $A$ is valid iff $\emptyset_a \vdash w:A$ is derivable in $\lsbbi$, for any arbitrary $w$.
\end{corollary}
\begin{proof}
Follows from the soundness and completeness proof. Since $w$ is
arbitrary, $A$ is true at any world for any valuation $v$, mapping
$\rho$, and monoid structure $(\mathcal{M}, \simp, \epsilon)$.\qed
\end{proof}


\section{Cut-elimination}
\label{sec:cutelim}

This section proves the cut-elimination theorem for our labelled
sequent calculus. The general proof outlined here is similar to the cut-elimination proof for
labelled systems for modal logic~\cite{Negri2005}, i.e., we start by proving a substitution lemma for
labels, followed by proving the
invertibility of inference rules, weakening admissibility, and
contraction admissibility, before proceeding to the main cut-elimination proof. As there are many case analyses in these proofs, we only outline the important parts here. More details are available in Appendix~\ref{app:proofs}

Given a derivation $\Pi$, its {\em height}
$ht(\Pi)$ is
defined as the length of a longest branch in the derivation tree of $\Pi.$


The substitution lemma shows that provability is preserved under arbitrary substitutions of
labels. 

\begin{lemma}[Substitution]
\label{subs}
If
$\Pi$ is an $\lsbbi$ derivation for the sequent $\Gamma \vdash
\Delta$
then there is an $\lsbbi$ derivation $\Pi'$ of 
the sequent $\Gamma[y/x] \vdash \Delta[y/x]$ 
where every occurrence
of label $x$ ($x \not = \epsilon$) is replaced by label $y$, such that
$ht(\Pi') \leq ht(\Pi)$.
\end{lemma}

\begin{proof}
By induction on $ht(\Pi).$
We do case analyses on the last rule of $\Pi$. Most of the cases are
similar to Negri's labelled calculus for modal logic~\cite{Negri2005},
the only non-trivial cases are when the last rule is either
$\top^* L$, $Eq_1$ or $Eq_2$, and the labels $x$ or $y$ are used in the
principal formula/relational atom. The full proof is in Appendix~\ref{app:subs_labels}.\qed
\end{proof}

Admissibility of weakening is proved by a simple induction on the
height of derivations so 
we state the lemma
sans
proof. 

\begin{lemma}[Weakening admissibility]
\label{lm:weak}
If $\Gamma \vdash \Delta$ is derivable in $\lsbbi$, then for all structures $\Gamma'$ and $\Delta'$, 
the sequent 
$\Gamma;\Gamma' \vdash \Delta;\Delta'$ is derivable with the same
height in $\lsbbi$.
\end{lemma}

Combining Lemma~\ref{subs} and Lemma~\ref{lm:weak}, 
we can replace a formula that is never used in a derivation by any structure. More supplementary lemmas related to weakening are listed in Appendix~\ref{app:weak_lsbbi}.



\begin{lemma}[Invertibility of rules]
\label{lm:invert}
If $\Pi$ is a cut-free $\lsbbi$ derivation of the conclusion of a rule,
then there is a cut-free $\lsbbi$ derivation for each premise, with
height  at most $ht(\Pi).$
\end{lemma}
\begin{proof}
  Most of the rules are trivially invertible.  The proofs for
  the additive rules are similar to those
  for the additive rules from labelled calculi for modal logic or $G3c$
  (cf.~\cite{Negri2001}) since the rules are the same.
  The slightly non-trivial cases for the rules
  involving substitutions of labels follow from Lemma~\ref{subs}. The proof is detailed in Appendix~\ref{app:invert}.\qed
\end{proof}

The proof of the admissibility of contraction on additive formulae is similar to that
for classical sequent calculus since the  $\lsbbi$ rules for these
connectives are the same. In the multiplicative
rules, the principal formula is retained in the premise, so admissibility of
contraction on multiplicative formulae follows trivially. 
We need to prove that contraction on relational atoms is admissible,
as stated in the next lemma.
 
\begin{lemma}
\label{ctr_r}
For all structures $\Gamma,\Delta$, and ternary relations $(x,y\simp z)$:
if \;$\Pi$ a cut-free $\lsbbi$ derivation of 
$(x,y\simp z);(x,y\simp z);\Gamma \vdash \Delta$, then there is a cut-free
$\lsbbi$ derivation $\Pi'$ of $(x,y\simp z);\Gamma \vdash \Delta$ with $ht(\Pi') \leq ht(\Pi).$
\end{lemma}
\begin{proof}{\em (Outline)}
Let $n = ht(\Pi)$. The proof is by induction on $n$. Most of structural rules only has one principal relational atom, so it is easy to show that contraction can permute through them.





The case for $A$ needs more care, as it involves two principal
relations. If the two principal relations are different, then the
admissibility of contraction follows similarly as above. But if the
principal relations are a pair of identical relations, the situation is a bit tricky. The original derivation runs as follows.
\begin{center}
\alwaysNoLine
\AxiomC{$\Pi$}
\UnaryInfC{$(x,w\simp x);(y,y\simp w);(x,y\simp x);(x,y\simp x);\Gamma \vdash \Delta$}
\alwaysSingleLine
\RightLabel{$A$}
\UnaryInfC{$(x,y\simp x);(x,y\simp x);\Gamma \vdash \Delta$}
\DisplayProof
\end{center}

There is no obvious way to make this case admissible, and this is the
reason we have a special case of the rule $A$, namely $A_C$. In the
rule $A_C$, contraction is absorbed so that there is only one
principal relation. The new derivation is as follows.
\begin{center}
\alwaysNoLine
\AxiomC{$\Pi'$}
\UnaryInfC{$(x,w\simp x);(y,y\simp w);(x,y\simp x);\Gamma \vdash \Delta$}
\alwaysSingleLine
\RightLabel{$A_C$}
\UnaryInfC{$(x,y\simp x);\Gamma \vdash \Delta$}
\DisplayProof
\end{center}

For $Eq_1$ and $Eq_2$, as the principal relation is carried to the
premise (although some labels may be changed), so admissibility of
contraction on those relations is obvious.
\qed
\end{proof}

The admissibility of contraction on formulae are straightforward, the most of cases are analogous to the ones in Negri's labelled calculus for modal logic~\cite{Negri2005}. For details please see Appendix~\ref{app:ctr_f}.

\begin{lemma}[Contraction admissibility]
\label{lm:ctr}
If $\Gamma;\Gamma \vdash \Delta;\Delta$ is derivable in $\lsbbi$, 
then $\Gamma \vdash \Delta$ is derivable with the same height in $\lsbbi$.
\end{lemma}

\subsection*{\bf Cut Elimination Theorem}



We define the complexity of an application of the $cut$ rule as
$(|f|,ht(\Pi_1)+ht(\Pi_2))$, where $|f|$ denotes the size of the cut
formula (i.e., the number of connectives in the formula), and
$ht(\Pi_1)$, $ht(\Pi_2)$ are the heights of the derivations above the
$cut$ rule, the sum of them is call the \textit{cut height}. If there
are multiple branches in $\Pi_1$, then $ht(\Pi_1)$ shall be the height
of the longest branch, similarly for $ht(\Pi_2)$. The strict ordering
for both parts of the pair is $>$ on natural numbers.

\begin{theorem}[Cut-elimination]
\label{cut}
If $\Gamma \vdash \Delta$ is derivable in $\lsbbi$, then it is also derivable without using the $cut$ rule.
\end{theorem}

\begin{proof}
By induction on the complexity of the proof in $\lsbbi$. We show that
each application of $cut$ can either be eliminated, or be replaced by
one or more $cut$ rules of less complexity. The argument for
termination is similar to the cut-elimination proof for
$G3ip$~\cite{Negri2001}. We start to eliminate the topmost $cut$
first, and repeat this procedure until there is no $cut$ in the
derivation. We first show that $cut$ can be eliminated when the
\textit{cut height} is the lowest, i.e., at least one premise is of
height 1. Then we show that the \textit{cut height} is reduced in all
cases in which the cut formula is not principal in both premises of
cut. If the cut formula is principal in both premises, then the $cut$
is reduced to one or more $cut$s on smaller formulae or shorter
derivations. Since atoms cannot be principal in logical rules, finally
we can either reduce all $cut$s to the case where the cut formula is
not principal in both premises, or reduce those $cut$s on compound
formulae until their \textit{cut height}s are minimal and then
eliminate those $cut$s. The case analyses are shown in
Appendix~\ref{app:cut_elim}. Here we only present one interesting case where the
cut formula is principal in both premises, and the rules applied on each
premise are $\mand R$ and $\mand L$ respectively.

\begin{figure}[htp!]
\small
\alwaysNoLine
\AxiomC{$\Pi_1$}
\UnaryInfC{$(x,y\simp z);\Gamma \vdash x:A;z:A\mand B;\Delta$}
\AxiomC{$\Pi_2$}
\UnaryInfC{$(x,y\simp z);\Gamma \vdash y:B;z:A\mand B;\Delta$}
\alwaysSingleLine
\RightLabel{\tiny $\mand R$}
\BinaryInfC{$(x,y\simp z);\Gamma \vdash z:A\mand B;\Delta$}
\alwaysNoLine
\AxiomC{$\Pi_3$}
\UnaryInfC{$(x',y'\simp z);\Gamma';x':A;y':B \vdash \Delta'$}
\alwaysSingleLine
\RightLabel{\tiny $\mand L$}
\UnaryInfC{$\Gamma';z:A\mand B \vdash \Delta'$}
\RightLabel{$cut$}
\BinaryInfC{$(x,y\simp z);\Gamma;\Gamma' \vdash \Delta;\Delta'$}
\DisplayProof
$\leadsto$\\[10px]
\alwaysNoLine
\AxiomC{$\Pi_1'$}
\UnaryInfC{$(x,y\simp z);\Gamma;\Gamma' \vdash y:B;\Delta;\Delta'$}
\AxiomC{$\Pi_2'$}
\UnaryInfC{$(x,y\simp z);\Gamma;\Gamma' \vdash x:A;\Delta;\Delta'$}
\alwaysNoLine
\AxiomC{$\Pi_3'$}
\UnaryInfC{$(x,y\simp z);\Gamma';x:A;y:B \vdash \Delta'$}
\alwaysSingleLine
\RightLabel{$cut$}
\BinaryInfC{$(x,y\simp z);(x,y\simp z);\Gamma;\Gamma';\Gamma';y:B \vdash \Delta;\Delta';\Delta'$}
\RightLabel{$cut$}
\BinaryInfC{$(x,y\simp z);(x,y\simp z);(x,y\simp z);\Gamma;\Gamma;\Gamma';\Gamma';\Gamma' \vdash \Delta;\Delta;\Delta';\Delta';\Delta'$}
\dashedLine
\RightLabel{Lemma~\ref{lm:ctr}}
\UnaryInfC{$(x,y\simp z);\Gamma;\Gamma' \vdash \Delta;\Delta'$}
\DisplayProof
\caption{The cut reduction for $\mand$ where the cut formula is
  principal in both premises.}
\label{fig:cut-red}
\end{figure}

The cut transformation in this case is given in Figure~\ref{fig:cut-red}. There, the derivation
$\Pi_1'$ (likewise, $\Pi_2'$) is obtained by applying a cut to $\Pi_1$ (resp., $\Pi_2$) and the right 
premise of the original cut. We must also apply Lemma~\ref{lm:ctr} to remove
excess contexts. 
\qed
\end{proof}


\section{Localising structural rules}
\label{sec:local-struct-rules}

As a first step towards designing an effective proof search procedure for $\lsbbi$,
we need to restrict the use of structural rules. 

We remark the fact that the structural rules in
$\lsbbi$ can permute through all other rules except for $id$, $\top^*
R$, $\mand R$, and $\mimp L$. We refer to these four rules as
\textit{positive rules}, and the rest logical rules in $\lsbbi$ as
\textit{negative rules}. The main reason is,
all negative rules do not rely on the relational
atoms. This is formalised in the following lemma, and proved in
Appendix~\ref{app:str_permute}.

\begin{lemma}
\label{lem:str_permute}
The structural rules in $\lsbbi$ can permute upwards through any
negative rules in $\lsbbi$.
\end{lemma}

Then we design a more compact proof system  where applications of structural rules are
separated into a special entailment relation for relational atoms. We shall see in the next section that
proof search in this proof system can be separated into two phases: 
guessing the shape of the proof tree, and deriving
the relational atoms needed. The latter will be phrased in terms of a constraint system. 

In this section we localise the structural rules in two steps: we first 
deal with $Eq_1$ and $Eq_2$, and then the other structural rules. 


\subsection{Localising $Eq_1$ and $Eq_2$}

Allowing substitutions in a proof rule simplifies the cut-elimination proof for $\lsbbi$. 
However, for proof search, this creates a problem as $Eq_1$ and $Eq_2$ do not permute
over certain rules
that require matching of two labels (e.g., $\mand R$ or $\mimp L$).
Our first intermediate proof system $\ilsbbi$ aims to remove substitutions from $\lsbbi$. Instead, the equality
between labels is captured via a special entailment relation. 
To define its inference rules, we first need a few preliminary definitions. 

Let $r$ be an instance of a structural rule. We can view $r$ as a function that takes
a set of relational atoms (in the conclusion of the rule) 
and outputs another set (in the premise). We shall write $r(\Gcal,\theta)$, where $\Gcal$
is the set of principal relational atoms and $\theta$ is a substitution, 
to denote the set of relational atoms introduced in the premise 
of an instance of $r$ with conclusion containing $\Gcal$,
and where the substitution used in the rule is $\theta$, which is 
the identity substitution in all cases except when $r$ is $Eq_1$
or $Eq_2$. 
Let $\sigma$ be a sequence of instances of structural rules $[r_1(\Gcal_1,\theta_1);\cdots;r_n(\Gcal_n,\theta_n)]$.
Given a set of relation atoms $\Gcal$, the result of the application of $\sigma$ to $\Gcal$,
denoted by $\Scal(\Gcal, \sigma)$, 
is defined inductively as follows:
\begin{eqnarray*}
\Scal(\Gcal, \sigma)
 & = &
 \left\{
     \begin{tabular}{l@{\extracolsep{0.5cm}}cl}
        $\Gcal$ & if & $\sigma = [~]$
     \\
        $\Scal(\Gcal \theta \cup r(\Gcal',\theta), \sigma')$
            & if 
            & $\Gcal' \subseteq \Gcal \mathrm{~and~} 
               \sigma = [r(\Gcal',\theta);\sigma']$
     \\
        undefined & & otherwise 
     \end{tabular}
 \right.
\end{eqnarray*}
Given a $\sigma = [r_1(\Gcal_1,\theta_1);\cdots;r_n(\Gcal_n,\theta_n)]$, we denote
with $\subst(\sigma)$ the composite substitution $\theta_1 \circ
\cdots \circ \theta_n$, where $t(\theta_1\circ \theta_2)$ means $(t\theta_1)\theta_2$.

\begin{definition}
Let $\Gcal$ be a set of relational atoms. The entailment relation $\Gcal \vdash_E u = v$
holds iff there exists a sequence $\sigma$ of 
$Eq_1$ or $Eq_2$ structural rules
such that $\Scal(\Gcal, \sigma)$ is defined, and 
$u \theta = v\theta$, where $\theta = \subst(\sigma).$
\end{definition}

We now define the proof system $\ilsbbi$ as $\lsbbi\setminus\{Eq_1,Eq_2\}$ (i.e., $\lsbbi$ without rules $Eq_1,Eq_2$) where certain rules
modified according to Figure~\ref{fig:ILS_BBI}.


\begin{figure*}[ht!]
\small
\begin{center}
\begin{tabular}{cc}
\multicolumn{2}{c}{
\AxiomC{$\mathcal{G}\vdash_E (w_1 = w_2)$}
\RightLabel{\tiny $id$}
\UnaryInfC{$\Gamma;w_1:P \vdash w_2:P;\Delta$}
\DisplayProof
$\qquad$
\AxiomC{$(\epsilon,w\simp \epsilon);\Gamma \vdash \Delta$}
\RightLabel{\tiny $\top^* L$}
\UnaryInfC{$\Gamma;w:\top^* \vdash \Delta$}
\DisplayProof
$\qquad$
\AxiomC{$\mathcal{G}\vdash_E (w = \epsilon)$}
\RightLabel{\tiny $\top^* R$}
\UnaryInfC{$\Gamma \vdash w:\top^*;\Delta$}
\DisplayProof
}\\[15px]
\multicolumn{2}{c}{
\AxiomC{$(x,y \simp z');\Gamma \vdash x:A;z:A\mand B;\Delta$}
\AxiomC{$(x,y \simp z');\Gamma \vdash y:B;z:A\mand B;\Delta$}
\AxiomC{$\mathcal{G} \vdash_E (z = z')$}
\RightLabel{\tiny $\mand R$}
\TrinaryInfC{$(x,y \simp z');\Gamma \vdash z:A\mand B;\Delta$}
\DisplayProof
}\\[15px]
\multicolumn{2}{c}{
\AxiomC{$(x,y' \simp z);\Gamma;y:A\mimp B \vdash x:A;\Delta$}
\AxiomC{$(x,y' \simp z);\Gamma;y:A\mimp B; z:B \vdash \Delta$}
\AxiomC{$\mathcal{G} \vdash_E (y = y')$}
\RightLabel{\tiny $\mimp L$}
\TrinaryInfC{$(x,y' \simp z);\Gamma;y:A\mimp B \vdash \Delta$}
\DisplayProof
}\\[15px]
\multicolumn{2}{c}{
\AxiomC{$(u,w \simp z);(y,v \simp w);(x,y \simp z);(u,v \simp x');\Gamma \vdash \Delta$}
\AxiomC{$\mathcal{G}\vdash_E (x = x')$}
\RightLabel{\tiny $A$}
\BinaryInfC{$(x,y \simp z);(u,v \simp x');\Gamma \vdash \Delta$}
\DisplayProof
}\\[15px]
\multicolumn{2}{c}{
\AxiomC{$(x,w\simp x');(y,y\simp w);(x,y\simp x');\Gamma \vdash \Delta$}
\AxiomC{$\mathcal{G}\vdash_E (x = x')$}
\RightLabel{\tiny $A_C$}
\BinaryInfC{$(x,y\simp x');\Gamma \vdash \Delta$}
\DisplayProof
}
\end{tabular}\\[15px]
$\mathcal{G}$ is the set of relational atoms on the left hand side of the conclusion sequent.
\end{center}
\caption{The changed rules in $\ilsbbi$.}
\label{fig:ILS_BBI}
\end{figure*} 

Note that the new $\top^* L$ rule does not modify any labels, instead,
the relational atom $(\epsilon,w\simp \epsilon)$ in the premise
ensures that the derivability of $(w = \epsilon)$ is
preserved. 
The point of this intermediate step is to
avoid label substitutions in the proof system.

\begin{theorem}
A sequent $\Gamma \vdash \Delta$ is derivable in $\lsbbi$ if and only
if it is derivable in $\ilsbbi.$
\end{theorem}
\begin{proof}
{\em (Outline) \ }
One direction, from $\ilsbbi$ to $\lsbbi$ is straightforward, as $\vdash_E$
is essentially just a sequence of applications of $Eq_1$ and $Eq_2$. The other direction
can be proved by showing that $Eq_1$ and $Eq_2$ are admissible in $\ilsbbi.$ A more 
detailed proof is given in
Appendix~\ref{app:sound_ilsbbi},~\ref{app:comp_ilsbbi}, for soundness
and completeness respectively.\qed
\end{proof}

\subsection{Localising the rest of the structural rules}

As a second step, we isolate the rest structural rules into a separate entailment relation, 
as we did with $Eq_1$ and $Eq_2.$


\begin{definition}[Relation Entailment $\vdash_R$]
The entailment relation $\vdash_R$ has the following two forms:
\begin{enumerate}
\item $\Gcal \vdash_R (w_1 = w_2)$ is true 
iff there is a sequence $\sigma$ of $E$, $U$, $A$, $A_C$ applications so that 
$\Scal(\mathcal{G},\sigma) \vdash_E (w_1 = w_2)$. 
\item $\mathcal{G} \vdash_R (w_1,w_2\simp w_3)$ is true iff there is a
  sequence $\sigma$ of $E$, $U$, $A$, $A_C$ applications so that $(w_1',w_2'\simp w_3')\in \mathcal{S}(\mathcal{G},\sigma)$ 
and the following hold:
$\mathcal{S}(\mathcal{G},\sigma) \vdash_E (w_1 = w_1')$, 
$\mathcal{S}(\mathcal{G},\sigma) \vdash_E (w_2 = w_2')$, and 
$\mathcal{S}(\mathcal{G},\sigma) \vdash_E (w_3 = w_3').$
\end{enumerate}
\end{definition}

The entailment $\vdash_R$ is stronger than $\vdash_E$. For example, if $\Gcal$ only contains $(x,\epsilon\simp y)$, then $\Gcal\not\vdash_E (x = y)$; but $\Gcal\vdash_R (x = y)$ by applying $E$ to obtain $(\epsilon,x\simp y)$, then apply $Eq_1$ or $Eq_2$ on the new relational atom.

\begin{figure}[ht!]
\small
\begin{center}
\begin{tabular}{cc}
\AxiomC{$\mathcal{G} \vdash_R (w_1 = w_2)$}
\RightLabel{\tiny $id$}
\UnaryInfC{$\mathcal{G}||\Gamma;w_1:P \vdash w_2:P;\Delta$}
\DisplayProof
&
\AxiomC{$\mathcal{G} \vdash_R (w = \epsilon)$}
\RightLabel{\tiny $\top^* R$}
\UnaryInfC{$\mathcal{G}||\Gamma \vdash w:\top^*;\Delta$}
\DisplayProof
\\[15px]
\multicolumn{2}{c}{
\AxiomC{$\mathcal{S}(\mathcal{G},\sigma)||\Gamma \vdash x:A;w:A\mand B;\Delta$}
\AxiomC{$\mathcal{S}(\mathcal{G},\sigma)||\Gamma \vdash y:B;w:A\mand B;\Delta$}
\AxiomC{$\mathcal{G} \vdash_R (x,y\simp w)$}
\RightLabel{\tiny $\mand R^\sharp$}
\TrinaryInfC{$\mathcal{G}||\Gamma \vdash w:A\mand B;\Delta$}
\DisplayProof
}\\[15px]
\multicolumn{2}{c}{
\AxiomC{$\mathcal{S}(\mathcal{G},\sigma)||\Gamma;w:A\mimp B \vdash x:A;\Delta$}
\AxiomC{$\mathcal{S}(\mathcal{G},\sigma)||\Gamma;w:A\mimp B; z:B \vdash \Delta$}
\AxiomC{$\mathcal{G} \vdash_R (x,w\simp z)$}
\RightLabel{\tiny $\mimp L^\natural$}
\TrinaryInfC{$\mathcal{G}||\Gamma;w:A\mimp B \vdash \Delta$}
\DisplayProof
}\\[15px]
$\sharp$: $\sigma$ is the derivation of $\mathcal{G}
\vdash_R (x,y\simp w)$
&
$\natural$: $\sigma$ is the derivation of $\mathcal{G} \vdash_R (x,w\simp z)$
\end{tabular}
\end{center}
\caption{Changed rules in $\iilsbbi$.}
\label{fig:IILS_BBI}
\end{figure}

The changed rules in the second intermediate system $\iilsbbi$ is given in
Figure~\ref{fig:IILS_BBI}
where 
we use a slightly different notation for sequents. We write 
$
\Gcal || \Gamma \vdash \Delta
$
to emphasize 
that the left hand side of a sequent
is partitioned into two parts: $\Gcal$, which contains only relational atoms, and $\Gamma$, which 
contains only labelled formulae.


The following is an immediate result, the proof is divided in two
parts for soundness and completeness, detailed in
Appendix~\ref{app:sound_iilsbbi} and~\ref{app:comp_iilsbbi} respectively.

\begin{theorem}
A sequent $\Gamma \vdash \Delta$ is derivable in $\ilsbbi$ if and only
if it is derivable in $\iilsbbi.$
\end{theorem}

\section{Mapping proof search to constraint solving}
\label{sec:mapping-proof-search}


We now consider a proof search strategy for $\iilsbbi$. 
As we have isolated all the structural rules into the entailment relation $\vdash_R$, proof search in
$\iilsbbi$ consists of guessing the shape of the derivation tree, and then checking that each
entailment $\vdash_R$ can be proved. The latter involves guessing a splitting of labels
in the $\mand R$ and $\mimp L$ rules which also satisfies the equality
constraints in the $id$ and
$\top^* R$ rules. We formalise this via a symbolic proof system, where splitting
and equality are handled lazily, via the introduction of {\em free variables} which are essentially existential variables (or logic
variables) that must be instantiated to concrete labels satisfying all the entailment constraints in
the proof tree, for a symbolic derivation to be sound. 

Free variables are denoted by $\fvx$, $\fvy$ and $\fvz$.
We use $\fvu, \fvv, \fvw$ to denote either labels or free variables, and $a,b,c$ are ordinary labels. 
A {\em symbolic sequent} is just a sequent but possibly with occurrences of free variables in place of
labels. We shall sometimes refer to the normal (non-symbolic) sequent as a {\em ground sequent}
to emphasize the fact that it contains no free variables. 
The symbolic proof system $\fvlsbbi$ is given in Figure~\ref{fig:fvlsbbi}. 
The rules are mostly similar to $\iilsbbi$, but lacking the entailment relations $\vdash_R.$
Instead, new free variables are introduced when applying $\mand R$ and $\mimp L$ backward. 
Notice also that in $\fvlsbbi$, the $\mand R$ and $\mimp L$ rules do not compute the set
$\Scal(\Gcal,\sigma)$. So the relational atoms in $\fvlsbbi$ are those that are created by $\mand L, \mimp R, \top^* L$.
In the following, given a derivation in $\fvlsbbi$, we shall assume that the free variables
that are created in different branches of the derivation are pairwise distinct. 
We shall sometimes refer to a derivation in $\fvlsbbi$ simply as a {\em symbolic derivation}.

\begin{figure*}[ht!]
\small
\begin{center}
\begin{tabular}{cc}
\multicolumn{2}{l}{
\textbf{Initial Sequent:}
\qquad\qquad\qquad
\AxiomC{$$}
\RightLabel{\tiny $id$}
\UnaryInfC{$\mathcal{G}||\Gamma;\fvw_1:P \vdash \fvw_2:P;\Delta$}
\DisplayProof
}
\\[15px]
\multicolumn{2}{l}{\textbf{Logical Rules:}}\\
\AxiomC{$$}
\RightLabel{\tiny $\bot L$}
\UnaryInfC{$\mathcal{G}||\Gamma; \fvw:\bot \vdash \Delta$}
\DisplayProof
&
\AxiomC{$$}
\RightLabel{\tiny $\top R$}
\UnaryInfC{$\mathcal{G}||\Gamma \vdash \fvw:\top;\Delta$}
\DisplayProof
\\[15px]
\AxiomC{$\mathcal{G};(\epsilon,\fvw\simp\epsilon)||\Gamma \vdash \Delta$}
\RightLabel{\tiny $\top^* L$}
\UnaryInfC{$\mathcal{G}||\Gamma; \fvw:\top^* \vdash \Delta$}
\DisplayProof
&
\AxiomC{$$}
\RightLabel{\tiny $\top^* R$}
\UnaryInfC{$\mathcal{G}||\Gamma \vdash \fvw:\top^*;\Delta$}
\DisplayProof\\[15px]
\AxiomC{$\mathcal{G}||\Gamma; \fvw:A; \fvw:B \vdash \Delta$}
\RightLabel{\tiny $\land L$}
\UnaryInfC{$\mathcal{G}||\Gamma; \fvw:A\land B \vdash \Delta$}
\DisplayProof 
&
\AxiomC{$\mathcal{G}||\Gamma \vdash \fvw:A;\Delta$}
\AxiomC{$\mathcal{G}||\Gamma \vdash \fvw:B;\Delta$}
\RightLabel{\tiny $\land R$}
\BinaryInfC{$\mathcal{G}||\Gamma \vdash \fvw:A\land B;\Delta$}
\DisplayProof\\[15px]
\AxiomC{$\mathcal{G}||\Gamma \vdash \fvw:A;\Delta$}
\AxiomC{$\mathcal{G}||\Gamma; \fvw:B \vdash \Delta$}
\RightLabel{\tiny $\limp L$}
\BinaryInfC{$\mathcal{G}||\Gamma; \fvw:A\limp B \vdash \Delta$}
\DisplayProof
&
\AxiomC{$\mathcal{G}||\Gamma; \fvw:A \vdash \fvw:B;\Delta$}
\RightLabel{\tiny $\limp R$}
\UnaryInfC{$\mathcal{G}||\Gamma \vdash \fvw:A\limp B; \Delta$}
\DisplayProof\\[15px]
\AxiomC{$\mathcal{G};(a,b\simp \fvw)||\Gamma;a:A;b:B \vdash \Delta$}
\RightLabel{\tiny $\mand L^\dag$}
\UnaryInfC{$\mathcal{G}||\Gamma; \fvw:A\mand B \vdash \Delta$}
\DisplayProof
&
\AxiomC{$\mathcal{G};(a, \fvw \simp c)||\Gamma;a:A \vdash c:B;\Delta$}
\RightLabel{\tiny $\mimp R^\ddag$}
\UnaryInfC{$\mathcal{G}||\Gamma \vdash \fvw:A\mimp B;\Delta$}
\DisplayProof\\[15px]
\multicolumn{2}{c}{
\AxiomC{$\mathcal{G}||\Gamma \vdash \fvx:A; \fvw:A\mand B;\Delta$}
\AxiomC{$\mathcal{G}||\Gamma \vdash \fvy:B; \fvw:A\mand B;\Delta$}
\RightLabel{\tiny $\mand R^\sharp$}
\BinaryInfC{$\mathcal{G}||\Gamma \vdash \fvw:A\mand B;\Delta$}
\DisplayProof
}\\[15px]
\multicolumn{2}{c}{
\AxiomC{$\mathcal{G}||\Gamma;\fvw:A\mimp B \vdash \fvx:A;\Delta$}
\AxiomC{$\mathcal{G}||\Gamma;\fvw:A\mimp B; \fvz:B \vdash \Delta$}
\RightLabel{\tiny $\mimp L^\natural$}
\BinaryInfC{$\mathcal{G}||\Gamma;\fvw:A\mimp B \vdash \Delta$}
\DisplayProof
}\\[15px]
\noindent $\dag$: $a$ and $b$ must be fresh in $\mand L$ &
\noindent $\ddag$: $a$ and $c$ must be fresh in $\mimp R$\\
\noindent $\sharp$: $\fvx$ and $\fvy$ are new free variables in $\mand
R$ &
\noindent $\natural$: $\fvx$ and $\fvz$ are new free variables in
$\mimp L$
\end{tabular}
\end{center}
\caption{Labelled Sequent Calculus $\fvlsbbi$ for Boolean BI.}
\label{fig:fvlsbbi}
\end{figure*}

An {\em equality constraint} is an expression of the form $\Gcal \vdash_R^? (\fvu = \fvv)$, and
a {\em relational constraint} is an expression of the form $\Gcal \vdash_R^? (\fvu, \fvv \simp \fvw).$
In both cases, we refer to $\Gcal$ as the left hand side of the constraints, and $(\fvu = \fvv)$ and 
$(\fvu, \fvv \simp \fvw)$ as the right hand side. 
Constraints are ranged over by $\cfr, \cfr', \cfr_1,\cfr_2$, etc. Given a constraint $\cfr$, we write 
$\Gcal(\cfr)$ 
for the left hand side of $\cfr.$ A {\em constraint system} is just a set of constraints.
We write $\Gcal \vdash^?_R C$ for either an equality
or relational constraint. We write $fv(\cfr)$ to denote the set of free variables in $\cfr$, and $fv(\Ccal)$ to denote the set of free variables in a set of  constraints $\Ccal$.

\begin{definition}[Constraint systems]
A {\em constraint system} is a pair $(\Ccal, \preceq)$ of a set of constraints and
a well-founded partial order on elements of $\Ccal$ satisfying
\textbf{Monotonicity}: $\cfr_1 \preceq \cfr_2$ implies $\Gcal(\cfr_1) \subseteq \Gcal(\cfr_2).$
It is {\em well-formed} if it also satisfies 
\textbf{Unique variable origin}: $\forall \fvx$ in $\Ccal$, 
there exists a unique minimum (w.r.t. $\preceq$) constraint $\cfr(\fvx) = \Gcal_x \vdash_R^? (\fvu, \fvv \simp \fvw)$
s.t. $\fvx$ occurs in $(\fvu, \fvv \simp \fvw)$,
but not in $\Gcal_x$, and $\fvx$ does not occur in any $\cfr'$
where $\cfr'\preceq \cfr(\fvx)$. 
Such a $\cfr(\fvx)$ is the origin of $\fvx$.
\end{definition}

From now on, we shall denote with $\cfr(\fvx)$ the constraint where $\fvx$ originates from, as defined
in the above definition. We use the letter $\Cbb$ to range over constraint systems.


We write $\cfr_i \prec \cfr_j$ when $\cfr_i \preceq \cfr_j$ and $\cfr_i \not= \cfr_j$. Further,
we define a direct successor relation $\lessdot$ as follows: $\cfr_i \lessdot \cfr_j$
iff $\cfr_i \prec \cfr_j$ and there does not exist any $\cfr_k$ such that $\cfr_i \prec \cfr_k \prec \cfr_j.$

During proof search, associated constraints are generated as follows.



\begin{definition}
\label{def:intro_constr}
To a given symbolic derivation $\Pi$, we associate a set of
constraints $\Ccal(\Pi)$ as follows where the lowest rule
instance of $\Pi$ is:

\begin{tabular}[c]{l@{\extracolsep{2em}}p{0.8\textwidth}}
  $id$ & $\Ccal(\Pi) = \{\Gcal \vdash^?_R (\fvw_1 = \fvw_2)\}$ 
  \\
  $\top^* R$ &
  $\Ccal(\Pi) = \{\Gcal \vdash^?_R (\fvw = \epsilon)\}$
  \\
  $\mand R$ & $\Ccal(\Pi) = \Ccal(\Pi_1) \cup \Ccal(\Pi_2)
  \cup \{ \Gcal \vdash^?_R (\fvx, \fvy \simp \fvw) \}$ 
  where the left
  premise derivation is $\Pi_1$ and the right-premise derivation is
  $\Pi_2$
\\
$\mimp L$ & $\Ccal(\Pi) = \Ccal(\Pi_1) \cup \Ccal(\Pi_2)
  \cup \{ \Gcal \vdash^?_R (\fvx, \fvw \simp \fvy) \}$ where the left
  premise derivation is $\Pi_1$ and the right-premise derivation is
  $\Pi_2$
\\
 -- & If $\Pi$ ends with any other rule, with premise
 derivations
  $\{ \Pi_1,\dots,\Pi_n\}$, then $\Ccal(\Pi) = 
  \Ccal(\Pi_1) \cup \dots \cup \Ccal(\Pi_n).$
\end{tabular}
\end{definition}

Each constraint $\cfr \in \Ccal(\Pi)$ corresponds to a rule instance 
$r(\cfr)$ in $\Pi$ where $\cfr$ is generated. 
The ordering of the rules in the derivation tree of $\Pi$ then naturally
induces a partial order on $\Ccal(\Pi)$. That is, 
let $\preceq^\Pi$ be an ordering on $\Ccal(\Pi)$ defined as follows:
$\cfr_1 \preceq^\Pi \cfr_2$ iff the conclusion of $r(\cfr_1)$ 
appears in the path from the root sequent to the conclusion of $r(\cfr_2).$
Then obviously $\preceq^\Pi$ is a partial order.

The following property of $\Ccal(\Pi)$ is easy to verify.
\begin{lemma}
\label{lm:derivation-constraint}
Let $\Pi$ is a symbolic derivation.
Then $(\Ccal(\Pi), \preceq^\Pi)$ is a constraint system. 
Moreover, if the root sequent is ground, then $(\Ccal(\Pi), \preceq^\Pi)$
is well-formed.
\end{lemma}

Given a symbolic derivation $\Pi$, we define $\Cbb(\Pi)$ as the
constraint system $(\Ccal(\Pi), \preceq^\Pi)$ as defined above.


A consequence of Lemma~\ref{lm:derivation-constraint} is 
that if $\Cbb(\Pi) \not = \{~\}$, then there exists a 
minimum constraint $\cfr$, w.r.t. the partial order $\preceq^\Pi$, such that $\Gcal(\cfr)$ is ground. 

We now define what it means for a constraint system to be solvable. This is a bit complicated,
because we need to capture that (ternary) relational atoms created by the solution need to be
accumulated across different constraints, in order to guarantee soundness of $\fvlsbbi.$
A {\em free-variable substitution} $\theta$ is a mapping from free variables to 
free-variables or labels with finite domain. 
We denote with $dom(\theta)$ the domain of $\theta$. 
Given $\theta$ and a set $V$ of free variables, $\theta\uparrow V$ is
the substitution obtained from $\theta$ by restricting the domain to $V$, i.e.,
$$
\fvx(\theta\uparrow V) = 
\left\{
\begin{array}{ll}
\fvx\theta & \mbox{if $\fvx \in V$}\\
\fvx & \mbox{otherwise.}
\end{array}
\right.
$$
Given $\theta$ and $\theta'$ such that $dom(\theta') \subseteq dom(\theta)$,
we define $\theta\setminus \theta'$ as the substitution: 
$$
\fvx(\theta \setminus \theta') = 
\left\{
\begin{array}{ll}
\fvx\theta & \mbox{if $\fvx \not \in dom(\theta')$ } \\
\fvx & \mbox{otherwise.}
\end{array}
\right.
$$


\begin{definition}[Simple constraints and their solutions]
A constraint $\cfr$ is {\em simple}  if its left hand side $\Gcal(\cfr)$  contains no free variables. 
A {\em solution} $(\theta, \sigma)$ to a simple constraint $\cfr$ is a
substitution $\theta$ and 
a sequence $\sigma$ of structural rules such that:
\begin{itemize}
\item If $\cfr$ is $\Gcal \vdash_R^? (\fvu =\fvv)$ then $\sigma$ is a derivation of 
$\Gcal \vdash_R (\fvu\theta = \fvv\theta).$
\item If $\cfr$ is  $\Gcal \vdash_R^? (\fvu, \fvv \simp \fvw)$ then $\sigma$ is a derivation of 
$\Gcal \vdash_R (\fvu\theta, \fvv\theta \simp \fvw\theta).$
\end{itemize}
\end{definition}

The minimum constraints of a well-formed constraint system are simple. 

\begin{definition}[Restricting a constraint system]
\label{def:res-cs}
Let $\Cbb=(\Ccal, \preceq)$ be a well-formed constraint system, and $\cfr$ be a minimum (simple) constraint in $\Cbb$.
Let $(\theta, \sigma)$ be a solution to $\cfr$ and  
$\Gcal' = \Scal(\Gcal(\cfr), \sigma)$. 
Define a function $f$ on constraints: 
$$
f(\cfr') = 
\left\{
\begin{array}{ll}
(\Gcal' \cup \Gcal\theta \vdash^?_R C\theta) 
& \mbox{ if $\cfr' = (\Gcal \vdash^?_R C) \in \Ccal\setminus\{\cfr\} \ and \ \cfr \preceq \cfr'$,} \\
\cfr' & \mbox{ otherwise.}
\end{array}
\right.
$$
The {\em restriction of $\Cbb$}
by $(\cfr,\theta,\sigma)$, 
written $\Cbb \uparrow
(\cfr,\theta,\sigma)$, 
is 
the pair $(\Ccal', \preceq')$, where
(1) $\Ccal' = \{ f(\cfr') \mid \cfr' \in \Ccal\setminus \{\cfr\} \}$ and
(2) $f(\cfr_1) \preceq' f(\cfr_2)$ iff $\cfr_1 \preceq \cfr_2$.
\end{definition}


\begin{lemma}
The pair $\Cbb \uparrow (\cfr,\theta,\sigma)$ as defined
in Definition~\ref{def:res-cs} is a well-formed constraint system.
\end{lemma}



\begin{definition}[Solution to a well-formed constraint system]
\label{dfn:sol_wf_cons}
Let $\Cbb = (\{\cfr_1,\dots,\cfr_n\}, \preceq)$ be a well-formed constraint system. A
{\em solution} $(\theta, \{\sigma_1,\dots,\sigma_n\})$ to $\Cbb$ is
a substitution and a set of sequences of
structural rules, such that: 
\begin{description}
\item[\rm If $n = 0$] then $(\theta, \{\sigma_1,\dots,\sigma_n\})$ is trivially a solution. 
\item[\rm If $n \geq 1$] then there must exist some minimum (simple) constraint in $\Cbb$.
For any minimum constraint
  $\cfr_i$, let $\theta_i = \theta\uparrow fv(\cfr_i)$, 
then
$(\theta_i, \sigma_i)$ is a
solution to $\cfr_i$, 
and $(\theta\setminus\theta_i, \{\sigma_1,\dots, \sigma_n \}\setminus \sigma_i)$ is a solution
to $\Cbb \uparrow (\cfr_i, \theta_i, \sigma_i)$.
\end{description}
\end{definition}

In Definition~\ref{dfn:sol_wf_cons}, suppose a constraint system $\Cbb
= (\{\cfr_1,\cdots,\cfr_n\},\preceq)$ has a solution
$(\theta,\{\sigma_1,\cdots,\sigma_n\})$, then for each constraint
$\cfr_i$ in $\Cbb$, let $\cfr_i'$ be the simple constraint obtained
from $\cfr_i$ in the process of restricting $\Cbb$, there is a
solution $(\theta_i,\sigma_i)$ for $\cfr_i'$, where $\theta_i$ is a
restriction on $\theta$ that contains the free variables in $\cfr_i'$,
and $\sigma_i \in \{\sigma_1,\cdots,\sigma_n\}$. In this case, we will
simply write $\sigma_i = dev(\cfr_i)$ for the mapping between a
(possibly ungrounded) constraint in the system and the corresponding
derivation in the solution.


\begin{theorem}[Soundness]
\label{thm:sound_fvlsbbi}
Let $\Pi$ be a symbolic derivation of a ground sequent 
$\Gcal || \Gamma \vdash \Delta$.
If $\Ccal(\Pi)$ is solvable, then $\Gcal || \Gamma \vdash \Delta$ is derivable in $\iilsbbi.$
\end{theorem}
The proof is done by induction on the height of symbolic derivations.
The basic idea of the proof is that one progressively ``ground'' a symbolic
derivation, starting from the root of the derivation. At each inductive step
we show that grounding the premises corresponds to restricting the constraint
system induced by the symbolic derivation. 
The detailed proof can be found in Appendix~\ref{app:sound_fvlsbbi}.



To prove the completeness of $\fvlsbbi$, we show that for every cut-free derivation 
$\Pi$ of a (ground) sequent
in $\iilsbbi$, there is a symbolic derivation $\Pi'$ of the same sequent
such that $\Ccal(\Pi')$ is solvable. It is quite obvious that $\Pi'$ should
have exactly the same rule applications as $\Pi$; the only difference is that
some relational atoms are omitted in the derivation, but instead are accumulated in the constraint system. Additionally, some (new) labels are replaced with free variables.
This is formalised in the following definition. 

\begin{definition}
Given a sequent in a $\iilsbbi$ derivation, let $\mathcal{G}$ be the
set of its relational atoms, we define $\mathcal{G}_E$ as the subset of
$\mathcal{G}$ that contains those ternary relational atoms created by
$\mand L$, $\mimp R$, and $\top^* L$. We define $\mathcal{G}_S =
\mathcal{G}\setminus \mathcal{G}_E$. We refer to $\mathcal{G}_E$ as
the \textbf{essential} subset of $\mathcal{G}$, and $\mathcal{G}_S$ as
the \textbf{supplementary} subset of $\mathcal{G}$.
\end{definition}


For a list $L$, we denote by $head(L)$ the first element
in the list $L$ and $tail(L)$ the list of $L$ without the first
element, and $end(L)$ the last element in $L$. We denote by $L_1@L_2$
the concatenation of two lists $L_1$ and $L_2$, and $pre(x)$ the
predecessor of $x$ in a list $L$, and $suc(x)$ the successor of $x$ in
$L$.

Given a well-formed constraint system $(\Ccal, \preceq)$, 
we can define a partial order $\preceq^v$ on free variables of $\Ccal$ as follows:
$\fvx \preceq^v \fvy$ iff $\cfr(\fvx) \preceq \cfr(\fvy).$ That is, free variables are ordered
according to their originations.
The relations $\prec^v$ and $\lessdot^v$ are defined analogously to $\prec$ and $\lessdot,$
i.e., as the non-reflexive subset of $\preceq^v$ and the successor relation. 

\begin{definition}[A thread of variables]
\label{dfn:linear_fv}
Let $\Cbb = (\Ccal, \preceq)$ be a well-formed constraint system, and let $X$ be a list of free variables 
$\fvx_1,\ldots,\fvx_n,$ where $n \geq 0.$ Let $\preceq^v$ be the partial order on variables, 
derived from $\preceq.$ 
We say $X$ is a {\em thread of free variables of $\Cbb$} (or simply a {\em thread} of $\Cbb$) 
iff it satisfies the following conditions:
\begin{enumerate}
\item $\forall \fvx \in X$, $\fvx\in fv(\Ccal)$
\item For every $i \in \{1,\dots,n-1\}$, $\fvx_i \lessdot^v \fvx_{i+1}.$ 
\item If $n\geq 1$, then $\fvx_1$ is a minimum element and $\fvx_n$ is a maximum element of $\preceq^v.$
\item If $n\geq 1$, then $\cfr(\fvx_1)$ is a minimum constraint in $\Ccal$.
\end{enumerate}
\end{definition}

A thread is effectively those variables that are generated along a
certain branch in a $\fvlsbbi$ symbolic derivation. It is not hard to
verify that in a valid symbolic derivation in $\fvlsbbi$ of a ground sequent, the set of
free variables in any symbolic sequent in the derivation can be linearly ordered as a thread.

\begin{definition}
\label{dfn:constraints-ext}
Let $\Cbb = (\Ccal_1, \preceq_1)$ be a well-formed constraint system, let $X$ be a thread of $\Cbb_1$
and let $\Cbb_2 = (\Ccal_2, \preceq_2)$ be a constraint system (but not neccessarily well-formed) such that 
$X$ consists of free variables in $fv(\Ccal_1) \cap fv(\Ccal_2)$. 
Furthermore, assume that every variable $\fvx$ in $\Ccal'$, except for 
those in $X$, satisfies the variable origination property, i.e., $\fvx$ originates from
a constraint in $\Ccal'.$ The {\em composition} of $\Cbb_1$ and $\Cbb_2$
along the thread $X$, written $\Cbb_1 \circ^X \Cbb_2$,  
is the constraint system $(\Ccal, \preceq)$ such that:
\begin{itemize}
\item $\Ccal = \Ccal_1 \cup \Ccal_2$; and
\item For $\cfr_1,\cfr_2 \in \Ccal$, $\cfr_1 \preceq \cfr_2$ iff either one of the following holds:
\begin{itemize}
\item $\cfr_1 \preceq_1 \cfr_2$,
\item $\cfr_1 \preceq_2 \cfr_2$, or
\item $X$ is non-empty, $\fvy = end(X)$, $\cfr_1 = \cfr(\fvy)$ 
and $\cfr_2 \in \Ccal_2.$
\end{itemize}
\end{itemize}
\end{definition}
This definition basically says that the composition of $\Cbb_1$ and $\Cbb_2$ along $X$
is obtained by simply ordering the constraints so that all constraints $\Ccal_2$
are greater than $\cfr(\fvy)$, where $\fvy$ is the last variable in $X$. 
If $X$ is empty, then $\Ccal_1$ and $\Ccal_2$ are independent, and 
$\prec$ is simply the union of $\prec_1$ and $\prec_2.$

\begin{lemma}
Let $(\Ccal, \preceq)$ be as defined in Definition~\ref{dfn:constraints-ext}.
Then $(\Ccal, \preceq)$ is well-formed.
\end{lemma}
\begin{proof}
This follows straightforwardly from the definition.\qed
\end{proof}

\begin{lemma}
\label{lem:cons_merge}
Let $\Cbb = (\Ccal, \preceq)$ be a well-formed constraint system and let $X$ be a thread of $\Cbb.$
Let $\Pi$ be a symbolic derivation such that the free variables in its end sequent are exactly
those in $X$.
Then $\Cbb \circ^X \Cbb(\Pi)$ is well-formed.
\end{lemma}

\begin{definition}
Let $\Cbb = (\Ccal, \preceq)$ be 
a well-formed constraint system and let $S = (\theta, \{\vec{\sigma}\})$
be its solution. Let $X$ be a thread of $\Cbb$. 
Define a set of relational atoms $\Scal^*(\Cbb, S, X)$ 
inductively as follows:
\begin{itemize}
\item If $n=0$ then $\Scal^*(\Cbb,S,[]) = \emptyset$
\item Suppose $n > 0$. Let $head(X) = \fvx$. Then $\cfr(\fvx)\in\Ccal$ 
is a minimum constraint of $\Cbb$, 
and there exists $\sigma_{\fvx} \in \{\vec \sigma\}$ such that
$(\theta_{\fvx}, \sigma_{\fvx})$ is a solution to $\cfr(\fvx)$, 
where $\theta_{\fvx} = \theta\uparrow fv(\cfr(\fvx)).$ 
In this case,
$\Scal^*(\Cbb, S, X)$ is defined as follows. 
$$
\Scal^*(\Cbb, S, X) = 
\Scal(\Gcal(\cfr(\fvx)), \sigma_{\fvx}) \cup 
\Scal^*(\Cbb\uparrow (\cfr(\fvx),\theta_{\fvx},\sigma_{\fvx}), S', tail(X))
$$
where $S' = (\theta\setminus \theta_{\fvx}, \{\vec{\sigma}\}\setminus \{\sigma_{\fvx}\})$.

\end{itemize}
\end{definition}


Notice that by
the definition of restriction to a constraint system, every time a minimum
constraint $\cfr_\fvx$ is eliminated in the second clause in the above definition, 
$\Scal(\Gcal(\cfr_\fvx),\sigma_\fvx)$ is also
added to the left hand side of every successor constraints of $\cfr_\fvx$ in $\Ccal.$
Therefore it is straightforward that the following proposition holds.

\begin{proposition}
\label{prop:scalstar}
Let $\Cbb = (\Ccal, \preceq)$ be a well-formed constraint system. 
Let $\Gcal = \Scal^*(\Cbb,S, X)$, for some thread $X$ of $\Cbb$,
let $\fvx_e = end(X)$ and let $S = (\theta,\{\overset{\rightharpoonup}{\sigma}\})$ be a solution to 
$\Cbb$. Let $\cfr = \Gcal_\cfr \vdash^?_R C_\cfr$ be a constraint not in $\Ccal$, such that $\Gcal_\cfr$ only contains free variables that occur in $\Cbb$. Let $\fvx$ be a new variable
occurring only on the right hand side of $\cfr.$
Let $\Cbb' = (\Ccal', \preceq')$ be the following constraint system: 
\begin{itemize}
\item $\Ccal' = \Ccal \cup \{\cfr \}$;
\item $\preceq'$ is the smallest extension of $\preceq$ such that
$\cfr(\fvx_e) \lessdot \cfr.$
\end{itemize}
Let $(\theta_x,\sigma_x)$ be the solution to 
$\cfr' = \Gcal\cup\Gcal_\cfr\theta\vdash^?_R C_\cfr\theta$, $S' = (\theta\cup\theta_x,\{\overset{\rightharpoonup}{\sigma},\sigma_x\})$, and $X' = X@[\fvx]$. 
Then $\Scal^*(\Cbb',S', X') = \Scal(\Gcal\cup \Gcal_\cfr\theta,\sigma_x)$.
\end{proposition}

\begin{theorem}
\label{thm:comp_fvlsbbi}
Let $\Pi$ be a derivation of a sequent in $\iilsbbi$. Then there exists a symbolic
derivation $\Pi'$ of the same sequent such that $\Ccal(\Pi')$ is solvable. 
\end{theorem}

The heart of the proof for this theorem is that we can recover the
supplementary subset for each sequent from the constraint
system using $\Scal^*$. Since the constraint system accumulates the relational atoms
in the derivation, this is not a surprising result. The proof is given
in more details in Appendix~\ref{app:comp_fvlsbbi}.


\section{A heuristics for proof search}
\label{sec:heur-proof-search}

Suppose we want to prove $((a \mand b) \mand c) \limp (a \mand (b \mand c))$. Using
$\fvlsbbi$, we build a symbolic derivation as in Figure~\ref{fig:sym} 
(right associativity for connectives is assumed).
\begin{figure*}[t]
Let $\Gamma_1 := \{a2:c~;~a3:a\}$ and $\Gamma_2 := \{a3:a~;~a4:b\}$
in\\[5px]
\centering
\footnotesize
\AxiomC{}
\RightLabel{$id$}
\UnaryInfC{$a2:c;a3:a;a4:b \vdash \fvx5:a$}
\AxiomC{}
\RightLabel{$id$}
\UnaryInfC{$\Gamma_1; a4:b \vdash \fvx7:b$}
\AxiomC{}
\RightLabel{$id$}
\UnaryInfC{$a2:c; \Gamma_2 \vdash \fvx8:c$}
\RightLabel{$\mand R$}
\BinaryInfC{$a2:c;a3:a;a4:b \vdash \fvx6:b \mand   c$}
\RightLabel{$\mand R$}
\BinaryInfC{$a2:c;a3:a;a4:b \vdash a0:a \mand   b \mand   c$}
\RightLabel{$\mand L$}
\UnaryInfC{$a1:a \mand   b;a2:c \vdash a0:a \mand   b \mand   c$}
\RightLabel{$\mand L$}
\UnaryInfC{$a0:(a \mand   b) \mand   c \vdash a0:a \mand   b \mand   c$}
\RightLabel{$\limp R$}
\UnaryInfC{$ \vdash a0:(a \mand   b) \mand   c \limp   a \mand   b \mand   c$}
\DisplayProof
\caption{A symbolic derivation for $((a \mand b) \mand c) \limp (a
  \mand (b \mand c))$.}
\label{fig:sym}
\end{figure*}
The following constraints are generated from this derivation:

\begin{tabular}[c]{l@{\extracolsep{2em}}cl}
$(a1,a2\simp a0);(a3,a4\simp a1)$ & $\vdash_R$ & $(a3=\fvx5)$\\
$(a1,a2\simp a0);(a3,a4\simp a1)$ & $ \vdash_R$ & $(a4=\fvx7)$\\
$(a1,a2\simp a0);(a3,a4\simp a1)$ & $ \vdash_R$ & $(a2=\fvx8)$\\
$(a1,a2\simp a0);(a3,a4\simp a1)$ & $ \vdash_R$ & $(\fvx7,\fvx8\simp \fvx6)$\\
$(a1,a2\simp a0);(a3,a4\simp a1)$ & $ \vdash_R$ & $(\fvx5,\fvx6\simp a0)$\\
\end{tabular}

Since the first three constraints are required by the $id$ rule, we
must accept them. Thus we are only left with the last two constraints
with free variables $\fvx5, \fvx7, \fvx8$ assigned. In the following, we shall
write $(a1,a2\simp a0);(a3,a4\simp a1)$ as $\Gcal$, and $(a3,\fvx6\simp
a0);(a2,a4\simp \fvx6)$ as $C$. Now $\fvx6$ is the only remaining free
variable. We can apply the rule $A$ (upward) on the left hand side
$\Gcal$ to obtain $(a3,w\simp a0);(a2,a4\simp w)$, where $w$ is a new
label. Then apply the rule $E$ (upward) to obtain $(a4,a2\simp w)$. 
The two constraints can be solved by assigning $w$ to $\fvx6$.

But there is a simpler way to see that such a $w$ must exist: 
the two ternary relational atoms in $\Gcal$ 
manifest that
$a0$ can be split into $a2,a3,a4$. This is exactly what $C$ says. 
For any variant of $\Gcal$ that describes the same splitting of
$a0$ as $C$, the ``internal'' node $\fvx6$ can always be assigned to
either an existing label or a label generated by the associativity
rule. In the example, $\fvx6$ cannot be matched to any existing label,
so we can assign $\fvx6$ to be a fresh label globally, and add $C$ to
the l.h.s.\ 
of the successor constraints in the partial order
$\prec$. 
Similarly
for any variant of $C$ with the same
splitting of $a0$. 

We can extend this method to a chain of multiple relational
atoms which forms a labelled binary tree. 
We define a labelled binary tree as a binary tree where each node is
associated with a label. Each node in a labelled binary tree has a
left child and a right child. The minimum labelled binary tree has a root and two leaves, which corresponds to a single relational atom. We define the following function
inductively from a labelled binary tree to a set of relational atoms.
\begin{definition}
Let $tr$ be a labelled binary tree, the set of relational atoms
w.r.t. $tr$, written as $\Rel(tr)$, is defined as follows.
\begin{itemize}
\item (Base case): $tr$ only contains a root node labelled with $r$ and
  two leaves labelled with $a,b$ respectively. Then $\Rel(tr) =
  \{(a,b\simp r)\}$
\item (Inductive case): $tr$ contains a root node labelled with $r$ and
its left and right children labelled with $a$ and $b$ respectively. 
  Then $\Rel(tr) = \Rel(tr_a)\cup\Rel(tr_b)\cup\{(a,b\simp r)\}$,
where $tr_a$ and $tr_b$ are the subtrees rooted at,
respectively, the left child and the right child of the root node
of $tr.$
\end{itemize}
\end{definition}

The \textit{width} of a labelled binary tree is defined as the number
of leaves in the tree. 
A labelled binary tree is a {\em variant} of another labelled binary tree 
if either they are exactly the same, or they differ only in the labels of the internal nodes.

We say that a set $R$ of relational atoms {\em forms a labelled binary tree} $tr$
when $R = \Rel(tr).$
In this case, the leaves in $tr$ are actually a ``splitting'' of the root node. Commutativity and
associativity guarantee that we can split a node arbitrarily, as long as the
leaves in the tree are the same. Moreover, since all internal nodes are
free variables, we can assign them to either existing labels or fresh
labels (created by $A,A_C$) without having clash with existing relational atoms. This idea
is formalised in the following lemma, and is proved in Appendix~\ref{app:heuristics}.

\begin{lemma}
\label{lem:heuristics}
Given constraints $\cfr_1\lessdot \cdots\lessdot \cfr_n$
with $\Gcal = \Gcal(\cfr_1) = \cdots = \Gcal(\cfr_n)$ and the r.h.s. of these constraints gives the set $R$ of relational atoms. If the following hold:
\begin{enumerate}
\item $R = \Rel(tr)$, for some labelled binary tree $tr$ where
every internal node label is a free variable $\fvx$ which only occurs once in $tr$, and $\cfr_1\preceq\cfr(\fvx)$. 
\item The other node labels in $tr$ are non-$\epsilon$ labels.
\item There exist $\Gcal'\subseteq\Gcal$ and $tr'$ such that $\Gcal' = \Rel(tr')$ and $tr'$ has the same root and leaves as $tr$.
\end{enumerate}
Then $\cfr_1,\cdots,\cfr_n$ are solvable.
\end{lemma}

\section{Experiment}
\label{sec:experiment}

\begin{table*}[t!]
\centering
\footnotesize
\begin{tabular}{|l|l|l|l|l|}
\hline
Formula & BBeye  & Naive & $\fvlsbbi$ \\
        & (opt) & (Vamp)& Heuristic \\
\hline
$(a \mimp b) \land (\top \mand (\top^* \land a)) \limp b$ & d(2) 0 & 0.003 & 0.001\\
$(\top^* \mimp \lnot (\lnot a \mand \top^*)) \limp a$ & d(2) 0 & 0.003 & 0.000\\
$\lnot ((a \mimp \lnot (a \mand b)) \land ((\lnot a \mimp \lnot b) \land b))$ & d(2) 0 & 0.004 & 0.001\\
$\top^* \limp ((a \mimp (b \mimp c)) \mimp ((a \mand b) \mimp c))$ & d(2) 0.015 & 0.017 & 0.001\\
$\top^* \limp ((a \mand (b \mand c)) \mimp ((a \mand b) \mand c))$ & d(2) 0.036 & 0.006 & 0.000\\
$\top^* \limp ((a \mand ((b \mimp e) \mand c)) \mimp ((a \mand (b \mimp e)) \mand c))$ & d(2) 0.07 & 0.019 & 0.001\\
$\lnot ((a \mimp \lnot (\lnot (d \mimp \lnot (a \mand (c \mand b))) \mand a)) \land c \mand (d \land (a \mand b)))$ & d(2) 0.036 & 0.037 & 0.001\\
$\lnot ((c \mand (d \mand e)) \land B)$ where & d(2) 0.016 & 0.075 & 0.039\\
$ B := ((a \mimp \lnot (\lnot (b \mimp \lnot (d \mand (e \mand c)))
\mand a)) \mand (b \land (a \mand \top)))$ & & &\\
$\lnot ( C \mand (d \land (a \mand (b \mand e))))$ where & d(3) 96.639
& 0.089 & 0.038\\
C := $((a \mimp \lnot (\lnot (d \mimp \lnot ((c \mand e) \mand (b
\mand a))) \mand a)) \land c)$ & & & \\
$(a \mand (b \mand (c \mand d))) \limp (d \mand (c \mand (b \mand a)))$ & d(2) 0.009 & 0.048 & 0.001\\
$(a \mand (b \mand (c \mand d))) \limp (d \mand (b \mand (c \mand a)))$ & d(3) 0.03 & 0.07 & 0.001\\
$(a \mand (b \mand (c \mand (d \mand e)))) \limp (e \mand (d \mand (a \mand (b \mand c))))$ & d(3) 1.625 & 1.912 & 0.001\\
$(a \mand (b \mand (c \mand (d \mand e)))) \limp (e \mand (b \mand (a \mand (c \mand d))))$ & d(4) 20.829 & 0.333 & 0.001\\
$\top^* \limp (a \mand ((b \mimp e) \mand (c \mand d)) \mimp ((a \mand d) \mand (c \mand (b \mimp e))))$ & d(3) 6.258 & 0.152 & 0.007\\
\hline
\end{tabular}
\\
\caption{Initial experimental results.}
\label{tab:test_park}
\vspace{-1em}
\end{table*}

We used a Dell Optiplex 790 desktop with Intel CORE i7 2600 @ 3.4 GHz CPU
and 8GB memory as the platform, and tested the following provers on the formulae from~Park et
al.~\cite{Park2013}. 
(1) BBeye: the OCaml prover from Park et al. based upon nested sequents~\cite{Park2013}; (2) Naive (Vamp): translates a BBI formula into a first-order formula using the standard translation, then uses Vampire 2.6~\cite{Vampire2.6} to solve it; (3) $\fvlsbbi$ Heuristic: backward proof search in $\fvlsbbi$, using the heuristic-based method to solve the set of constraints.
%

The results are shown in
Table~\ref{tab:test_park}. The BBeye (opt) column shows the results from Park et al's prover where
the d() indicates the depth of proof search. The other two columns are
for the two methods stated above. We see that naive translation is
comparable with BBeye in most cases, but the latter is not stable. When
the tested formulae involves more interaction between structural
rules, BBeye runs significantly slower. The heuristic method outperforms
all other methods in the tested cases.

Nonetheless, 
our prover is slower
than BBeye for formulae which
contain many occurrences of the
same atomic formulae, giving (id) instances such as:
\begin{center}
$\Gamma;w_1:P;w_2:P;\cdots;w_n:P\vdash \fvx:P;\Delta$
\end{center}
We have to choose some $w_i$ to match with $\fvx$ without knowing
which choice satisfies other constraints. In the worst case, we have
to try each using backtracking.
Multiple branches of this form lead to a combinatorial explosion. 
Determinising the concrete labels (worlds) for formulae in proof search in $\lsbbi$ or BBeye~\cite{Park2013} avoids this
problem. 
Further work is needed to solve this in $\fvlsbbi$.

Even though we do not claim the completeness of our heuristics method,
it appears to be a fast way to solve certain problems. Completeness can be restored by fully implementing $\lsbbi$ or $\fvlsbbi$.
The derivations in $\lsbbi$ are
generally shorter than those in the Display Calculus or Nested Sequent
Calculus for BBI. The reader can verify that most of formulae in
Table~\ref{tab:test_park} can even be proved by hand in a reasonable
time using our labelled system. The optimisations of the
implementation, however, is out of the scope of this paper.

\section{Conclusion and Future Work}
\label{sec:conc}
Our main contribution is a labelled sequent calculus for
$BBI_{ND}$
that is sound, complete, and enjoys cut-elimination. 
There are no explicit contraction rules in $\lsbbi$ and all structural rules
can be restricted so that proof search is entirely driven by
logical rules. We further propose a free variable system to restrict the proof search space so that some applications of $\mand R,\mimp L$ rules can be guided by zero-premise rules. 
Although we can structure proof search to be more manageable compared to
the unrestricted (labelled or display) calculus,
the undecidability of BBI implies that there is no terminating
proof search strategy for a sound and complete system. The essence of
proof search now resides in 
guessing which relational atom
to use in the $\mand R$ and $\mimp L$ rules and whether 
they need to be applied more than once to a formula.
Nevertheless, our initial experimental results already raise the hope that a more efficient
proof search strategy can be developed based on our calculus. 

An immediate task is to find a complete and terminating (if possible) constraint solving strategy.
A counter-model construction for $BBI_{PD}$ has been studied by
Larchey-Wendling using labelled tableaux~\cite{Wendling2012}, the
possibility to adapt his method to $BBI_{ND}$ using our calculus is
also a future work.   
\begin{figure}[t]
\small
\centering
\begin{tabular}{c@{\hskip 15px}c}
\AxiomC{$(a,b\simp c);\Gamma[c/d]\vdash\Delta[c/d]$}
\RightLabel{\tiny $P$}
\UnaryInfC{$(a,b\simp c);(a,b\simp d);\Gamma\vdash\Delta$}
\DisplayProof
&
\AxiomC{$(a,b\simp c);\Gamma\vdash\Delta$}
\RightLabel{\tiny $T$}
\UnaryInfC{$\Gamma\vdash\Delta$}
\DisplayProof\\[15px]
\AxiomC{$(\epsilon,\epsilon\simp \epsilon);\Gamma[\epsilon/a][\epsilon/b]\vdash\Delta[\epsilon/a][\epsilon/b]$}
\RightLabel{\tiny $IU$}
\UnaryInfC{$(a,b\simp \epsilon);\Gamma\vdash\Delta$}
\DisplayProof
&
\AxiomC{$(a,b\simp c);\Gamma[b/d]\vdash\Delta[b/d]$}
\RightLabel{\tiny $C$}
\UnaryInfC{$(a,b\simp c);(a,d\simp c);\Gamma\vdash\Delta$}
\DisplayProof
\end{tabular}\\[5px]
\begin{tabular}{l}
In $T$, $a$, $b$ do occur in the conclusion but $c$ does not\\
In all substitutions $[y/x]$, $x\not= \epsilon$
\end{tabular}
\caption{Some auxiliary structural rules.}
\label{fig:aux_str}
\end{figure}

Another interesting topic is to extend our calculus to handle some semantics other
than the non-deterministic monoidal ones. Our design of the structural rules in $\lsbbi$ can be generalised as follows. If there is a semantic condition of the form $(w_{11},w_{12}\simp w_{13})\land\cdots\land(w_{i1},w_{i2}\simp w_{i3}) \Rightarrow (w'_{11},w'_{12}\simp w'_{13})\land\cdots\land(w'_{j1},w'_{j2}\simp w'_{j3})\land (x_{11} = x_{12})\land \cdots\land$ $(x_{k1} = x_{k2})$, we create a rule:
\begin{center}
\small
\AxiomC{$(w'_{11},w'_{12}\simp w'_{13});\cdots;(w'_{j1},w'_{j2}\simp w'_{j3});(w_{11},w_{12}\simp w_{13});\cdots;(w_{i1},w_{i2}\simp w_{i3});\Gamma\vdash\Delta$}
\RightLabel{\tiny $r$}
\UnaryInfC{$(w_{11},w_{12}\simp w_{13});\cdots;(w_{i1},w_{i2}\simp w_{i3});\Gamma\vdash\Delta$}
\DisplayProof
\end{center}   
And apply substitutions $[x_{12}/x_{11}]\cdots [x_{k2}/x_{k1}]$ globally on the premise, where $\epsilon$ is not substituted. Many additional features can be added in this way. 
We summarise the following desirable ones: (1) PD-semantics: the
composition of two elements is either the empty set or a singleton, i.e., $(a,b\simp c)\land (a,b\simp d) \Rightarrow (c = d)$; (2) TD-semantics: the composition of any two elements is always defined as a singleton, i.e., $\forall a,b, \exists c$ s.t. $(a,b\simp c)$; (3) indivisible unit: (cf. Section~\ref{sec:intro}) $(a,b\simp \epsilon) \Rightarrow (a = \epsilon) \land (b = \epsilon)$; and (4) cancellative: if $w\circ w'$ is defined and $w\circ w' = w\circ w''$, then $w' = w''$, i.e., $(a,b\simp c)\land (a,d\simp c)\Rightarrow (b = d)$. Note that (2) and (4) are in addition to (1). The above are formalised in rules $P$, $T$, $IU$, $C$ respectively in Figure~\ref{fig:aux_str}.

The formula 
\begin{equation}
\label{fm:partial}
(F * F) \limp F \text{, where } F = \lnot(\top \mimp \lnot \top^*), 
\end{equation}
differentiates $BBI_{ND}$ and
$BBI_{PD}$~\cite{Wendling2010} and is provable using $\lsbbi + P$.
Using $\lsbbi+T$, we can prove 
$(\lnot \top^* \mimp \bot) \limp \top^*$, which is valid in $BBI_{TD}$ but not
in $BBI_{PD}$~\cite{Wendling2010}, and
also  
$(\top^* \land ((p \mand q)
\mimp \bot)) \limp ((p \mimp \bot) \lor (q \mimp \bot))$, which is valid in 
separation models iff the composition is total~\cite{Brotherston13}. 
These additional rules do not break cut-elimination.

Oddly,
the formula $\lnot (\top^* \land A \land (B \mand \lnot (C \mimp
(\top^* \limp A))))$, which is valid in $BBI_{ND}$, is 
very hard to prove in the display calculus and Park et al.'s
method. We ran this formula using Park et al.'s prover for a week on a
CORE i7 2600 processor, without success. Very short proofs of this
formula exist in $\lsbbi$ or Larchey-Wendling and Galmiche's labelled
tableaux (this formula must also be valid in $BBI_{PD}$). We are
currently investigating this phenomenon. 

Furthermore, we point out that Park et al.'s following observations on $BBI_{ND}$~\cite[page 2]{Park2013} are incorrect:
\begin{quote}
\begin{enumerate}
\item A node can have multiple parent nodes, but each parent node determines a unique sibling node. Hence no node can have two parent nodes with the same sibling node.
\item A node can have multiple child nodes, but each child node determines another unique child node. Hence we can divide all child nodes into groups of two sibling nodes.
\end{enumerate}
\end{quote}
Observation 1 is the claim that partial-determinism holds and
Observation 2 is the claim that cancellativity holds, but neither of
these properties hold for $BBI_{ND}$, which is what their nested
sequent calculus is designed for. An earlier version of Park et al.'s
technical report said that the two observations also imply each other,
which, by Larchey-Wendling and Galmiche's
work~\cite{larcheywendling2014}, is only correct about
partial-determinism (Observation 1) implying cancellativity
(Observation 2). The other direction, however, is so far not known,
and this claim does not appear in the latest version of Park et
al.'s paper~\cite{Park2013}. Interestingly, 
Park et al.'s prover BBeye could not prove
Formula~\ref{fm:partial}
which is valid if partial-determinism (Observation 1)
holds.
We believe that their prover for $BBI_{ND}$ is actually behaving
correctly since Observation 1 is false in $BBI_{ND}$, so their own
prover does not support their Observation 1 above.

The proofs for the formulae
in this section can be found in Appendix~\ref{app:formulae_conc}.


\bibliographystyle{plain}
\bibliography{main}

\begin{thebibliography}{10}

\bibitem{Brotherston2010}
James Brotherston.
\newblock A unified display proof theory for bunched logic.
\newblock {\em ENTCS}, 265:197--211, September 2010.

\bibitem{BrostherstonCal10}
James Brotherston and Cristiano Calcagno.
\newblock Classical {BI}: Its semantics and proof theory.
\newblock {\em LMCS}, 6(3), 2010.

\bibitem{BrotherstonK10}
James Brotherston and Max Kanovich.
\newblock Undecidability of propositional separation logic and its neighbours.
\newblock In {\em LICS}, pages 130--139, 2010.

\bibitem{Brotherston13}
James Brotherston and Max Kanovich.
\newblock Undecidability of propositional separation logic and its neighbours.
\newblock {\em submitted to the Journal of ACM}, 2013.

\bibitem{Galmiche2006}
Didier Galmiche and Dominique Larchey-Wendling.
\newblock Expressivity properties of {B}oolean {BI} through relational models.
\newblock In {\em FSTTCS}, pages 358--369, 2006.

\bibitem{Vampire2.6}
Krystof Hoder and Andrei Voronkov.
\newblock Comparing unification algorithms in first-order theorem proving.
\newblock KI'09, pages 435--443. Springer-Verlag, 2009.

\bibitem{Wendling2012}
Dominique Larchey-Wendling.
\newblock The formal strong completeness of partial monoidal {B}oolean {BI}.
\newblock {\em Journal of Logic and Computation}, 2014.

\bibitem{Wendling2009}
Dominique Larchey-Wendling and Didier Galmiche.
\newblock Exploring the relation between intuitionistic {BI} and {B}oolean
  {BI}: An unexpected embedding.
\newblock {\em MSCS}, 19(3):435--500, 2009.

\bibitem{Wendling2010}
Dominique Larchey-Wendling and Didier Galmiche.
\newblock The undecidability of {B}oolean {BI} through phase semantics.
\newblock {\em LICS}, 0:140--149, 2010.

\bibitem{Wendling2013}
Dominique Larchey-Wendling and Didier Galmiche.
\newblock Non-deterministic phase semantics and the undecidability of {B}oolean
  {BI}.
\newblock {\em ACM TOCL}, 14(1), 2013.

\bibitem{larcheywendling2014}
Dominique Larchey-Wendling and Didier Galmiche.
\newblock Looking at separation algebras with {B}oolean {BI}-eyes.
\newblock {\em Theoretical Computer Science, TCS 2014}, 2014.

\bibitem{Negri2005}
Sara Negri.
\newblock Proof analysis in modal logic.
\newblock {\em JPL}, 34(5-6):507--544, 2005.

\bibitem{Negri2001}
Sara Negri and Jan von Plato.
\newblock {\em Structural Proof Theory}.
\newblock CUP, 2001.

\bibitem{Ohearn1999}
Peter~W. O'Hearn and David~J. Pym.
\newblock The logic of bunched implications.
\newblock {\em BSL}, 5(2):215--244, 1999.

\bibitem{Park2013}
Jonghyun Park, Jeongbong Seo, and Sungwoo Park.
\newblock A theorem prover for {B}oolean {BI}.
\newblock POPL '13, pages 219--232, New York, NY, USA, 2013. ACM.

\bibitem{Pym2002}
David~J. Pym.
\newblock {\em The Semantics and Proof Theory of the Logic of Bunched
  Implications}.
\newblock Applied Logic Series. Kluwer Academic Publishers, 2002.

\bibitem{Reynolds2002LICS}
John~C. Reynolds.
\newblock Separation logic: A logic for shared mutable data structures.
\newblock LICS '02, pages 55--74. IEEE Computer Society, 2002.

\bibitem{Troelstra96}
Anne~S. Troelstra and Helmut Schwichtenberg.
\newblock {\em Basic Proof Theory}.
\newblock CUP, 1996.

\end{thebibliography}


\appendix

\section{Appendix}
\label{app:proofs}

This section provides the details of the proofs in this paper.

\subsection{Soundness of $\lsbbi$}
\label{app:sound_lsbbi}

Proof for Theorem~\ref{thm:soundness}.

\begin{proof}
To prove the soundness of $\lsbbi$, we show that each rule preserves
falsifiability upwards, as this is a more natural direction in terms
of backward proof search. Therefore to prove that a rule is sound, we
need to show that if the conclusion is falsifiable, then at least one
of the premises is falsifiable (usually in the same choice of $v$,
$\rho$, and $\mathcal{M}$). Most of the cases are easy, we show some
samples here.

\begin{description}
\item[$id$] Since there is no premise in this rule, we simply need 
to show that the conclusion is not falsifiable.

Suppose the sequent $\Gamma;w:P \vdash w:P;\Delta$ is falsifiable,
then $\Gamma$ must be true and $\rho(w)\Vdash A$ and $\rho(w) \not
\Vdash A$ and $\Delta$ must be false. However, $\rho(w)\Vdash A$ and
$\rho(w) \not \Vdash A$ cannot hold at the same time for any
$(\mathcal{M},\simp,\epsilon)$, $v$ and $\rho$, so we have a
contradiction, thus this sequent is not falsifiable.

\item[$\top^* L$] Assume $\Gamma;w:\top^* \vdash \Delta$ is falsifiable, then $\Gamma$ is true and $\rho(w)\Vdash \top^*$ and $\Delta$ is false.

From the semantics of $\top^*$ we know that $\rho(w)\Vdash \top^*$ iff
$\rho(w) = \epsilon$. Therefore by choosing the same $\rho$, $v$, and
$\mathcal{M}$ for the premise, replacing every $w$ by $\epsilon$ in
$\Gamma$ and $\Delta$ preserves their valuations, as we know that
$\rho(\epsilon) = \epsilon$. That is, $\Gamma[\epsilon/w]$ must be
true and $\Delta[\epsilon/w]$ must be false. So the premise is
falsifiable.

\item[$\mand L$] Assume the conclusion is falsifiable, so under some $v$, $\rho$, $\mathcal{M}$, we have that $\Gamma$ is true and $\rho(z)\Vdash A\mand B$ and $\Delta$ is false.

From the semantics of $A\mand B$, we know that $\exists a,b$
s.t. $a,b\simp \rho(z)$ and $a \Vdash A$ and $b \Vdash B$. So we can
choose a mapping $\rho'$ with $\rho' = (x\mapsto a)\cup (y\mapsto b)
\cup \rho$. Since $x$ and $y$ are fresh, they should not affect
anything in $\rho$. Then, under $\rho'$, the following hold: $(x,y\simp z)$ is true
and $\Gamma$ is true and $\rho'(x)\Vdash A$ and $\rho'(y) \Vdash B$
and $\Delta$ is false. Thus the premise is falsifiable in $v$,
$\rho'$, and $\mathcal{M}$.

\item[$\mand R$] Assume under some $v$, $\rho$, and $\mathcal{M}$, $(x,y\simp z)$ is true and $\Gamma$ is true and $\rho(z)\not \Vdash A\mand B$ and $\Delta$ is false.

The semantics of $A\mand B$ yields the following:
\begin{align*}
\rho(z) \not \Vdash A\mand B & \Leftrightarrow \lnot (\exists a,b. \ (a,b\simp \rho(z) \ and \  a \Vdash A \  and \ b \Vdash B))\\
& \Leftrightarrow \forall a,b. \ (a,b\simp \rho(z) \ doesn't \ hold \ or \ a\not \Vdash A \ or \ b \not \Vdash B) 
\end{align*}

If we pick the same set of $v$, $\rho$, $\mathcal{M}$ for the
premises, however, in both premises the relational atom $(x,y\simp z)$
already exists, which means $\rho(x),\rho(y) \simp \rho(z)$ holds. So
the possibility is only that either $\rho(x)\not\Vdash A$ or
$\rho(y)\not\Vdash B$. Assume the former one holds, then the left
premise is falsifiable, otherwise the right premise is falsifiable.

\end{description}

Rules for additive connectives are straightforward, the cases for
$\mimp$ can be proved similarly as for $\mand$ above. Structural rules
$E$, $A$ (and $A_C$), $Eq_1$ (and $Eq_2$ and $U$) can be proved by
using the commutativity, associativity, and identity properties of the
monoid structure respectively.\qed
\end{proof}

\subsection{Substitution for labels}
\label{app:subs_labels}

The proof for Lemma~\ref{subs}.

\begin{proof}
By induction on $ht(\Pi).$

\noindent (Base case) If $ht(\Pi) = 0$, then the only applicable rules are
$id$, $\bot L$, $\top R$ and $\top^* R$. If the label $x \not =
\epsilon$ being substituted is not on the principal formula, then the
substitution does not affect the original derivation. Note that since
we do not allow to substitute for the label $\epsilon$, the proof for
$\top^* R$ can only be this case. Otherwise we obtain the new
derivation by simply replacing the label of the principal formula.

\noindent (Inductive case) If $ht(\Pi) > 0$, then consider the last rule applied in the derivation. We consider three main cases.
\begin{enumerate}
\item Neither $x$ nor $y$ is the label of the principal formula. 
\begin{enumerate}
\item Suppose the last rule applied is $\top^* L$, and $x \not = w$ and $y \not = w$, and $\Pi$ is the following derivation:
\begin{center}
\alwaysNoLine
\AxiomC{$\Pi_1$}
\UnaryInfC{$\Gamma'[\epsilon/w] \vdash \Delta[\epsilon/w]$}
\alwaysSingleLine
\RightLabel{$\top^* L$}
\UnaryInfC{$\Gamma'; w:\top^* \vdash \Delta$}
\DisplayProof
\end{center}

By the induction hypothesis, there is a derivation $\Pi_1'$ of
$\Gamma'[\epsilon/w][y/x] \vdash \Delta[\epsilon/w][y/x]$ with $ht(\Pi_1') \leq ht(\Pi_1)$. 
Since $x$ and $y$ are different from $w$, this sequent is equal to
$\Gamma'[y/x][\epsilon/w] \vdash \Delta[y/x][\epsilon/w]$. Therefore $\Pi'$ 
is constructed as follows.
\begin{center}
\alwaysNoLine
\AxiomC{$\Pi_1'$}
\UnaryInfC{$\Gamma'[y/x][\epsilon/w] \vdash \Delta[y/x][\epsilon/w]$}
\alwaysSingleLine
\RightLabel{$\top^* L$}
\UnaryInfC{$\Gamma'[y/x];w:\top^* \vdash \Delta[y/x]$}
\DisplayProof
\end{center}
Obviously $ht(\Pi') \leq ht(\Pi).$

\item If the last rule applied is $Eq_1$, we distinguish the following cases: $x$ is not $w$ or $w'$; $x = w$; $x = w'$.
\begin{enumerate}
\item $x \not = w$ and $x \not = w'$. The original derivation is as
  follows.
\begin{center}
\alwaysNoLine
\AxiomC{$\Pi_1$}
\UnaryInfC{$(\epsilon,w\simp w);\Gamma'[w/w'] \vdash \Delta[w/w']$}
\alwaysSingleLine
\RightLabel{$Eq_1$}
\UnaryInfC{$(\epsilon,w'\simp w);\Gamma' \vdash \Delta$}
\DisplayProof
\end{center}

\begin{enumerate}
\item If $y \not = w$ and $y \not = w'$, by the induction hypothesis,
  there is a derivation $\Pi_1'$ of $(\epsilon,w\simp w);\Gamma'[w/w'][y/x]
  \vdash \Delta[w/w'][y/x]$ with $ht(\Pi_1') \leq ht(\Pi_1)$. 
  Since $x$, $y$, $w$, $w'$ are different
  labels, this sequent is equal to $(\epsilon,w\simp w);\Gamma'[y/x][w/w'] \vdash \Delta[y/x][w/w']$. 
  Thus the derivation $\Pi'$ is constructed as follows.
\begin{center}
\alwaysNoLine
\AxiomC{$\Pi_1'$}
\UnaryInfC{$(\epsilon,w\simp w);\Gamma'[y/x][w/w'] \vdash \Delta[y/x][w/w']$}
\alwaysSingleLine
\RightLabel{$Eq_1$}
\UnaryInfC{$(\epsilon,w'\simp w);\Gamma'[y/x] \vdash \Delta[y/x]$}
\DisplayProof
\end{center}

\item If $y = w$, this case is similar to Case 1.(b).i.A.

\item Suppose $y = w'$.  Then we need to derive $(\epsilon,y\simp
  w);\Gamma'[y/x] \vdash \Delta[y/x]$. 
  If $y \not = \epsilon$, we construct $\Pi'$ by first
applying $Eq_1$ bottom-up:
\begin{center}
\AxiomC{$(\epsilon,w\simp w);\Gamma'[y/x][w/y] \vdash \Delta[y/x][w/y]$}
\RightLabel{$Eq_1$}
\UnaryInfC{$(\epsilon,y\simp w);\Gamma'[y/x] \vdash \Delta[y/x]$}
\DisplayProof
\end{center}
Now the premise is equal to $(\epsilon,w\simp w);\Gamma'[w/y][w/x]
\vdash \Delta[w/y][w/x]$, and by the induction hypothesis, there is a
derivation $\Pi_1'$ of this sequent, with $ht(\Pi_1') \leq ht(\Pi_1)$.

If $y = \epsilon$, then we need to apply $Eq_2$, instead of $Eq_1$: 
\begin{prooftree}
\AxiomC{$(\epsilon,\epsilon \simp \epsilon);\Gamma'[\epsilon/x][\epsilon/w] 
   \vdash \Delta[\epsilon/x][\epsilon/w]$}
\RightLabel{$Eq_2$}
\UnaryInfC{$(\epsilon,\epsilon \simp w);\Gamma'[\epsilon/x] \vdash \Delta[\epsilon/x]$}
\end{prooftree}
Note that the sequent 
$(\epsilon,\epsilon \simp \epsilon);\Gamma'[\epsilon/x][\epsilon/w] \vdash \Delta[\epsilon/x][\epsilon/w]$
is the same as 
$$(\epsilon,\epsilon \simp \epsilon);\Gamma'[w/w'][\epsilon/w][\epsilon/x] 
\vdash \Delta[w/w'][\epsilon/w][\epsilon/x].$$
So the premise can be proved by two successive applications of the induction hypothesis to $\Pi_1$, one using
substitution $[\epsilon/w]$ and the other using substitution $[\epsilon/x].$ Here we can apply 
the induction hypothesis twice to $\Pi_1$ because substitution does not increase the height of derivations. 

\end{enumerate}

\item $x = w$ (so $w$ cannot be $\epsilon$). 

\begin{enumerate}
\item If $y \not = w'$, then $\Pi$ has the form: 
\begin{center}
\alwaysNoLine
\AxiomC{$\Pi_1$}
\UnaryInfC{$(\epsilon,x\simp x);\Gamma'[x/w'] \vdash \Delta[x/w']$}
\alwaysSingleLine
\RightLabel{$Eq_1$}
\UnaryInfC{$(\epsilon,w'\simp x);\Gamma' \vdash \Delta$}
\DisplayProof
\end{center}

By the induction hypothesis we have the folowing derivation:  
\begin{center}
\alwaysNoLine
\AxiomC{$\Pi_1'$}
\UnaryInfC{$(\epsilon,y\simp y);\Gamma'[x/w'][y/x] \vdash \Delta[x/w'][y/x]$}
\DisplayProof
\end{center}

The end sequent is equal to the following: 
\begin{center}
$(\epsilon,y\simp y);\Gamma'[y/x][y/w'] \vdash
  \Delta[y/x][y/w']$. 
\end{center}
Then by using $Eq_1$, we construct $\Pi'$ as follows.
\begin{center}
\alwaysNoLine
\AxiomC{$\Pi_1'$}
\UnaryInfC{$(\epsilon,y\simp y);\Gamma'[y/x][y/w'] \vdash \Delta[y/x][y/w']$}
\alwaysSingleLine
\RightLabel{$Eq_1$}
\UnaryInfC{$(\epsilon,w'\simp y);\Gamma'[y/x] \vdash \Delta[y/x]$}
\DisplayProof
\end{center}

\item If $y = w'$, then $\Pi$ has the form:
\begin{center}
\alwaysNoLine
\AxiomC{$\Pi_1$}
\UnaryInfC{$(\epsilon,x\simp x);\Gamma'[x/y] \vdash \Delta[x/y]$}
\alwaysSingleLine
\RightLabel{$Eq_1$}
\UnaryInfC{$(\epsilon,y\simp x);\Gamma' \vdash \Delta$}
\DisplayProof
\end{center}

By the induction hypothesis, we have the following derivation:
\begin{center}
\alwaysNoLine
\AxiomC{$\Pi_1'$}
\UnaryInfC{$(\epsilon,y\simp y);\Gamma'[x/y][y/x] \vdash \Delta[x/y][y/x]$}
\DisplayProof
\end{center}

Since in the end sequent, we replace every $y$ by $x$, and then
change every $x$ back to $y$, the effect is the same as just keeping
every $y$ unchanged and only replace every $x$ by $y$. Thus the end
sequent is equal to:
\begin{center}
$(\epsilon,y\simp y);\Gamma'[y/x] \vdash \Delta[y/x]$
\end{center}
which is exactly what we need to derive. Therefore we let $\Pi' = \Pi_1'$.
Notice that in this case $ht(\Pi') < ht(\Pi).$
\end{enumerate}

\item $x = w'$.

\begin{enumerate}
\item If $y \not = w$ and $y \not = \epsilon$, the original derivation is as follows.
\begin{center}
\alwaysNoLine
\AxiomC{$\Pi_1$}
\UnaryInfC{$(\epsilon,w\simp w);\Gamma'[w/x] \vdash \Delta[w/x]$}
\alwaysSingleLine
\RightLabel{$Eq_1$}
\UnaryInfC{$(\epsilon,x\simp w);\Gamma' \vdash \Delta$}
\DisplayProof
\end{center}

By the induction hypothesis (instead of replacing every $x$ by $y$, we now
replace every $y$ by $w$), we have the following derivation:
\begin{center}
\alwaysNoLine
\AxiomC{$\Pi_1'$}
\UnaryInfC{$(\epsilon,w\simp w);\Gamma'[w/x][w/y] \vdash
  \Delta[w/x][w/y]$}
\DisplayProof
\end{center}

The end sequent is equal to:

\begin{center}
$(\epsilon,w\simp w);\Gamma'[y/x][w/y] \vdash
  \Delta[y/x][w/y]$
\end{center}

Thus $\Pi'$ is constructed as follows.
\begin{center}
\alwaysNoLine
\AxiomC{$\Pi_1'$}
\UnaryInfC{$(\epsilon,w\simp w);\Gamma'[y/x][w/y] \vdash
  \Delta[y/x][w/y]$}
\alwaysSingleLine
\RightLabel{$Eq_1$}
\UnaryInfC{$(\epsilon,y\simp w);\Gamma'[y/x] \vdash \Delta[y/x]$}
\DisplayProof
\end{center}

\item If $y = \epsilon$ and $w \not = \epsilon$, we need to derive the following sequent:
\begin{center}
$(\epsilon,\epsilon\simp w);\Gamma'[\epsilon/x] \vdash \Delta[\epsilon/x]$
\end{center}

By induction hypothesis, replacing every $w$ by $\epsilon$ in $\Pi_1$, then using
the rule $Eq_2$, we get the
new derivation:
\begin{center}
\alwaysNoLine
\AxiomC{$\Pi_1'$}
\UnaryInfC{$(\epsilon,\epsilon\simp \epsilon);\Gamma'[\epsilon/x][\epsilon/w] \vdash
  \Delta[\epsilon/x][\epsilon/w]$}
\alwaysSingleLine
\RightLabel{$Eq_2$}
\UnaryInfC{$(\epsilon, \epsilon\simp w);\Gamma'[\epsilon/x] \vdash \Delta[\epsilon/x]$}
\DisplayProof
\end{center}

\item If $y = w$, then the premise of the last rule is
  exactly what we need to derive.

\end{enumerate}
\end{enumerate}

\item If the last rule applied is $Eq_2$, we consider three cases: $x
  \not = w$ and $y \not = w$; $x = w$; and $y = w$.
These are symmetric to the case where the last rule is $Eq_1$, already discussed above. 

\end{enumerate}

\item $y$ is the label of the principal formula. Most of the cases
  follow similarly as above, except for $\top^* L$. In this case the
  original derivation is as follows.
\begin{center}
\alwaysNoLine
\AxiomC{$\Pi_1$}
\UnaryInfC{$\Gamma'[\epsilon / y] \vdash \Delta[\epsilon /y]$}
\alwaysSingleLine
\RightLabel{$\top^* L$}
\UnaryInfC{$\Gamma';y:\top^* \vdash \Delta$}
\DisplayProof
\end{center}

Our goal is to derive $\Gamma'[y/x]; y: \top^* \vdash \Delta[y/x]$. Applying $\top^* L$ as in backward proof search, we get 
\begin{center}
$\Gamma'[y/x][\epsilon / y] \vdash \Delta [y/x][\epsilon /y]$
\end{center}

Note that this sequent is equal to $\Gamma'[\epsilon/y][\epsilon/x] \vdash \Delta [\epsilon/y][\epsilon/x]$, 
and from induction hypothesis we know that there is a derivation of this sequent of height 
less than or equal to $ht(\Pi)$.

\item $x$ is the label of the principal formula.
\begin{enumerate}
\item For the additive rules, since the labels stay the same
  in the premises and conclusions of the rules, even if the label of the principal formula is replaced by
  some other label, we can still apply the induction hypothesis on the
  premise, then use the rule to derive the conclusion.

For $\land L$,
\begin{center}
\alwaysNoLine
\AxiomC{$\Pi_1$}
\UnaryInfC{$\Gamma';x:A;x:B \vdash \Delta$}
\alwaysSingleLine
\RightLabel{$\land L$}
\UnaryInfC{$\Gamma';x:A\land B \vdash \Delta$}
\DisplayProof \qquad $\leadsto$ \qquad
\alwaysNoLine
\AxiomC{$\Pi_1'$}
\UnaryInfC{$\Gamma'[y/x];y:A;y:B \vdash \Delta[y/x]$}
\alwaysSingleLine
\RightLabel{$\land L$}
\UnaryInfC{$\Gamma'[y/x];y:A\land B \vdash \Delta[y/x]$}
\DisplayProof
\end{center}

For $\land R$,
\begin{center}
\alwaysNoLine
\AxiomC{$\Pi_1$}
\UnaryInfC{$\Gamma' \vdash x:A;\Delta$}
\AxiomC{$\Pi_2$}
\UnaryInfC{$\Gamma' \vdash x:B;\Delta$}
\alwaysSingleLine
\RightLabel{$\land R$}
\BinaryInfC{$\Gamma' \vdash x:A\land B;\Delta$}
\DisplayProof $\leadsto$\\[10px]
\alwaysNoLine
\AxiomC{$\Pi'_1$}
\UnaryInfC{$\Gamma'[y/x] \vdash y:A;\Delta[y/x]$}
\AxiomC{$\Pi'_2$}
\UnaryInfC{$\Gamma'[y/x] \vdash y:B;\Delta[y/x]$}
\alwaysSingleLine
\RightLabel{$\land R$}
\BinaryInfC{$\Gamma'[y/x] \vdash y:A\land B;\Delta[y/x]$}
\DisplayProof
\end{center}

For $\limp L$,
\begin{center}
\alwaysNoLine
\AxiomC{$\Pi_1$}
\UnaryInfC{$\Gamma' \vdash x:A;\Delta$}
\AxiomC{$\Pi_2$}
\UnaryInfC{$\Gamma';x:B \vdash \Delta$}
\alwaysSingleLine
\RightLabel{$\limp L$}
\BinaryInfC{$\Gamma'; x:A\limp B \vdash \Delta$}
\DisplayProof $\leadsto$\\[10px]
\alwaysNoLine
\AxiomC{$\Pi'_1$}
\UnaryInfC{$\Gamma'[y/x] \vdash y:A;\Delta[y/x]$}
\AxiomC{$\Pi'_2$}
\UnaryInfC{$\Gamma'[y/x];y:B \vdash \Delta[y/x]$}
\alwaysSingleLine
\RightLabel{$\limp L$}
\BinaryInfC{$\Gamma'[y/x]; y:A\limp B \vdash \Delta[y/x]$}
\DisplayProof
\end{center}

For $\limp R$,
\begin{center}
\alwaysNoLine
\AxiomC{$\Pi_1$}
\UnaryInfC{$\Gamma';x:A \vdash x:B;\Delta$}
\alwaysSingleLine
\RightLabel{$\limp R$}
\UnaryInfC{$\Gamma' \vdash x:A\limp B;\Delta$}
\DisplayProof \qquad $\leadsto$ \qquad
\alwaysNoLine
\AxiomC{$\Pi_1'$}
\UnaryInfC{$\Gamma'[y/x];y:A \vdash y:B;\Delta[y/x]$}
\alwaysSingleLine
\RightLabel{$\limp R$}
\UnaryInfC{$\Gamma'[y/x] \vdash y:A\limp B;\Delta[y/x]$}
\DisplayProof
\end{center}

\item For multiplicative rules that do not produce eigenvariables ($\mand
  R, \mimp L, \top^* L$), we can proceed similarly as in the additive
  cases, except for the $\top^* L$ rule. For the $\top^* L$ rule, if
  the label $x$ of the principal formula is replaced by some (other)
  label $y$, i.e., $\Pi$ is 
\begin{prooftree}
\AxiomC{$\Pi_1$}
\noLine
\UnaryInfC{$\Gamma'[\epsilon/x] \vdash \Delta[\epsilon/x]$}
\RightLabel{$\top^* L$}
\UnaryInfC{$\Gamma' ; x : \top^* \vdash \Delta$}
\end{prooftree}
then we then need a derivation of the sequent 
$\Gamma'[y/x];y:\top^* \vdash \Delta [y/x]$. Using $\top^* L$ rule we have:
\begin{center}
\AxiomC{$\Gamma [y/x][\epsilon/y] \vdash \Delta[y/x][\epsilon/y]$}
\RightLabel{$\top^* L$}
\UnaryInfC{$\Gamma [y/x];y:\top^* \vdash \Delta [y/x]$}
\DisplayProof
\end{center}
Note that the premise now is equal to $\Gamma [\epsilon/x][\epsilon/y]
\vdash \Delta[\epsilon/x][\epsilon/y]$, and can be proved using the induction
hypothesis on $\Pi_1$. 

If $y = \epsilon$, then $\Pi'$ is obtained by applying  
Lemma~\ref{weak_fr} to $\Pi_1.$ 

\item For the multiplicative rules that have eigenvariables ($\mand L$
  and $\mimp R$), if the label of the principal formula is replaced by
  a label other than the newly created labels in the rules, then we
  proceed similarly as in additive cases. If the label of the
  principal formula is replaced by one of the newly created labels,
  then we just need to create a different new label in the new
  relation.

For $\mand L$, we have the derivation:
\begin{center}
\alwaysNoLine
\AxiomC{$\Pi_1$}
\UnaryInfC{$(y,z \simp x);\Gamma';y:A;z:B \vdash \Delta$}
\alwaysSingleLine
\RightLabel{$\mand L$}
\UnaryInfC{$\Gamma';x:A\mand B \vdash \Delta$}
\DisplayProof
\end{center}

If $x$ is substituted by $y$ (the case for substituting to $z$ is
symmetric), then we need a derivation of $\Gamma'[y/x];y:A\mand B
\vdash \Delta[y/x]$. Note that since the $\mand L$ rule requires the
relation $(y,z\simp x)$ to be fresh, so in the original derivation $y$
and $z$ cannot be in $\Gamma$ or $\Delta$. Therefore by induction
hypothesis we must have a derivation $\Pi_1'$ for 
$$(y',z' \simp x);\Gamma';y':A;z':B \vdash \Delta,$$ 
where $y'$ and $z'$ are new
labels, such that $ht(\Pi_1') \leq ht(\Pi_1)$. 
Applying the induction hypothesis again to $\Pi_1'$, 
we have a derivation $\Pi_1''$ 
$(y',z' \simp y);\Gamma'[y/x];y':A;z':B \vdash \Delta[y/x]$,
with $ht(\Pi_1'') \leq ht(\Pi_1)$. Thus the derivation $\Pi'$ is constructed
as follows.
\begin{center}
\alwaysNoLine
\AxiomC{$\Pi_1''$}
\UnaryInfC{$(y',z' \simp y);\Gamma'[y/x];y':A;z':B \vdash \Delta[y/x]$}
\alwaysSingleLine
\RightLabel{$\mand L$}
\UnaryInfC{$\Gamma'[y/x];y:A\mand B \vdash \Delta[y/x]$}
\DisplayProof
\end{center}

The case for $\mimp R$ is similar. suppose $\Pi$ is:
\begin{center}
\alwaysNoLine
\AxiomC{$\Pi_1$}
\UnaryInfC{$(y,x \simp z);\Gamma;y:A \vdash z:B;\Delta'$}
\alwaysSingleLine
\RightLabel{$\mimp R$}
\UnaryInfC{$\Gamma \vdash x:A\mimp B;\Delta'$}
\DisplayProof
\end{center}

If $x$ is replaced by $y$, then we have the following derivation.
\begin{center}
\alwaysNoLine
\AxiomC{$\Pi_1'$}
\UnaryInfC{$(y',y \simp z');\Gamma[y/x];y':A \vdash z':B;\Delta'[y/x]$}
\alwaysSingleLine
\RightLabel{$\mimp R$}
\UnaryInfC{$\Gamma[y/x] \vdash y:A\mimp B;\Delta'[y/x]$}
\DisplayProof
\end{center}

If $x$ is replaced by $z$, then we have the following derivation.
\begin{center}
\alwaysNoLine
\AxiomC{$\Pi_1'$}
\UnaryInfC{$(y',z \simp z');\Gamma[z/x];y':A \vdash z':B;\Delta'[z/x]$}
\alwaysSingleLine
\RightLabel{$\mimp R$}
\UnaryInfC{$\Gamma[z/x] \vdash z:A\mimp B;\Delta'[z/x]$}
\DisplayProof
\end{center}

\end{enumerate}
\end{enumerate}
\qed
\end{proof}

\subsection{Weakening admissibility of $\lsbbi$}
\label{app:weak_lsbbi}
\begin{lemma}
\label{weak_fr}
For all structures $\Gamma,\Delta$, labelled formula $w:A$, and ternary relation $(x,y\simp z)$, 
if $\Gamma \vdash \Delta$ is derivable, then there exists a derivation of the same height for each of 
the following sequents:
$$
\Gamma;w:A \vdash \Delta
\qquad
\Gamma \vdash w:A;\Delta
\qquad
(x,y\simp z);\Gamma \vdash \Delta.
$$
\end{lemma}

\begin{proof}
By induction on $ht(\Pi)$.  
Since $id$, $\bot L$, $\top
R$, and $\top^* R$ all have weakening built in, the base case
trivially holds. For the inductive cases, the only nontrivial case is for
$\mand L$ and $\mimp R$, where new labels have to be
introduced. These labels can be systematically renamed to make sure that they 
do not clash with the labels in the weakened formula/relational atom.\qed
\end{proof}

This yields the proof for Lemma~\ref{lm:weak} in the paper. Furthermore, we can prove more useful lemmas based on the weakening property.

The next lemma shows that the assumption $\epsilon:\top^*$ in the antecedent of
a sequent is not used in any derivation, and since there is no rule that can be applied to 
it, so it can be removed without affecting provability. 

\begin{lemma}
\label{topstar}
If $\Gamma; \epsilon:\top^* \vdash \Delta$ is derivable, then $\Gamma \vdash \Delta$
is derivable with the same series of rule applications.
\end{lemma}

\begin{proof}
By a straightforward induction on the height of derivation $n$.

\noindent (Base case) If $n = 0$, then $\Gamma; \epsilon:\top^* \vdash
\Delta$ must be the conclusion of one of $id$, $\bot L$, $\top R$,
$\top^* R$. Note that $\epsilon:\top^*$ in the antecedent cannot be
the principal formula of any of those rules, therefore those rules are
applicable to $\Gamma \vdash \Delta$ as well.

\noindent (Inductive case) If $n > 0$, consider the last rule in the
derivation. It is obvious that $\epsilon:\top^*$ in the antecedent of
a sequent cannot be the principal formula of any rules, therefore it
has to appear in the premise(s) of the last rule. Thus we can apply
the induction hypothesis on the premise(s) and then use the corresponding
rule to derive $\Gamma \vdash \Delta$.\qed
\end{proof}

In general, if a formula is never principal in a derivation, it can obviously
be omitted.

\begin{lemma}
\label{r_weak_f} 
If $w:A$ is not the principal formula of any rule application in the derivation 
of $\Gamma;w:A\vdash \Delta$ ($\Gamma \vdash w:A;\Delta$ resp.), then there is a 
derivation of $\Gamma \vdash \Delta$ with the same series of rule applications.
\end{lemma}


If we combine the above lemma and the admissibility of weakening, then we can replace a formula that is never used in a derivation by any structure.

\begin{lemma}
\label{subs_str}
If $w:A$ is not the principal formula of any rule application (even
though the label might be changed) in the
derivation of $\Gamma;w:A\vdash \Delta$ ($\Gamma \vdash w:A;\Delta$
resp.), then there is a derivation of $\Gamma;\Gamma' \vdash \Delta$
($\Gamma \vdash \Delta';\Delta$ resp.), and in the new derivation, the
structure $\Gamma$ ($\Delta$ resp.) is not altered except that certain
labels in $\Gamma$ ($\Delta$ resp.) are changed.
\end{lemma}

\begin{proof}
By induction on the height of derivation $n$.

\noindent (Base case) If $n = 0$, since $w:A$ is not the principal formula, the
substituted sequent is also the conlcusion of rules $id$, $\bot L$,
$\top R$, $\top^* R$. This is the same as the base case of the proof
for Lemma~\ref{topstar}.

\noindent (Inductive case) If $n > 0$, consider the last rule in the
derivation. Since $w:A$ is not the principal formula, for all rules
except $\top^* L$, the original derivation has $w:A$ in the premise(s)
of the last rule, therefore we can apply the induction hypothesis on the
premise(s) and then use the rule to get the desired derivation. We
give an example here.

For $\land L$, suppose $w:A$ is in the antecedent, the original
derivation is converted as follows.
\begin{center}
\alwaysNoLine
\AxiomC{$\Pi$}
\UnaryInfC{$\Gamma;w:A;x:B;x:C \vdash \Delta$}
\alwaysSingleLine
\RightLabel{$\land L$}
\UnaryInfC{$\Gamma;w:A;x:B\land C \vdash \Delta$}
\DisplayProof \qquad $\leadsto$ \qquad
\alwaysNoLine
\AxiomC{$\Pi'$}
\UnaryInfC{$\Gamma;\Gamma';x:B;x:C \vdash \Delta$}
\alwaysSingleLine
\RightLabel{$\land L$}
\UnaryInfC{$\Gamma;\Gamma';x:B\land C \vdash \Delta$}
\DisplayProof 
\end{center}

Other cases except $\top^* L$ are similar.

If the last rule is $\top^* L$, then we convert the derivation as
follows.
\begin{center}
\alwaysNoLine
\AxiomC{$\Pi$}
\UnaryInfC{$\Gamma[\epsilon/w];\epsilon:A \vdash \Delta[\epsilon/w]$}
\alwaysSingleLine
\RightLabel{$\top^* L$}
\UnaryInfC{$\Gamma;w:A;w:\top^* \vdash \Delta$}
\DisplayProof \qquad $\leadsto$ \qquad
\alwaysNoLine
\AxiomC{$\Pi'$}
\UnaryInfC{$\Gamma[\epsilon/w];\Gamma'[\epsilon/w] \vdash \Delta[\epsilon/w]$}
\alwaysSingleLine
\RightLabel{$\top^* L$}
\UnaryInfC{$\Gamma;\Gamma';w:\top^* \vdash \Delta$}
\DisplayProof 
\end{center}

Note that we incorporate two steps here. First, by induction
hypothesis, we have a derivation of $\Gamma[\epsilon/w];\Gamma' \vdash
\Delta[\epsilon/w]$. Then by the Substitution Lemma, there is a derivation $\Pi'$
of $\Gamma[\epsilon/w];\Gamma'[\epsilon/w] \vdash
\Delta[\epsilon/w]$, from which we can derive the final sequent.

Therefore the only change to $\Gamma'$ in the new derivation is that
some of its labels might be changed by the rules $\top^* L$, $Eq_1$,
or $Eq_2$.\qed
\end{proof}

\begin{note}
The admissibility of general weakening shows that if $\Gamma \vdash
\Delta$ is derivable, then $\Gamma;\Gamma' \vdash \Delta;\Delta'$ is
derivable. A stronger argument here is that in the derivation of the
latter sequent, $\Gamma'$ and $\Delta'$ are never changed except that
some labels might be changed. This is similar as in Lemma~\ref{subs_str}.
\end{note} 

\subsection{Invertibility of rules in $\lsbbi$}
\label{app:invert}

Proof for Lemma~\ref{lm:invert}.

\begin{proof}
As the additive rules in $\lsbbi$ are exactly the same as those in
Negri's labelled system for Modal logic or $G3c$
(cf.~\cite{Negri2001}), the proof for them is similar. The main
difference is that the rest of our rules are of different
forms. However, as most of our rules do not modify the side
structures, simply by applying the induction hypothesis and then using
the corresponding rule, we get the new derivation. The cases where the
last rule applied is $\top^* L$, $Eq_1$, or $Eq_2$ follow essentially
the same, except a global substitution needs to be considered, but
that is of no harm.

Rules $E$, $A$, $U$, $A_C$, $\mand R$ and $\mimp L$ are trivially
invertible as the conclusion is a subset of the premise, and weakening
is height-preserving admissible.

To prove the cases for $\mand L$ and $\mimp R$, we do inductions on the
height $n$ of the derivation. In each case below, it is obvious that 
each premise is always cut-free derivable with less or same height as the conclusion.

The case for $\mand L$ is as follows.

\noindent (Base case) If $n = 0$, then the conclusion of $\mand L$
is one of the conlucsions of $id$, $\bot L$, $\top R$, $\top^* R$,
notice that the identity rule is restricted to propositions, therefore
the
premise of $\mand L$ is also the conclusions of the
corresponding axiom rule.

\noindent (Inductive case) If $n > 0$, and the last rule applied is not $\mand
L$ or $\mimp R$, then no fresh labels are involved, so we can safely
apply the induction hypothesis on the premise of the last rule and then
use the rule to get the derivation. If the last rule is $\mand L$ or
$\mimp R$, but the principal formula is in $\Gamma$ or $\Delta$, we
proceed similarly, and use the Substitution Lemma to ensure that the
eigenvariables are new. If the principal formula is $z:A\mand B$, then the
premise of the last rule yields the desired conclusion.

The case for $\mimp R$ follows similarly.

For $\top^* L$, again, we do an induction on the height $n$ of the derivation.

\noindent (Base case) If $n = 0$, then $\Gamma; x:\top^* \vdash
\Delta$ is the conclusion of one of $id$, $\bot L$, $\top R$, $\top^*
R$, and $x:\top^*$ cannot be the principal formula. Note that in the
first three cases the principal formulae can be labelled with
anything. Since, in the sequent $\Gamma[\epsilon/x] \vdash
\Delta[\epsilon/x]$, the label $x$ is uniformly replaced by
$\epsilon$, this sequent can be the conclusion of the corresponding
rule as well. For $\top^* R$, since $\top^*$ on the right hand side
can only be labelled with $\epsilon$, so replacing $x$ to $\epsilon$
does not change its label. Thus this case is not broken either.

\noindent (Inductive case) If $n > 0$, consider the last rule applied in the derivation. 
\begin{enumerate}
\item If the principal formula or relation does not involve the label
  $x$, then we can apply the induction hypothesis directly on the
  premise of the last rule, then use the last rule to get the
  derivation.

\item Otherwise, if the principal formula or relation has label $x$,
  and the last rule is not $\top^* L$, we proceed similarly, except
  replacing the label in the principal relation or formula. The detail
  is exemplified using $\mand L$.

For $\mand L$, we have the following derivation:
\begin{center}
\alwaysNoLine
\AxiomC{$\Pi$}
\UnaryInfC{$(y,z\simp x);\Gamma;x:\top^*;y:A;z:B \vdash \Delta$}
\alwaysSingleLine
\RightLabel{$\mand L$}
\UnaryInfC{$\Gamma;x:\top^*;x:A\mand B \vdash \Delta$}
\DisplayProof
\end{center}

The condition of the rule $\mand L$ guarantees that $y$ and $z$ cannot be in $\Gamma$ and $\Delta$, so we do not have to worry if they are identical to $x$. By applying the induction hypothesis and then using the rule, we get the following derivation:
\begin{center}
\alwaysNoLine
\AxiomC{$\Pi'$}
\UnaryInfC{$(y,z\simp \epsilon);\Gamma[\epsilon/x];y:A;z:B \vdash \Delta[\epsilon/x]$}
\alwaysSingleLine
\RightLabel{$\mand L$}
\UnaryInfC{$\Gamma[\epsilon/x];\epsilon:A\mand B \vdash \Delta[\epsilon/x]$}
\DisplayProof
\end{center}

Another way to do this is by using the Substitution Lemma, replacing $x$ by $\epsilon$, we get a derivation to the premise that has a redundant $\epsilon:\top^*$, since we know that this labelled formula on the left hand side does not contribute to the derivation, we can safely derive the sequent without it using the same inference, cf. Lemma~\ref{r_weak_f}.

The case where the last rule is $\mimp R$ is similar. 

If the last rule is $Eq_1$, we consider the following cases:
\begin{enumerate}
\item The label of $\top^*$ is not in the principal relation (i.e., $x
  \not = w$ and $x \not = w'$). The
  original derivation is as follows.
\begin{center}
\alwaysNoLine
\AxiomC{$\Pi$}
\UnaryInfC{$(\epsilon,w\simp w);\Gamma[w/w'];x:\top^* \vdash
  \Delta[w/w']$}
\alwaysSingleLine
\RightLabel{$Eq_1$}
\UnaryInfC{$(\epsilon,w'\simp w);\Gamma;x:\top^* \vdash
  \Delta$}
\DisplayProof
\end{center}

By the induction hypothesis, we have the following derivation:
\begin{center}
\alwaysNoLine
\AxiomC{$\Pi'$}
\UnaryInfC{$(\epsilon,w\simp w);\Gamma[w/w'][\epsilon/x] \vdash
  \Delta[w/w'][\epsilon/x]$}
\DisplayProof
\end{center}

Note that since $x$, $w$, $w'$ are all different, the end sequent is
equal to the following:
\begin{center}
$(\epsilon,w\simp w);\Gamma [\epsilon/x][w/w'] \vdash
  \Delta [\epsilon/x][w/w']$
\end{center}

from which we can use the rule $Eq_1$ and derive $(\epsilon,w'\simp w);\Gamma [\epsilon/x] \vdash
  \Delta [\epsilon/x]$.

\item $x = w$. The original derivation is as follows.
\begin{center}
\alwaysNoLine
\AxiomC{$\Pi$}
\UnaryInfC{$(\epsilon,x\simp x);\Gamma[x/w'];x:\top^* \vdash
  \Delta[x/w']$}
\alwaysSingleLine
\RightLabel{$Eq_1$}
\UnaryInfC{$(\epsilon,w'\simp x);\Gamma;x:\top^* \vdash
  \Delta$}
\DisplayProof
\end{center}

By the substitution lemma, replacing every $x$ by $\epsilon$ in the
premise of the last rule, we get the following derivation:
\begin{center}
\alwaysNoLine
\AxiomC{$\Pi'$}
\UnaryInfC{$(\epsilon,\epsilon\simp \epsilon);\Gamma[x/w'][\epsilon/x];\epsilon:\top^* \vdash
  \Delta[x/w'][\epsilon/x]$}
\DisplayProof
\end{center}

The end sequent is equal to:
\begin{center}
$(\epsilon,\epsilon\simp \epsilon);\Gamma [\epsilon/x][\epsilon/w'];\epsilon:\top^* \vdash
  \Delta [\epsilon/x][\epsilon/w']$
\end{center}

By Lemma~\ref{topstar}, $\epsilon:\top^*$ in the antecedent can be
omitted. Apply the $Eq_1$ rule on this sequent without $\epsilon:\top^*$, we finally get $(\epsilon,w'\simp \epsilon);\Gamma [\epsilon/x] \vdash
  \Delta [\epsilon/x]$.

\item $x = w'$. The original derivation is as follows.
\begin{center}
\alwaysNoLine
\AxiomC{$\Pi$}
\UnaryInfC{$(\epsilon,w\simp w);\Gamma[w/x];w:\top^* \vdash
  \Delta[w/x]$}
\alwaysSingleLine
\RightLabel{$Eq_1$}
\UnaryInfC{$(\epsilon,x\simp w);\Gamma;x:\top^* \vdash
  \Delta$}
\DisplayProof
\end{center}

By the induction hypothesis, we have the following derivation:
\begin{center}
\alwaysNoLine
\AxiomC{$\Pi'$}
\UnaryInfC{$(\epsilon,\epsilon\simp \epsilon);\Gamma[w/x][\epsilon/w] \vdash
  \Delta[w/x][\epsilon/w]$}
\DisplayProof
\end{center}

Now the end sequent is equal to:
\begin{center}
$(\epsilon,\epsilon\simp \epsilon);\Gamma[\epsilon/x][\epsilon/w] \vdash
  \Delta[\epsilon/x][\epsilon/w]$
\end{center} 

By using the rule $Eq_2$ on this sequent, we derive
$(\epsilon,\epsilon\simp w);\Gamma[\epsilon/x] \vdash
  \Delta[\epsilon/x]$.
\end{enumerate} 

The case where the last rule is $Eq_2$ is similar to the case for $Eq_1.$


 


If the last rule is $\top^* L$,
then the derivation to the premise of the last rule yields the new derivation.

\end{enumerate}

The invertibility of $Eq_1$ and $Eq_2$ follows from the Substitution
Lemma, as the reverse versions of these two rules are only about replacing labels.\qed
\end{proof}

\subsection{Contraction admissibility of $\lsbbi$}
\label{app:ctr_f}

\begin{lemma}
\label{ctr_f}
For all structures $\Gamma,\Delta$, and labelled formula $w:A$, the following holds in $\lsbbi$:
\begin{enumerate}
\item If there is a cut-free derivation $\Pi$ of $\Gamma;w:A;w:A \vdash \Delta$, 
then there is a cut-free derivation $\Pi'$ of $\Gamma;w:A \vdash \Delta$ with $ht(\Pi') \leq ht(\Pi).$
\item If there is a cut-free derivation $\Pi$ of 
$\Gamma \vdash w:A;w:A;\Delta$, then there is a cut-free derivation $\Pi'$ of
$\Gamma \vdash w:A;\Delta$ with $ht(\Pi') \leq ht(\Pi).$
\end{enumerate}
\end{lemma}

\begin{proof}
By simultaneous induction on the height of derivations for the left and
right contraction. Let $n = ht(\Pi).$

\noindent (Base case) If $n = 0$, the premise is one of the conclusions of $id$,
$\bot L$, $\top R$ and $\top^* R$, then the contracted sequent is also
the conclusion of the corresponding rules.

\noindent (Inductive case) If $n > 0$, consider the last rule applied to the
premise of the contraction.

(i) If the contracted formula is not principal in the last rule, then
we can apply the induction hypothesis on the premise(s) of the last rule,
then use the rule to get the derivation. 

(ii) If the contracted formula is the principal formula of the last
rule, we have several cases. For the additive rules the cases are reduced to contraction on
smaller formulae, cf.~\cite{Negri2001}.

For $\top^* L$, we have the following derivation:
\begin{center}
\alwaysNoLine
\AxiomC{$\Pi$}
\UnaryInfC{$\Gamma[\epsilon/x];\epsilon:\top^* \vdash
  \Delta[\epsilon/x]$}
\alwaysSingleLine
\RightLabel{$\top^* L$}
\UnaryInfC{$\Gamma;x:\top^*;x:\top^* \vdash \Delta$}
\DisplayProof
\end{center}

Note that the only case where $\top^*$ is useful on the left hand side
is when it is labelled with a world other than $\epsilon$. Since the substitution 
$[\epsilon/\epsilon]$ does not do anything to the sequent, $\Pi$
can also be the derivation for $\Gamma[\epsilon/x] \vdash \Delta[\epsilon/x]$, cf. Lemma~\ref{topstar}, which leads to $\Gamma;x:\top^* \vdash \Delta$.

For $\mand R$ and $\mimp L$, we can apply the induction hypothesis
directly on the premise of the corresponding rule since the rules carry the principal
formula into the premise(s).

For $\mand L$, we have a derivation as follows.
\begin{center}
\alwaysNoLine
\AxiomC{$\Pi$}
\UnaryInfC{$(x,y\simp z);\Gamma;z:A\mand B;x:A;y:B \vdash \Delta$}
\alwaysSingleLine
\RightLabel{$\mand L$}
\UnaryInfC{$\Gamma;z:A\mand B;z:A\mand B \vdash \Delta$}
\DisplayProof
\end{center}

Apply the Invertibility Lemma on the premise of $\mand L$, we have:
\begin{center}
\alwaysNoLine
\AxiomC{$\Pi'$}
\UnaryInfC{$(x,y\simp z);(x',y'\simp z);\Gamma;x':A;y':B;x:A;y:B \vdash \Delta$}
\DisplayProof
\end{center}

The Substitution Lemma yields a derivation for $(x,y\simp z);(x,y\simp
z);\Gamma;x:A;y:B;x:A;y:B \vdash \Delta$. Apply the induction hypothesis twice
and admissibility of contraction on relational atoms
on this sequent, to get a derivation for $(x,y\simp
z);\Gamma;x:A;y:B \vdash \Delta$. Apply $\mand L$ on this
sequent to get $\Gamma;z:A\mand B \vdash \Delta$.

The case for $\mimp R$ follows similarly. We have a derivation as follows.
\begin{center}
\alwaysNoLine
\AxiomC{$\Pi$}
\UnaryInfC{$(x,y\simp z);\Gamma;x:A \vdash z:B;y:A\mimp B;\Delta$}
\alwaysSingleLine
\RightLabel{$\mimp R$}
\UnaryInfC{$\Gamma \vdash y:A\mimp B;y:A\mimp B;\Delta$}
\DisplayProof
\end{center}

The Invertibility of $\mimp R$ in the premise yields:
\begin{center}
\alwaysNoLine
\AxiomC{$\Pi$}
\UnaryInfC{$(x,y\simp z);(x',y\simp z');\Gamma;x:A;x':A \vdash z:B;z':B;\Delta$}
\DisplayProof
\end{center}

We obtain $(x,y\simp z);(x,y\simp z);\Gamma;x:A;x:A \vdash
z:B;z:B;\Delta$ by the Substitution Lemma. Apply induction hypothesis
twice, and the admissibility of contraction on relations on this sequent, to get $(x,y\simp z);\Gamma;x:A \vdash z:B\Delta$. Finally, apply $\mimp R$, to derive $\Gamma \vdash y:A\mimp B;\Delta$ in the $n$th step.\qed
\end{proof}

\subsection{Cut elimination}
\label{app:cut_elim}

The proof for Theorem~\ref{cut}.

\begin{proof}
By induction on the complexity of the proof in $\lsbbi$. We show that
each application of $cut$ can either be eliminated, or be replaced by
one or more $cut$ rules of less complexity. The argument for termination is similar to the cut-elimination proof for $G3ip$~\cite{Negri2001}. We start to eliminate the topmost $cut$ first, and repeat this procedure until there is no $cut$ in the derivation. We first show that $cut$ can be eliminated when the \textit{cut height} is the lowest, i.e., at least one premise is of height 1. Then we show that the \textit{cut height} is reduced in all cases in which the cut formula is not principal in both premises of cut. If the cut formula is principal in both premises, then the $cut$ is reduced to one or more $cut$s on smaller formulae or shorter derivations. Since atoms cannot be principal in logical rules, finally we can either reduce all $cut$s to the case where the cut formula is not principal in both premises, or reduce those $cut$s on compound formulae until their \textit{cut height}s are minimal and then eliminate those $cut$s.  

\noindent (Base case) If at least one premise of the $cut$ rule is $id$, $\bot L$, $\top R$, or $\top^* R$, we consider the following cases:

\begin{enumerate}
\item The left premise of $cut$ is an application of $id$, and the cut formula is not principal, then the derivation is transformed as follows.
\begin{center}
\AxiomC{}
\RightLabel{$id$}
\UnaryInfC{$\Gamma;y:B\vdash y:B;x:A;\Delta$}
\AxiomC{$\Pi$}
\alwaysNoLine
\UnaryInfC{$\Gamma';x:A \vdash \Delta'$}
\alwaysSingleLine
\RightLabel{$cut$}
\BinaryInfC{$\Gamma;\Gamma';y:B\vdash y:B;\Delta;\Delta'$}
\DisplayProof $\leadsto$\\[10px]
\AxiomC{}
\RightLabel{$id$}
\UnaryInfC{$\Gamma;\Gamma';y:B\vdash y:B;\Delta;\Delta'$}
\DisplayProof
\end{center}

The same transformation works for $\bot L$, $\top R$, $\top^* R$ in this case.

\item The left premise of $cut$ is an application of $id$, and the cut formula is principal, then the derivation is transformed as follows.
\begin{center}
\AxiomC{}
\RightLabel{$id$}
\UnaryInfC{$\Gamma;x:A \vdash x:A;\Delta$}
\AxiomC{$\Pi$}
\alwaysNoLine
\UnaryInfC{$\Gamma';x:A \vdash \Delta'$}
\alwaysSingleLine
\RightLabel{$cut$}
\BinaryInfC{$\Gamma;\Gamma';x:A \vdash \Delta;\Delta'$}
\DisplayProof $\leadsto$\\[10px]
\AxiomC{$\Pi$}
\alwaysNoLine
\UnaryInfC{$\Gamma';x:A \vdash \Delta'$}
\alwaysSingleLine
\dashedLine
\RightLabel{Lemma~\ref{lm:weak}}
\UnaryInfC{$\Gamma;\Gamma';x:A\vdash \Delta;\Delta'$}
\DisplayProof
\end{center}

\item The left premise of $cut$ is an application of $\top R$, and the cut formula is principal, then the derivation is transformed as follows.
\begin{center}
\AxiomC{}
\RightLabel{$\top R$}
\UnaryInfC{$\Gamma \vdash x:\top;\Delta$}
\AxiomC{$\Pi$}
\alwaysNoLine
\UnaryInfC{$\Gamma';x:\top \vdash \Delta'$}
\alwaysSingleLine
\RightLabel{$cut$}
\BinaryInfC{$\Gamma;\Gamma' \vdash \Delta;\Delta'$}
\DisplayProof $\leadsto$\\[10px]
\AxiomC{$\Pi'$}
\alwaysNoLine
\UnaryInfC{$\Gamma' \vdash \Delta'$}
\alwaysSingleLine
\dashedLine
\RightLabel{Lemma~\ref{lm:weak}}
\UnaryInfC{$\Gamma;\Gamma' \vdash \Delta;\Delta'$}
\DisplayProof
\end{center}

As $x:\top$ cannot be a principal formula in the antecedent, by Lemma~\ref{r_weak_f} there is a derivation $\Pi'$ of $\Gamma' \vdash \Delta'$.

The same holds for $\top^* R$.

\item The right premise of $cut$ is an application of $id$, $\bot L$, $\top R$ or $\top^* R$, and the cut formula is not principal. This case is similar to case 1.

\item The right premise of $cut$ is an application of $id$, and the cut formula is principal. This case is similar to case 2.

\item The right premise of $cut$ is an application of $\bot L$, and the cut formula is principal. This case is similar to case 3.

\end{enumerate}

\noindent (Inductive case) If both premises are not in one of the base cases, we distinguish three cases here: the cut formula is not principal in the left premises; the cut formula is only principal in the left premise; and the cut formula is principal in both premises. 
\begin{enumerate}

\item The cut formula is not principal in the left premise. Suppose the left premise ends with
a rule $r$.
\begin{enumerate}
\item If $r$ is $\top^* L$, w.l.o.g. we assume the label of the principal formula is $y$ (which might be equal to $x$). The original derivation is as follows.
\begin{center}
\alwaysNoLine
\AxiomC{$\Pi_1$}
\UnaryInfC{$\Gamma[\epsilon/y] \vdash x:A;\Delta[\epsilon/y]$}
\alwaysSingleLine
\RightLabel{$\top^* L$}
\UnaryInfC{$\Gamma;y:\top^* \vdash x:A;\Delta$}
\alwaysNoLine
\AxiomC{$\Pi_2$}
\UnaryInfC{$\Gamma';x:A\vdash \Delta'$}
\alwaysSingleLine
\RightLabel{$cut$}
\BinaryInfC{$\Gamma;\Gamma';y:\top^*\vdash\Delta;\Delta'$}
\DisplayProof
\end{center}

By the Substitution lemma, 
there is a derivation $\Pi_2'$ of $\Gamma'[\epsilon/y];x:A \vdash \Delta[\epsilon/y]$. 
Thus we can transform the derivation into the following:
\begin{center}
\alwaysNoLine
\AxiomC{$\Pi_1$}
\UnaryInfC{$\Gamma[\epsilon/y] \vdash x:A;\Delta[\epsilon/y]$}
\AxiomC{$\Pi_2'$}
\UnaryInfC{$\Gamma'[\epsilon/y];x:A\vdash \Delta'[\epsilon/y]$}
\alwaysSingleLine
\RightLabel{$cut$}
\BinaryInfC{$\Gamma[\epsilon/y];\Gamma'[\epsilon/y] \vdash\Delta[\epsilon/y];\Delta'[\epsilon/y]$}
\RightLabel{$\top^* L$}
\UnaryInfC{$\Gamma;\Gamma';y:\top^*\vdash\Delta;\Delta'$}
\DisplayProof
\end{center}

If $x = y$ in the original derivation, then the new derivation cuts on $\epsilon:A$ instead. As substitution is height preserving, the cut height in this case is reduced as well.

\item If $r$ is $Eq_1$, and the label $x$ of the principal formula is not equal to $w'$, the original derivation is as follows.
\begin{center}
\alwaysNoLine
\AxiomC{$\Pi_1$}
\UnaryInfC{$(\epsilon,w\simp w);\Gamma[w/w'] \vdash x:A;\Delta[w/w']$}
\alwaysSingleLine
\RightLabel{$Eq_1$}
\UnaryInfC{$(\epsilon,w'\simp w);\Gamma \vdash x:A;\Delta$}
\alwaysNoLine
\AxiomC{$\Pi_2$}
\UnaryInfC{$\Gamma';x:A \vdash \Delta'$}
\alwaysSingleLine
\RightLabel{$cut$}
\BinaryInfC{$(\epsilon,w'\simp w);\Gamma;\Gamma' \vdash \Delta;\Delta'$}
\DisplayProof
\end{center}

This $cut$ is reduced in the same way as the $\top^* L$ case, where we get $\Pi_2'$ from the Substitution Lemma:
\begin{center}
\alwaysNoLine
\AxiomC{$\Pi_1$}
\UnaryInfC{$(\epsilon,w\simp w);\Gamma[w/w'] \vdash x:A;\Delta[w/w']$}
\AxiomC{$\Pi_2'$}
\UnaryInfC{$\Gamma'[w/w'];x:A \vdash \Delta'[w/w']$}
\alwaysSingleLine
\RightLabel{$cut$}
\BinaryInfC{$(\epsilon,w\simp w);\Gamma[w/w'];\Gamma'[w/w'] \vdash \Delta[w/w'];\Delta'[w/w']$}
\RightLabel{$Eq_1$}
\UnaryInfC{$(\epsilon,w'\simp w);\Gamma;\Gamma' \vdash \Delta;\Delta'$}
\DisplayProof
\end{center}

If $x = w'$, then we cut on $w:A$ instead in the reduced version.

\item If $r$ is $Eq_2$, the procedure follows similarly as the case for $Eq_1$ above.

\item If $r$ is a unary inference except for $\top^*L$, $Eq_1$, and $Eq_2$, then the original derivation is as follows.
\begin{center}
\alwaysNoLine
\AxiomC{$\Pi_1$}
\UnaryInfC{$\Gamma_1 \vdash x:A;\Delta_1$}
\alwaysSingleLine
\RightLabel{$r$}
\UnaryInfC{$\Gamma \vdash x:A;\Delta$}
\alwaysNoLine
\AxiomC{$\Pi_2$}
\UnaryInfC{$\Gamma';x:A\vdash \Delta'$}
\alwaysSingleLine
\RightLabel{$cut$}
\BinaryInfC{$\Gamma;\Gamma'\vdash\Delta;\Delta'$}
\DisplayProof
\end{center}

Then we can delay the application of $cut$ as follows.
\begin{center}
\alwaysNoLine
\AxiomC{$\Pi_1$}
\UnaryInfC{$\Gamma_1 \vdash x:A;\Delta_1$}
\AxiomC{$\Pi_2$}
\UnaryInfC{$\Gamma';x:A\vdash \Delta'$}
\alwaysSingleLine
\RightLabel{$cut$}
\BinaryInfC{$\Gamma_1;\Gamma'\vdash\Delta_1;\Delta'$}
\RightLabel{$r$}
\UnaryInfC{$\Gamma;\Gamma'\vdash\Delta;\Delta'$}
\DisplayProof
\end{center}

Note that as all our rules except $\top^* L$, $Eq_1$, and $Eq_2$ do
not modify side structures, $\Gamma'$ and $\Delta'$ in the premise of
$r$ are not changed. The complexity of the original $cut$ is
$(|x:A|,|\Pi_1|+1+|\Pi_2|)$, whereas the complexity of the new $cut$
is $(|x:A|,|\Pi_1|+|\Pi_2|)$, so the \textit{cut height} reduces.

\item If $r$ is a binary inference, we can transform the derivation similarly. 

\begin{center}
\alwaysNoLine
\AxiomC{$\Pi_1$}
\UnaryInfC{$\Gamma_1 \vdash x:A;\Delta_1$}
\AxiomC{$\Pi_2$}
\UnaryInfC{$\Gamma_2 \vdash x:A;\Delta_2$}
\alwaysSingleLine
\RightLabel{$r$}
\BinaryInfC{$\Gamma \vdash x:A;\Delta$}
\alwaysNoLine
\AxiomC{$\Pi_3$}
\UnaryInfC{$\Gamma';x:A\vdash \Delta'$}
\alwaysSingleLine
\RightLabel{$cut$}
\BinaryInfC{$\Gamma;\Gamma'\vdash\Delta;\Delta'$}
\DisplayProof $\leadsto$\\[15px]
\alwaysNoLine
\AxiomC{$\Pi_1$}
\UnaryInfC{$\Gamma_1 \vdash x:A;\Delta_1$}
\AxiomC{$\Pi_3$}
\UnaryInfC{$\Gamma';x:A\vdash \Delta'$}
\alwaysSingleLine
\RightLabel{$cut$}
\BinaryInfC{$\Gamma_1;\Gamma'\vdash\Delta_1;\Delta'$}
\alwaysNoLine
\AxiomC{$\Pi_2$}
\UnaryInfC{$\Gamma_2 \vdash x:A;\Delta_2$}
\AxiomC{$\Pi_3$}
\UnaryInfC{$\Gamma';x:A\vdash \Delta'$}
\alwaysSingleLine
\RightLabel{$cut$}
\BinaryInfC{$\Gamma_2;\Gamma'\vdash\Delta_2;\Delta'$}
\RightLabel{$r$}
\BinaryInfC{$\Gamma;\Gamma'\vdash\Delta;\Delta'$}
\DisplayProof
\end{center}

The complexity of the original $cut$ is $(|x:A|,max(|\Pi_1|,|\Pi_2|)+1+|\Pi_3|)$, and that of the new two $cut$s are $(|x:A|,|\Pi_1|+|\Pi_3|)$ and $(|x:A|,|\Pi_2|+|\Pi_3|)$ respectively. Thus the cut heights are reduced.

\end{enumerate}

\item The cut formula is only principal in the left premise.  We only consider the last rule in the right branch. The proof of this case is symmetric to those in Case 1.

\item The cut formula is principal in both premises. We do a case analysis on the main connective of the cut formula. If the main connective is additive, then there is no need to substitute any labels.

For $\land$,
\begin{center}
\alwaysNoLine
\AxiomC{$\Pi_1$}
\UnaryInfC{$\Gamma \vdash x:A;\Delta$}
\AxiomC{$\Pi_2$}
\UnaryInfC{$\Gamma \vdash x:B;\Delta$}
\alwaysSingleLine
\RightLabel{$\land R$}
\BinaryInfC{$\Gamma \vdash x:A\land B;\Delta$}
\alwaysNoLine
\AxiomC{$\Pi_3$}
\UnaryInfC{$\Gamma';x:A;x:B \vdash \Delta'$}
\alwaysSingleLine
\RightLabel{$\land L$}
\UnaryInfC{$\Gamma';x:A\land B \vdash \Delta'$}
\RightLabel{$cut$}
\BinaryInfC{$\Gamma;\Gamma' \vdash \Delta;\Delta'$}
\DisplayProof $\leadsto$\\[10px]
\alwaysNoLine
\AxiomC{$\Pi_1$}
\UnaryInfC{$\Gamma \vdash x:A;\Delta$}
\AxiomC{$\Pi_2$}
\UnaryInfC{$\Gamma \vdash x:B;\Delta$}
\AxiomC{$\Pi_3$}
\UnaryInfC{$\Gamma';x:A;x:B \vdash \Delta'$}
\alwaysSingleLine
\RightLabel{$cut$}
\BinaryInfC{$\Gamma;\Gamma';x:A \vdash \Delta;\Delta'$}
\RightLabel{$cut$}
\BinaryInfC{$\Gamma;\Gamma;\Gamma' \vdash \Delta;\Delta;\Delta'$}
\dashedLine
\RightLabel{Lemma~\ref{lm:ctr}}
\UnaryInfC{$\Gamma;\Gamma' \vdash \Delta;\Delta'$}
\DisplayProof
\end{center}

For $\limp$,
\begin{center}
\alwaysNoLine
\AxiomC{$\Pi_1$}
\UnaryInfC{$\Gamma';x:A \vdash x:B;\Delta'$}
\alwaysSingleLine
\RightLabel{$\limp R$}
\UnaryInfC{$\Gamma' \vdash x:A\limp B;\Delta'$}
\alwaysNoLine
\AxiomC{$\Pi_2$}
\UnaryInfC{$\Gamma \vdash x:A;\Delta$}
\AxiomC{$\Pi_3$}
\UnaryInfC{$\Gamma;x:B \vdash \Delta$}
\alwaysSingleLine
\RightLabel{$\limp L$}
\BinaryInfC{$\Gamma;x:A\limp B \vdash \Delta$}
\RightLabel{$cut$}
\BinaryInfC{$\Gamma;\Gamma' \vdash \Delta;\Delta'$}
\DisplayProof $\leadsto$\\[10px]
\alwaysNoLine
\AxiomC{$\Pi_2$}
\UnaryInfC{$\Gamma \vdash x:A;\Delta$}
\AxiomC{$\Pi_1$}
\UnaryInfC{$\Gamma';x:A \vdash x:B;\Delta'$}
\AxiomC{$\Pi_3$}
\UnaryInfC{$\Gamma;x:B \vdash \Delta$}
\alwaysSingleLine
\RightLabel{$cut$}
\BinaryInfC{$\Gamma;\Gamma';x:A \vdash \Delta;\Delta'$}
\RightLabel{$cut$}
\BinaryInfC{$\Gamma;\Gamma;\Gamma' \vdash \Delta;\Delta;\Delta'$}
\dashedLine
\RightLabel{Lemma~\ref{lm:ctr}}
\UnaryInfC{$\Gamma;\Gamma' \vdash \Delta;\Delta'$}
\DisplayProof
\end{center}

For both $\land$ and $\limp$, $cut$ is reduced to applications on smaller formulae, therefore the complexity of the $cut$ reduces.

There is an asymmetry in the rules for $\top^*$. That is, the left rule for $\top^*$ requires that the label $w$ of $\top^*$ cannot be $\epsilon$, whereas the right rule for $\top^*$ restricts the label of $\top^*$ to be $\epsilon$ only. As a consequence, when the cut formula is $\top^*$, it cannot be the principal formula of both premises at the same time. Therefore the cases for $\top^*$ are handled in the proof above.

When the main connective of the cut formula is $\mand$ or $\mimp$, the case is more complicated. For $\mand$, we have the following two derivations as the premises of the $cut$ rule:

\begin{center}
\small
\alwaysNoLine
\AxiomC{$\Pi_1$}
\UnaryInfC{$(x,y\simp z);\Gamma \vdash x:A;z:A\mand B;\Delta$}
\AxiomC{$\Pi_2$}
\UnaryInfC{$(x,y\simp z);\Gamma \vdash y:B;z:A\mand B;\Delta$}
\alwaysSingleLine
\RightLabel{\tiny $\mand R$}
\BinaryInfC{$(x,y\simp z);\Gamma \vdash z:A\mand B;\Delta$}
\DisplayProof
\end{center} \ \\
and
\begin{center}
\small
\alwaysNoLine
\AxiomC{$\Pi_3$}
\UnaryInfC{$(x',y'\simp z);\Gamma';x':A;y':B \vdash \Delta'$}
\alwaysSingleLine
\RightLabel{\tiny $\mand L$}
\UnaryInfC{$\Gamma';z:A\mand B \vdash \Delta'$}
\DisplayProof
\end{center}

And the $cut$ rule gives the end sequent $(x,y\simp z);\Gamma;\Gamma' \vdash \Delta;\Delta'$. The complexity of this $cut$ is $(|A\mand B|,max(|\Pi_1|,|\Pi_2|)+1+|\Pi_3|+1)$.

We use several $cut$s with less complexity to derive $(x,y\simp z);\Gamma;\Gamma' \vdash \Delta;\Delta'$ as follows.

Firstly, 
\begin{center}
\alwaysNoLine
\AxiomC{$\Pi_1$}
\UnaryInfC{$(x,y\simp z);\Gamma \vdash x:A;z:A\mand B;\Delta$}
\AxiomC{$\Pi_3$}
\UnaryInfC{$(x',y'\simp z);\Gamma';x':A;y':B \vdash \Delta'$}
\alwaysSingleLine
\RightLabel{$\mand L$}
\UnaryInfC{$\Gamma';z:A\mand B \vdash \Delta'$}
\RightLabel{$cut$}
\BinaryInfC{$(x,y\simp z);\Gamma;\Gamma' \vdash x:A;\Delta;\Delta'$}
\DisplayProof
\end{center}

The complexity of this $cut$ is $(|A\mand B|,|\Pi_1|+|\Pi_3|+1))$, thus is less than the original $cut$.

The second $cut$ works similarly.
\begin{center}
\alwaysNoLine
\AxiomC{$\Pi_2$}
\UnaryInfC{$(x,y\simp z);\Gamma \vdash y:B;z:A\mand B;\Delta$}
\AxiomC{$\Pi_3$}
\UnaryInfC{$(x',y'\simp z);\Gamma';x':A;y':B \vdash \Delta'$}
\alwaysSingleLine
\RightLabel{$\mand L$}
\UnaryInfC{$\Gamma';z:A\mand B \vdash \Delta'$}
\RightLabel{$cut$}
\BinaryInfC{$(x,y\simp z);\Gamma;\Gamma' \vdash y:B;\Delta;\Delta'$}
\DisplayProof
\end{center}

The third $cut$ works on a smaller formula.
\begin{center}
\AxiomC{$(x,y\simp z);\Gamma;\Gamma' \vdash x:A;\Delta;\Delta'$}
\alwaysNoLine
\AxiomC{$\Pi_3'$}
\UnaryInfC{$(x,y\simp z);\Gamma';x:A;y:B \vdash \Delta'$}
\alwaysSingleLine
\RightLabel{$cut$}
\BinaryInfC{$(x,y\simp z);(x,y\simp z);\Gamma;\Gamma';\Gamma';y:B \vdash \Delta;\Delta';\Delta'$}
\DisplayProof
\end{center}

The cut formula is $x:A$, thus the complexity of this $cut$ is less regardless of the height of the derivations.

Note that in the $\Pi_3$ branch, the $\mand L$ rule requires that the relation $(x',y'\simp z)$ is newly created, so $x'$ and $y'$ cannot be $\epsilon$ and they cannot be in $\Gamma'$ or $\Delta'$. Therefore we are allowed to use the substitution lemma to get a derivation $\Pi_3'$ of $(x,y\simp z);\Gamma';x:A;y:B \vdash \Delta'$ by just substituting $x'$ for $x$ and $y'$ for $y$.

Finally we cut on another smaller formula $y:B$.
\begin{center}
\small
\AxiomC{$(x,y\simp z);\Gamma;\Gamma' \vdash y:B;\Delta;\Delta'$}
\AxiomC{$(x,y\simp z);(x,y\simp z);\Gamma;\Gamma';\Gamma';y:B \vdash \Delta;\Delta';\Delta'$}
\RightLabel{$cut$}
\BinaryInfC{$(x,y\simp z);(x,y\simp z);(x,y\simp z);\Gamma;\Gamma;\Gamma';\Gamma';\Gamma' \vdash \Delta;\Delta;\Delta';\Delta';\Delta'$}
\DisplayProof
\end{center}

The complexity of this $cut$ is less than the original $cut$. We then apply the admissibility of contraction to derive $(x,y\simp z);\Gamma;\Gamma' \vdash \Delta;\Delta'$.

The case for $\mimp$ is similar. The two premises in the original $cut$ are as follows.
\begin{center}
\small
\alwaysNoLine
\AxiomC{$\Pi_1$}
\UnaryInfC{$(x',y\simp z');\Gamma';x':A \vdash z':B;\Delta'$}
\alwaysSingleLine
\RightLabel{\tiny $\mimp R$}
\UnaryInfC{$\Gamma' \vdash y:A\mimp B;\Delta'$}
\DisplayProof
\end{center} \ \\
and
\begin{center}
\small
\alwaysNoLine
\AxiomC{$\Pi_2$}
\UnaryInfC{$(x,y\simp z);\Gamma;y:A\mimp B \vdash x:A;\Delta$}
\AxiomC{$\Pi_3$}
\UnaryInfC{$(x,y\simp z);\Gamma;y:A\mimp B;z:B \vdash \Delta$}
\alwaysSingleLine
\RightLabel{\tiny $\mimp L$}
\BinaryInfC{$(x,y\simp z);\Gamma;y:A\mimp B \vdash \Delta$}
\DisplayProof
\end{center}

And the $cut$ rule yields the end sequent $(x,y\simp z);\Gamma;\Gamma' \vdash \Delta;\Delta'$. We use two cuts on the same formula, but with smaller derivation height.
\begin{center}
\small
\alwaysNoLine
\AxiomC{$\Pi_1$}
\UnaryInfC{$(x',y\simp z');\Gamma';x':A \vdash z':B;\Delta'$}
\alwaysSingleLine
\RightLabel{$\mimp R$}
\UnaryInfC{$\Gamma' \vdash y:A\mimp B;\Delta'$}
\alwaysNoLine
\AxiomC{$\Pi_2$}
\UnaryInfC{$(x,y\simp z);\Gamma;y:A\mimp B \vdash x:A;\Delta$}
\alwaysSingleLine
\RightLabel{$cut$}
\BinaryInfC{$(x,y\simp z);\Gamma;\Gamma' \vdash x:A;\Delta;\Delta'$}
\DisplayProof
\end{center}

\begin{center}
\small
\alwaysNoLine
\AxiomC{$\Pi_1$}
\UnaryInfC{$(x',y\simp z');\Gamma';x':A \vdash z':B;\Delta'$}
\alwaysSingleLine
\RightLabel{$\mimp R$}
\UnaryInfC{$\Gamma' \vdash y:A\mimp B;\Delta'$}
\alwaysNoLine
\AxiomC{$\Pi_3$}
\UnaryInfC{$(x,y\simp z);\Gamma;y:A\mimp B;z:B \vdash \Delta$}
\alwaysSingleLine
\RightLabel{$cut$}
\BinaryInfC{$(x,y\simp z);\Gamma;\Gamma';z:B \vdash \Delta;\Delta'$}
\DisplayProof
\end{center}

Then we cut on a smaller formula $x:A$.
\begin{center}
\alwaysNoLine
\AxiomC{$(x,y\simp z);\Gamma;\Gamma' \vdash x:A;\Delta;\Delta'$}
\AxiomC{$\Pi_1'$}
\UnaryInfC{$(x,y\simp z);\Gamma';x:A \vdash z:B;\Delta'$}
\alwaysSingleLine
\RightLabel{$cut$}
\BinaryInfC{$(x,y\simp z);(x,y\simp z);\Gamma;\Gamma';\Gamma' \vdash z:B;\Delta;\Delta';\Delta'$}
\DisplayProof
\end{center}

Again, in the original derivation, $x'$ and $z'$ are fresh in the premise of $\mimp R$ rule, thus by the Substitution Lemma we can have a derivation $\Pi_1'$ of the sequent $(x,y\simp z);\Gamma';x:A \vdash z:B;\Delta'$, with $x'$ substituted to $x$ and $z'$ substituted to $z$.

Then we cut on $z:B$.
\begin{center}
\small
\AxiomC{$(x,y\simp z);(x,y\simp z);\Gamma;\Gamma';\Gamma' \vdash z:B;\Delta;\Delta';\Delta'$}
\AxiomC{$(x,y\simp z);\Gamma;\Gamma';z:B \vdash \Delta;\Delta'$}
\RightLabel{$cut$}
\BinaryInfC{$(x,y\simp z);(x,y\simp z);(x,y\simp z);\Gamma;\Gamma;\Gamma';\Gamma';\Gamma' \vdash \Delta;\Delta;\Delta';\Delta';\Delta'$}
\DisplayProof
\end{center} 

In the end we use the theorem of admissibility of contraction to obtain the required sequent $(x,y\simp z);\Gamma;\Gamma' \vdash \Delta;\Delta'$.

\end{enumerate}
\qed
\end{proof}

\subsection{Permutation of structural rules in $\lsbbi$}
\label{app:str_permute}

Proof for Lemma~\ref{lem:str_permute}.
\begin{proof}
To prove this lemma, we need to show that if a derivation involves the
structural rules, we can always apply them exactly before $\mand R$
and $\mimp L$, or before zero-premise rules. We show this by an
induction on the height of the derivation. Since we do not permute
structural rules through zero-premise rules, the proof in the base
case and the inductive step are essentially the same. Here we give
some examples of the permutations. Assuming the lemma holds up to any
derivation of height $n-1$, consider a derivation of height $n$.

\begin{enumerate}
\item Permute the application of $Eq_1$ or $Eq_2$ through
  non-zero-premise logical rules except for $\mand R$ and $\mimp
  L$. Here we give some examples, the rest are similar.
\begin{enumerate}
\item Permute $Eq_2$ through additive logical rules is trivial, this
  is exemplified by $\land L$, assuming the label of the principal formula
  is modified by the $Eq_2$ application. The original derivation is as follows.
\begin{center}
\alwaysNoLine
\AxiomC{$\Pi$}
\UnaryInfC{$(\epsilon,\epsilon\simp
  \epsilon);\Gamma[\epsilon/w];\epsilon:A;\epsilon:B \vdash
  \Delta[\epsilon/w]$}
\alwaysSingleLine
\RightLabel{$\land L$}
\UnaryInfC{$(\epsilon,\epsilon\simp
  \epsilon);\Gamma[\epsilon/w];\epsilon:A\land B \vdash
  \Delta[\epsilon/w]$}
\RightLabel{$Eq_2$}
\UnaryInfC{$(\epsilon,\epsilon\simp
 w);\Gamma;w:A\land B \vdash
  \Delta$}
\DisplayProof
\end{center}

The derivation is changed to the following:
\begin{center}
\alwaysNoLine
\AxiomC{$\Pi$}
\UnaryInfC{$(\epsilon,\epsilon\simp
  \epsilon);\Gamma[\epsilon/w];\epsilon:A;\epsilon:B \vdash
  \Delta[\epsilon/w]$}
\alwaysSingleLine
\RightLabel{$Eq_2$}
\UnaryInfC{$(\epsilon,\epsilon\simp
 w);\Gamma;w:A; w:B \vdash
  \Delta$}
\RightLabel{$\land L$}
\UnaryInfC{$(\epsilon,\epsilon\simp
 w);\Gamma;w:A\land B \vdash
  \Delta$}
\DisplayProof
\end{center}

\item Permute $Eq_1$ through $\top^* L$, assuming the label of
  principal formula is $w$. The derivation is as follows.
\begin{center}
\alwaysNoLine
\AxiomC{$\Pi$}
\UnaryInfC{$(\epsilon,\epsilon\simp
  \epsilon);\Gamma[w/w'][\epsilon/w]\vdash \Delta[w/w'][\epsilon/w]$}
\alwaysSingleLine
\RightLabel{$\top^* L$}
\UnaryInfC{$(\epsilon,w\simp
 w);\Gamma[w/w'];w:\top^*\vdash \Delta[w/w']$}
\RightLabel{$Eq_1$}
\UnaryInfC{$(\epsilon,w'\simp
 w);\Gamma;w':\top^*\vdash \Delta$}
\DisplayProof
\end{center}

We modify the derivation as follows.
\begin{center}
\alwaysNoLine
\AxiomC{$\Pi$}
\UnaryInfC{$(\epsilon,\epsilon\simp
  \epsilon);\Gamma[\epsilon/w'][\epsilon/w]\vdash \Delta[\epsilon/w'][\epsilon/w]$}
\alwaysSingleLine
\RightLabel{$Eq_2$}
\UnaryInfC{$(\epsilon,\epsilon\simp
 w);\Gamma[\epsilon/w']\vdash \Delta[\epsilon/w']$}
\RightLabel{$\top^* L$}
\UnaryInfC{$(\epsilon,w'\simp
 w);\Gamma;w':\top^*\vdash \Delta$}
\DisplayProof
\end{center}

Notice that the premises of the two derivations below $\Pi$ are
exactly the same. The application of $Eq_1$ in the original derivation
is changed to an application of $Eq_2$ in the modified derivation. However,
this does not break the proof, as the induction hypothesis ensures that
either of them can be permuted upwards.

Also, the label of principal formula in the rule $\top^* L$ cannot be
the one that is replaced in the rule $Eq_2$ below it, this is the
reason we do not exemplify this situation using $Eq_2$.

\item Permute $Eq_2$ through $\mand L$, assuming the label of
  principal formula is $z$, and it is modified by the $Eq_2$ application.
\begin{center}
\alwaysNoLine
\AxiomC{$\Pi$}
\UnaryInfC{$(x,y\simp \epsilon);(\epsilon,\epsilon\simp \epsilon);\Gamma[\epsilon/z];x:A;y:B\vdash
  \Delta[\epsilon/z]$}
\alwaysSingleLine
\RightLabel{$\mand L$}
\UnaryInfC{$(\epsilon,\epsilon\simp \epsilon);\Gamma[\epsilon/z];\epsilon:A\mand B\vdash
  \Delta[\epsilon/z]$}
\RightLabel{$Eq_2$}
\UnaryInfC{$(\epsilon,\epsilon\simp z);\Gamma;z:A\mand B \vdash \Delta$}
\DisplayProof
\end{center}

Since $x$ and $y$ are fresh labels, they will not be affected by
$Eq_2$. Thus the derivation can be changed to the following:
\begin{center}
\alwaysNoLine
\AxiomC{$\Pi$}
\UnaryInfC{$(x,y\simp \epsilon);(\epsilon,\epsilon\simp \epsilon);\Gamma[\epsilon/z];x:A;y:B\vdash
  \Delta[\epsilon/z]$}
\alwaysSingleLine
\RightLabel{$Eq_2$}
\UnaryInfC{$(x,y\simp z);(\epsilon,\epsilon\simp z);\Gamma;x:A;y:B\vdash
  \Delta$}
\RightLabel{$\mand L$}
\UnaryInfC{$(\epsilon,\epsilon\simp y);\Gamma;z:A\mand B \vdash \Delta$}
\DisplayProof
\end{center}

Since $Eq_1$ and $Eq_2$ only globally replaces labels, their
action can be safely delayed through all the rules other than $\mand R$ and $\mimp L$. The applications of these two rules after the last $\mand R$ or $\mand L$ will be delayed until the zero-premise rule is necessary.
\end{enumerate}

\item Permute the applications of $E$, $U$, $A$, and $A_C$ through non-zero premise logical rules other than $\mand R$ and $\mimp L$. Again, we give some examples, the rest are similar.
\begin{enumerate} 
\item Permute $E$ through $\top^* L$, assuming the label of the principal
  formula is $y$.
The original derivation runs as follows.
\begin{center}
\alwaysNoLine
\AxiomC{$\Pi$}
\UnaryInfC{$(\epsilon,x\simp z);(x,\epsilon\simp z);\Gamma[\epsilon/y]
  \vdash \Delta[\epsilon/y]$}
\alwaysSingleLine
\RightLabel{$\top^* L$}
\UnaryInfC{$(y,x\simp z);(x,y\simp z);\Gamma;y:\top^*
  \vdash \Delta$}
\RightLabel{$E$}
\UnaryInfC{$(x,y\simp z);\Gamma;y:\top^*
  \vdash \Delta$}
\DisplayProof
\end{center}

The new derivation is as follows.
\begin{center}
\alwaysNoLine
\AxiomC{$\Pi$}
\UnaryInfC{$(\epsilon,x\simp z);(x,\epsilon\simp z);\Gamma[\epsilon/y]
  \vdash \Delta[\epsilon/y]$}
\alwaysSingleLine
\RightLabel{$E$}
\UnaryInfC{$(x,\epsilon\simp z);\Gamma[\epsilon/y]
  \vdash \Delta[\epsilon/y]$}
\RightLabel{$\top^* L$}
\UnaryInfC{$(x,y\simp z);\Gamma;y:\top^*
  \vdash \Delta$}
\DisplayProof
\end{center}

This shows that if the logical rule only does substitution, delaying
the application of structural rules makes no difference.

\item Permute $U$ through $\mand L$, assuming the label of the principal
  formula is $z$. The original derivation is as follows.
\begin{center}
\alwaysNoLine
\AxiomC{$\Pi$}
\UnaryInfC{$(x,y\simp z);(z,\epsilon\simp z);\Gamma;x:A;y:B \vdash
  \Delta$}
\alwaysSingleLine
\RightLabel{$\mand L$}
\UnaryInfC{$(z,\epsilon\simp z);\Gamma;z:A\mand B \vdash
  \Delta$}
\RightLabel{$U$}
\UnaryInfC{$\Gamma;z:A\mand B \vdash
  \Delta$}
\DisplayProof
\end{center}

The new derivation is as follows.
\begin{center}
\alwaysNoLine
\AxiomC{$\Pi$}
\UnaryInfC{$(z,\epsilon\simp z);(x,y\simp z);\Gamma;x:A;y:B \vdash
  \Delta$}
\alwaysSingleLine
\RightLabel{$U$}
\UnaryInfC{$(x,y\simp z);\Gamma;z:A\mand B \vdash
  \Delta$}
\RightLabel{$\mand L$}
\UnaryInfC{$\Gamma;z:A\mand B \vdash
  \Delta$}
\DisplayProof
\end{center}

Since the labels $x$ and $y$ are all fresh labels, it is safe to
change the order to rule applications as above.

Additive logical rules are totally independent on the relational
atoms, so those cases are similar as the one shown above, except that
those rules do not add relational atoms to the sequent.

\end{enumerate}
\end{enumerate}
\qed
\end{proof}

\subsection{Soundness of $\ilsbbi$}
\label{app:sound_ilsbbi}

\begin{theorem}
If there is a derivation $\Pi$ for a sequent $\Gamma\vdash \Delta$ in
$\ilsbbi$, then there is a derivation $\Pi'$ for the same sequent in
$\lsbbi$.
\end{theorem}

\begin{proof}
By induction on the height $n$ of $\Pi$.
\begin{enumerate}
\item Base case: $n = 1$. In this case the only rule must be a
  zero-premise rule. If the rule is $\bot L$ or $\top R$, then we can
  use the same rule in $\lsbbi$, since they are the same. Otherwise,
  suppose the rule is $id$, then $\Pi$ reads as follows.
\begin{center}
\AxiomC{$\mathcal{G}\vdash_E (w_1 = w_2)$}
\alwaysSingleLine
\RightLabel{\tiny $id$}
\UnaryInfC{$\Gamma;w_1:P \vdash w_2:P;\Delta$}
\DisplayProof
\end{center}

Since $\mathcal{G}\vdash_E (w_1 = w_2)$ is true, there is a sequence
$\sigma$ of $Eq_1,Eq_2$ applications such that $\Scal(\Gcal,\sigma)$ is
defined and $w_1\theta = w_2\theta$, where $\theta = subst(\sigma)$. Therefore we can construct $\Pi'$ are follows.
\begin{center}
\AxiomC{$$}
\alwaysSingleLine
\RightLabel{\tiny $id$}
\UnaryInfC{$\Gamma\theta;w_1\theta:P \vdash w_2\theta:P;\Delta\theta$}
\alwaysNoLine
\UnaryInfC{$\vdots \sigma$}
\UnaryInfC{$\Gamma;w_1:P \vdash w_2:P;\Delta$}
\DisplayProof
\end{center}

If the rule is $\top^* R$, $\Pi$ is:
\begin{center}
\AxiomC{$\mathcal{G}\vdash_E (w = \epsilon)$}
\RightLabel{\tiny $\top^* R$}
\UnaryInfC{$\Gamma \vdash w:\top^*;\Delta$}
\DisplayProof
\end{center}

We construct $\Pi'$ similarly, as $w\theta = \epsilon$ after the application of $\sigma$.
\begin{center}
\AxiomC{$$}
\RightLabel{\tiny $\top^* R$}
\UnaryInfC{$\Gamma\theta \vdash w\theta:\top^*;\Delta\theta$}
\alwaysNoLine
\UnaryInfC{$\vdots \sigma$}
\UnaryInfC{$\Gamma \vdash w:\top^*;\Delta$}
\DisplayProof
\end{center}

\item Inductive cases: suppose every sequent that is derivable in
  $\ilsbbi$ with height less than $n$ is also derivable in $\lsbbi$,
  consider a $\ilsbbi$ derivation of height $n$. We do a case analysis
  on the bottom rule in the derivation.
\begin{enumerate}
\item If the rule is $\land L$, $\land R$, $\limp L$, $\limp R$, $\mand L$, $\mand R$, $E$ or $U$, we can use the same rule in $\lsbbi$, since nothing is changed. 

\item If the rule is $\top^* L$, then $\Pi$ must be the following:
\begin{center}
\AxiomC{$\Pi_1$}
\alwaysNoLine
\UnaryInfC{$(\epsilon,w\simp \epsilon);\Gamma \vdash \Delta$}
\alwaysSingleLine
\RightLabel{\tiny $\top^* L$}
\UnaryInfC{$\Gamma;w:\top^* \vdash \Delta$}
\DisplayProof
\end{center}

By the induction hypothesis, $(\epsilon,w\simp \epsilon);\Gamma \vdash
\Delta$ is derivable in $\lsbbi$. Applying Lemma~\ref{subs}
(substitution for labels in $\lsbbi$) with
$[\epsilon/w]$, we obtain $(\epsilon,\epsilon\simp
\epsilon);\Gamma[\epsilon/w] \vdash \Delta[\epsilon/w]$. Thus we
construct $\Pi'$ as follows.
\begin{center}
\AxiomC{$\Pi_1'$}
\alwaysNoLine
\UnaryInfC{$(\epsilon,\epsilon\simp \epsilon);\Gamma[\epsilon/w] \vdash \Delta[\epsilon/w]$}
\alwaysSingleLine
\RightLabel{\tiny $\top^* L$}
\UnaryInfC{$(w,\epsilon\simp w);\Gamma;w:\top^* \vdash \Delta$}
\RightLabel{\tiny $U$}
\UnaryInfC{$\Gamma;w:\top^* \vdash \Delta$}
\DisplayProof
\end{center}

\item If the rule is $\mand R$, $\Pi$ runs as follows.
\begin{center}
\small
\AxiomC{$\Pi_1$}
\alwaysNoLine
\UnaryInfC{$(x,y \simp z');\Gamma \vdash x:A;z:A\mand B;\Delta$}
\AxiomC{$\Pi_2$}
\UnaryInfC{$(x,y \simp z');\Gamma \vdash y:B;z:A\mand B;\Delta$}
\alwaysSingleLine
\RightLabel{\tiny $\mand R$}
\BinaryInfC{$(x,y \simp z');\Gamma \vdash z:A\mand B;\Delta$}
\DisplayProof
\end{center}

The condition on the $\mand R$ rule is $\mathcal{G} \vdash_E (z = z')$. Let $\sigma$ be the sequence of $Eq_1,Eq_2$ applications such that
$\Scal(\Gcal,\sigma)$ is defined and, $z\theta = z'\theta$ holds,
where $\theta = subst(\sigma)$. Also, applying the
induction hypothesis on $\Pi_1$ and $\Pi_2$, we obtain the $\lsbbi$
derivations for each branch respectively. Then with the help of the
Substitution lemma, we get two derivations as follows. Note that
we use dashed lines when applying the Substitution lemmas.
\begin{center}
\small
\AxiomC{$\Pi_1'$}
\alwaysNoLine
\UnaryInfC{$(x,y \simp z');\Gamma \vdash x:A;z:A\mand B;\Delta$}
\alwaysSingleLine
\dashedLine
\RightLabel{\tiny Lemma~\ref{subs}}
\UnaryInfC{$(x\theta,y\theta \simp z'\theta);\Gamma\theta \vdash x\theta:A;z\theta:A\mand B;\Delta\theta$}
\DisplayProof
\end{center} \ \\
and
\begin{center}
\small
\AxiomC{$\Pi_2'$}
\alwaysNoLine
\UnaryInfC{$(x,y \simp z');\Gamma \vdash y:B;z:A\mand B;\Delta$}
\alwaysSingleLine
\dashedLine
\RightLabel{\tiny Lemma~\ref{subs}}
\UnaryInfC{$(x\theta,y\theta \simp z'\theta);\Gamma\theta \vdash y\theta:B;z\theta:A\mand B;\Delta\theta$}
\DisplayProof
\end{center}

Then we can apply $\mand R$ and obtain $(x\theta,y\theta \simp z'\theta);\Gamma \vdash z\theta:A\mand B;\Delta\theta$. Then by applying $\sigma$ we obtain the end sequent as follows.
\begin{center}
\small
\AxiomC{$(x\theta,y\theta \simp z'\theta);\Gamma \vdash z\theta:A\mand B;\Delta\theta$}
\alwaysNoLine
\UnaryInfC{$\vdots \sigma$}
\UnaryInfC{$(x,y \simp z');\Gamma \vdash z:A\mand B;\Delta$}
\DisplayProof
\end{center}

The case for $\mimp L$ is treated similarly.

\item If the rule is $A$, the treatment for the equality entailment is the same. $\Pi$ is in the following form:
\begin{center}
\AxiomC{$\Pi_1$}
\alwaysNoLine
\UnaryInfC{$(u,w \simp z);(y,v \simp w);(x,y \simp z);(u,v \simp x');\Gamma \vdash \Delta$}
\AxiomC{$\mathcal{G}\vdash_E (x = x')$}
\alwaysSingleLine
\RightLabel{\tiny $A$}
\BinaryInfC{$(x,y \simp z);(u,v \simp x');\Gamma \vdash \Delta$}
\DisplayProof
\end{center}

Let $\Scal(\Gcal,\sigma)$ yield $x\theta = x'\theta$, where $\theta=subst(\sigma)$, we obtain $\Pi'$ as follows.
\begin{center}
\AxiomC{$\Pi_1'$}
\alwaysNoLine
\UnaryInfC{$(u,w \simp z);(y,v \simp w);(x,y \simp z);(u,v \simp x');\Gamma \vdash \Delta$}
\alwaysSingleLine \dashedLine
\RightLabel{\tiny Lemma~\ref{subs}}
\UnaryInfC{$(u\theta,w\theta \simp z\theta);(y\theta,v\theta \simp w\theta);(x\theta,y\theta \simp z\theta);(u\theta,v\theta \simp x'\theta);\Gamma\theta \vdash \Delta\theta$}
\alwaysSingleLine
\RightLabel{\tiny $A$}
\UnaryInfC{$(x\theta,y\theta \simp z\theta);(u\theta,v\theta \simp x'\theta);\Gamma\theta \vdash \Delta\theta$}
\alwaysNoLine
\UnaryInfC{$\vdots \sigma$}
\UnaryInfC{$(x,y \simp z);(u,v \simp x');\Gamma \vdash \Delta$}
\DisplayProof
\end{center}

The case for $A_C$ is similar.

\end{enumerate}
\end{enumerate}
\qed
\end{proof}

\subsection{Completeness of $\ilsbbi$}
\label{app:comp_ilsbbi}

To prove the completeness of $\ilsbbi$, firstly we add $Eq_1$ and $Eq_2$ in $\ilsbbi$ and show that the resultant system has the same power as $\lsbbi$. Then we prove the admissibility of $Eq_1$ and $Eq_2$ in $\ilsbbi$.

\begin{lemma}
\label{lem:comp_ilsbbi_eq12}
If a sequent $\Gamma\vdash \Delta$ is derivable in $\lsbbi$, then it is derivable in $\ilsbbi + Eq_1 + Eq_2$.  
\end{lemma}

\begin{proof}
By induction on the height of the $\lsbbi$ derivation. Since with $Eq_1$ and $Eq_2$, most of other rules become identical, the only non-trivial case is $\top^* L$.

In $\lsbbi$, the derivation runs as follows.
\begin{center}
\AxiomC{$\Pi$}
\alwaysNoLine
\UnaryInfC{$\Gamma[\epsilon/w] \vdash \Delta[\epsilon/w]$}
\alwaysSingleLine
\RightLabel{\tiny $\top^* L$}
\UnaryInfC{$\Gamma;w:\top^* \vdash \Delta$}
\DisplayProof
\end{center}

By the induction hypothesis, there is a derivation for $\Gamma[\epsilon/w] \vdash \Delta[\epsilon/w]$ in $\ilsbbi + Eq_1 + Eq_2$. Therefore we construct the derivation as follows.
\begin{center}
\AxiomC{$\Pi'$}
\alwaysNoLine
\UnaryInfC{$\Gamma[\epsilon/w] \vdash \Delta[\epsilon/w]$}
\alwaysSingleLine \dashedLine
\RightLabel{\tiny Lemma~\ref{lm:weak}}
\UnaryInfC{$(\epsilon,\epsilon\simp\epsilon);\Gamma[\epsilon/w] \vdash \Delta[\epsilon/w]$}
\RightLabel{\tiny$Eq_1$}
\UnaryInfC{$(\epsilon,w\simp\epsilon);\Gamma \vdash \Delta$}
\RightLabel{\tiny $\top^* L$}
\UnaryInfC{$\Gamma;w:\top^* \vdash \Delta$}
\DisplayProof
\end{center}
\qed
\end{proof}

\begin{lemma}
\label{lem:subsg_g}
If $\mathcal{G}[x/y];(\epsilon,x\simp x) \vdash_E (w_1[x/y] = w_2[x/y])$ then $\mathcal{G};(\epsilon,y\simp x)\vdash_E (w_1 = w_2)$.
\end{lemma}

\begin{proof}
Let $\mathcal{G}' = \mathcal{G};(\epsilon,y\simp x)$ and
$\Scal(\mathcal{G}'[x/y],\sigma)$ yield $(w_1[x/y]\theta = w_2[x/y]\theta)$, we show that $\mathcal{G}'\vdash_E (x = y)$ by following:
\begin{center}
\AxiomC{$\mathcal{G}'[x/y]\theta\vdash_E (w_1[x/y]\theta = w_2[x/y]\theta)$}
\alwaysNoLine
\UnaryInfC{$\vdots \sigma$}
\UnaryInfC{$\mathcal{G}'[x/y]\vdash_E (w_1[x/y] = w_2[x/y])$}
\alwaysSingleLine
\RightLabel{\tiny $Eq_1$}
\UnaryInfC{$\mathcal{G};(\epsilon,y\simp x)\vdash_E (x = y)$}
\DisplayProof
\end{center}
\qed
\end{proof}

Now we show that $Eq_1$ is admissible in $\ilsbbi$.

\begin{lemma}
\label{lem:eq1_adms}
If $(\epsilon,x\simp x);\Gamma[x/y] \vdash \Delta[x/y]$ is derivable in $\ilsbbi$, then $(\epsilon,y\simp x);\Gamma\vdash \Delta$ is derivable in $\ilsbbi$.
\end{lemma}

\begin{proof}
We show that $Eq_1$ can always permute up through all other rules, and
eventually disappear when it hits the zero-premise rule. Since
Lemma~\ref{lem:str_permute} is sufficient to show the permutations
through nagative rules, here we particularly show the cases for
positive rules.
\begin{enumerate}
\item First let us show the cases for the zero-premise rules. $\bot L$ and $\top R$ are trivial, as they are applicable for an arbitrary label. The permutation for $id$ runs as follows, where $\mathcal{G}$ is the set of relational atoms in $(\epsilon,y\simp x);\Gamma$.
\begin{center}
\AxiomC{$\mathcal{G}[x/y]\vdash_E (w_1[x/y] = w_2[x/y])$}
\RightLabel{\tiny $id$}
\UnaryInfC{$(\epsilon,x\simp x);\Gamma[x/y];w_1[x/y]:P \vdash w_2[x/y]:P;\Delta$}
\RightLabel{\tiny $Eq_1$}
\UnaryInfC{$(\epsilon,y\simp x);\Gamma;w_1:P \vdash w_2:P;\Delta$}
\DisplayProof
\end{center} 

By Lemma~\ref{lem:subsg_g}, if $\mathcal{G}[x/y]\vdash_E (w_1[x/y] = w_2[x/y])$ then $\mathcal{G}\vdash_E (w_1 = w_2)$ (note that this is because $(\epsilon,y\simp x)\in \mathcal{G}$). Therefore we can apply $id$ directly on the bottom sequent, and eliminate the $Eq_1$ application.

The case for $\top^* R$ is treated similarly. As we have shown, structural rules can permute through $\top^* L$, $\land L$, $\land R$, $\limp L$, $\limp R$, $\mand L$ and $\mimp R$, so these cases are left out here.

\item Permute $Eq_1$ through $E$, assuming the label being replaced is
  $y$. The original derivation is as follows.
\begin{center}
\alwaysNoLine 
\AxiomC{$\Pi$}
\UnaryInfC{$(w,x,\simp z);(x,w\simp z);(\epsilon,w\simp w);\Gamma[w/y]
  \vdash \Delta[w/y]$}
\alwaysSingleLine
\RightLabel{\tiny $E$}
\UnaryInfC{$(x,w\simp z);(\epsilon,w\simp w);\Gamma[w/y]
  \vdash \Delta[w/y]$}
\RightLabel{\tiny $Eq_1$}
\UnaryInfC{$(x,y\simp z);(\epsilon,y\simp w);\Gamma
  \vdash \Delta$}
\DisplayProof
\end{center}

The permuted derivation is as follows.
\begin{center}
\alwaysNoLine 
\AxiomC{$\Pi$}
\UnaryInfC{$(w,x,\simp z);(x,w\simp z);(\epsilon,w\simp w);\Gamma[w/y]
  \vdash \Delta[w/y]$}
\alwaysSingleLine
\RightLabel{\tiny $Eq_1$}
\UnaryInfC{$(y,x\simp z);(x,y\simp z);(\epsilon,y\simp w);\Gamma
  \vdash \Delta$}
\RightLabel{\tiny $E$}
\UnaryInfC{$(x,y\simp z);(\epsilon,y\simp w);\Gamma
  \vdash \Delta$}
\DisplayProof
\end{center}

\item Premute $Eq_1$ through $U$, assuming the replaced label is
  $x$. Then the derivation runs as follows.
\begin{center}
\alwaysNoLine
\AxiomC{$\Pi$}
\UnaryInfC{$(w,\epsilon \simp w);(\epsilon,w\simp w);\Gamma[w/x]
  \vdash \Delta[w/x]$}
\alwaysSingleLine
\RightLabel{\tiny $U$}
\UnaryInfC{$(\epsilon,w\simp w);\Gamma[w/x]
  \vdash \Delta[w/x]$}
\RightLabel{\tiny $Eq_1$}
\UnaryInfC{$(\epsilon,x\simp w);\Gamma
  \vdash \Delta$}
\DisplayProof
\end{center}

We modify the derivation as follows.
\begin{center}
\alwaysNoLine
\AxiomC{$\Pi$}
\UnaryInfC{$(w,\epsilon \simp w);(\epsilon,w\simp w);\Gamma[w/x]
  \vdash \Delta[w/x]$}
\alwaysSingleLine
\RightLabel{\tiny $Eq_1$}
\UnaryInfC{$(x,\epsilon\simp x);(\epsilon,x\simp w);\Gamma
  \vdash \Delta$}
\RightLabel{\tiny $U$}
\UnaryInfC{$(\epsilon,x\simp w);\Gamma
  \vdash \Delta$}
\DisplayProof
\end{center}

Note that we can also generate $(w,\epsilon\simp w)$ directly using
the $U$ rule, but the effect is the same.

\item Permute $Eq_1$ through $\mand R$. Suppose the principal relational atom of $Eq_1$ is not the same as the one used in $\mand R$, let $\mathcal{G}$ be the set of relational atoms in $(\epsilon,w\simp w')(x,y \simp z');\Gamma$, the derivation runs as follows. Here we write $(\Gamma\vdash \Delta)[x/y]$ to mean that replace every $y$ by $x$ in the entire sequent. The equality entailment is $\mathcal{G}[w'/w] \vdash_E (z[w'/w] = z'[w'/w])$ (to save space, we do not write the constraint in the derivation). 
\begin{center}
\AxiomC{}
\RightLabel{\tiny $\mand R$}
\UnaryInfC{$((\epsilon,w'\simp w')(x,y \simp z');\Gamma \vdash z:A\mand B;\Delta)[w'/w]$}
\RightLabel{\tiny $Eq_1$}
\UnaryInfC{$(\epsilon,w\simp w')(x,y \simp z');\Gamma \vdash z:A\mand B;\Delta$}
\DisplayProof
\end{center}

The two premises of the $\mand R$ rule application are listed below.
\begin{center}
$((\epsilon,w'\simp w');(;x,y \simp z');\Gamma \vdash x:A;z:A\mand B;\Delta)[w'/w]$
$((\epsilon,w'\simp w');(x,y \simp z');\Gamma \vdash y:B;z:A\mand B;\Delta)[w'/w]$
\end{center}

By Lemma~\ref{lem:subsg_g}, since $\mathcal{G}[w'/w] \vdash_E (z[w'/w] = z'[w'/w])$, and $(\epsilon,w\simp w')\in\mathcal{G}$, $\mathcal{G}\vdash_E (z = z')$ holds. Therefore we have the following two derivations:
\begin{center}
\AxiomC{$((\epsilon,w'\simp w');(;x,y \simp z');\Gamma \vdash x:A;z:A\mand B;\Delta)[w'/w]$}
\RightLabel{\tiny $Eq_1$}
\UnaryInfC{$(\epsilon,w\simp w');(;x,y \simp z');\Gamma \vdash x:A;z:A\mand B;\Delta$}
\DisplayProof
\end{center}

and 
\begin{center}
\AxiomC{$((\epsilon,w'\simp w');(x,y \simp z');\Gamma \vdash y:B;z:A\mand B;\Delta)[w'/w]$}
\RightLabel{\tiny $Eq_1$}
\UnaryInfC{$(\epsilon,w\simp w');(x,y \simp z');\Gamma \vdash y:B;z:A\mand B;\Delta$}
\DisplayProof
\end{center}

then we use the $\mand R$ rule, where the equality entailment is $\mathcal{G} \vdash_E (z = z')$, to obtain the end sequent $(\epsilon,w\simp w')(x,y \simp z');\Gamma \vdash z:A\mand B;\Delta$.

If the principal relational atom is used in the $\mand R$ rule, the permutation is analogous. The permutation through $\mimp L$ is similar.

\item Permutation through $A$. We show the case where the principal relational atom in $Eq_1$ is not in $A$, the other cases are similar. The original derivation is as follows.
\begin{center}
\footnotesize
\AxiomC{$((\epsilon,w\simp w);(u,w \simp z);(y,v \simp w);(x,y \simp z);(u,v \simp x');\Gamma \vdash \Delta)[w/w']$}
\RightLabel{\tiny $A$}
\UnaryInfC{$((\epsilon,w\simp w);(x,y \simp z);(u,v \simp x');\Gamma \vdash \Delta)[w/w']$}
\RightLabel{\tiny $Eq_1$}
\UnaryInfC{$(\epsilon,w'\simp w);(x,y \simp z);(u,v \simp x');\Gamma \vdash \Delta$}
\DisplayProof
\end{center}

The condition on the $A$ rule is $\mathcal{G}[w/w']\vdash_E
  (x[w/w'] = x'[w/w'])$. By Lemma~\ref{lem:subsg_g}, $\mathcal{G}\vdash_E (x = x')$ holds. Therefore the derivation is transformed into the following:
\begin{center}
\AxiomC{$((\epsilon,w'\simp w);(u,w \simp z);(y,v \simp w);(x,y \simp z);(u,v \simp x');\Gamma \vdash \Delta)[w/w']$}
\RightLabel{\tiny $Eq_1$}
\UnaryInfC{$(\epsilon,w'\simp w);(u,w \simp z);(y,v \simp w);(x,y \simp z);(u,v \simp x');\Gamma \vdash \Delta$}
\RightLabel{\tiny $A$}
\UnaryInfC{$(\epsilon,w'\simp w);(x,y \simp z);(u,v \simp x');\Gamma \vdash \Delta$}
\DisplayProof
\end{center}
The condition on the $A$ rule is $\mathcal{G}\vdash_E (x = x')$.
$A_C$ is treated similarly.
\qed
\end{enumerate}
\end{proof}

\begin{lemma}
\label{lem:eq2_adms}
If $(\epsilon,y\simp y);\Gamma[y/x] \vdash \Delta[y/x]$ is derivable in $\ilsbbi$, then $(\epsilon,y\simp x);\Gamma\vdash \Delta$ is derivable in $\ilsbbi$.
\end{lemma}

\begin{proof}
Symmetric to the proof in Lemma~\ref{lem:eq1_adms}.\qed
\end{proof}

\begin{theorem}
If a sequent is derivable in $\lsbbi$, then it is also derivable in $\ilsbbi$.
\end{theorem}

\begin{proof}
Immediate by Lemma~\ref{lem:comp_ilsbbi_eq12},~\ref{lem:eq1_adms},~\ref{lem:eq2_adms}.\qed
\end{proof}

\subsection{Substitution lemma for $\ilsbbi$}
\label{app:subs_ilsbbi}

This section proves the substitution lemma for the intermediate system
$\ilsbbi$, as this will be used in some proofs.

\begin{lemma}
\label{lem:eqe_subs}
If $\mathcal{G}\vdash_E (x = y)$ then for any substitution $[s/t]$, where $t \not = \epsilon$, $\mathcal{G}[s/t]\vdash_E (x[s/t] = y[s/t])$. 
\end{lemma}

\begin{proof}
Let $(\mathcal{G},\sigma,\phi)$ be the solution to $\mathcal{G}\vdash_E (x = y)$, we prove this lemma by induction on the length of $\sigma$.

\begin{enumerate}
\item Base case, $\sigma$ is an empty sequence. In this case, the sequence of substitutions $\phi$ is also empty, therefore $x = y$. As a result, it must be the case that $x[s/t] = y[s/t]$, so $\mathcal{G}[s/t]\vdash_E (x[s/t] = y[s/t])$ trivially holds.

\item Inductive case, assume $|\sigma| = n$. Let us look at the first rule application in $\sigma$. Assume this rule is $Eq_1$ (the case for $Eq_2$ is symmetric), and the principal relational atom is $(\epsilon,u\simp v)$, then $\sigma$ is as follows.
\begin{center}
\AxiomC{$\mathcal{G}\phi\vdash_E (x\phi = y\phi)$}
\alwaysNoLine
\UnaryInfC{$\vdots \sigma'$}
\UnaryInfC{$\mathcal{G}'[v/u];(\epsilon,v\simp v)\vdash_E (x[v/u] = y[v/u])$}
\alwaysSingleLine
\RightLabel{\tiny $Eq_1$}
\UnaryInfC{$\mathcal{G}';(\epsilon,u\simp v)\vdash_E (x = y)$}
\DisplayProof
\end{center}

\begin{enumerate}
\item If $u = t$ and $v = s$, then the premise of the last rule application is already what we need.
\item If $u = t$ and $v \not = s$, we obtain the desired entailment as follows ($IH[x/y]$ stands for applying the induction hypothesis with the substitution $[x/y]$, we use double line to mean that the premise and the conclusion are equivalent).
\begin{center}
\AxiomC{$IH[v/s]$}
\alwaysNoLine
\UnaryInfC{$\mathcal{G}'[v/u][v/s];(\epsilon,v\simp v)\vdash_E (x[v/u][v/s] = y[v/u][v/s])$}
\alwaysSingleLine
\doubleLine
\UnaryInfC{$\mathcal{G}'[s/u][v/s];(\epsilon,v\simp v)\vdash_E (x[s/u][v/s] = y[s/u][v/s])$}
\RightLabel{\tiny $Eq_1$}
\UnaryInfC{$\mathcal{G}'[s/u];(\epsilon,s\simp v)\vdash_E (x[s/u] = y[s/u])$}
\DisplayProof
\end{center}

\item If $u = s$, we prove the substituted entailment as follows.
\begin{center}
\AxiomC{$IH[v/t]$}
\alwaysNoLine
\UnaryInfC{$\mathcal{G}'[v/u][v/t];(\epsilon,v\simp v)\vdash_E (x[v/u][v/t] = y[v/u][v/t])$}
\alwaysSingleLine
\doubleLine
\UnaryInfC{$\mathcal{G}'[u/t][v/u];(\epsilon,v\simp v)\vdash_E (x[u/t][v/u] = y[u/t][v/u])$}
\RightLabel{\tiny $Eq_1$}
\UnaryInfC{$\mathcal{G}'[u/t];(\epsilon,u\simp v)\vdash_E (x[u/t] = y[u/t])$}
\DisplayProof
\end{center}

Note that under this case if $v = t$, the proof is just a special case of the one above. 

\item If $v = t$, the case is shown below.
\begin{center}
\AxiomC{$IH[s/v]$}
\alwaysNoLine
\UnaryInfC{$\mathcal{G}'[v/u][s/v];(\epsilon,s\simp s)\vdash_E (x[v/u][s/v] = y[v/u][s/v])$}
\alwaysSingleLine
\doubleLine
\UnaryInfC{$\mathcal{G}'[s/v][s/u];(\epsilon,s\simp s)\vdash_E (x[s/v][s/u] = y[s/v][s/u])$}
\RightLabel{\tiny $Eq_1$}
\UnaryInfC{$\mathcal{G}'[s/v];(\epsilon,u\simp s)\vdash_E (x[s/v] = y[s/v])$}
\DisplayProof
\end{center}

\item If $v = s$, the proof is as follows.
\begin{center}
\AxiomC{$IH[v/t]$}
\alwaysNoLine
\UnaryInfC{$\mathcal{G}'[v/u][v/t];(\epsilon,v\simp v)\vdash_E (x[v/u][v/t] = y[v/u][v/t])$}
\alwaysSingleLine
\doubleLine
\UnaryInfC{$\mathcal{G}'[v/t][v/u];(\epsilon,v\simp v)\vdash_E (x[v/t][v/u] = y[v/t][v/u])$}
\RightLabel{\tiny $Eq_1$}
\UnaryInfC{$\mathcal{G}'[v/t];(\epsilon,u\simp v)\vdash_E (x[v/t] = y[v/t])$}
\DisplayProof
\end{center}

\item If $[s/t]$ and $[u/v]$ are independent, then we can switch the order of substitution, and derive the entailment as follows. 
\begin{center}
\AxiomC{$IH[s/t]$}
\alwaysNoLine
\UnaryInfC{$\mathcal{G}'[v/u][s/t];(\epsilon,v\simp v)\vdash_E (x[v/u][s/t] = y[v/u][s/t])$}
\alwaysSingleLine
\doubleLine
\UnaryInfC{$\mathcal{G}'[s/t][v/u];(\epsilon,v\simp v)\vdash_E (x[s/t][v/u] = y[s/t][v/u])$}
\RightLabel{\tiny $Eq_1$}
\UnaryInfC{$\mathcal{G}'[s/t];(\epsilon,u\simp v)\vdash_E (x[s/t] = y[s/t])$}
\DisplayProof
\end{center}

\end{enumerate}

\end{enumerate}
\qed 
\end{proof}
 
Since substitution does not break the equality entailment, we can show a substitution lemma for the system $\ilsbbi$.

\begin{lemma}[Substitution in $\ilsbbi$]
\label{lem:subs_ilsbbi}
If there is a derivation for the sequent $\Gamma \vdash \Delta$ in $\ilsbbi$ then there is a derivation of the same height for the sequent $\Gamma[y/x] \vdash \Delta[y/x]$ in $\ilsbbi$, where every occurrence of label $x$ ($x \not = \epsilon$) is replaced by label $y$.
\end{lemma}

\begin{proof}
The proof is basically the same as the one for $\lsbbi$, since there are a lot of common rules. For the rules that are changed, the case for $\top^* L$ is similar to those cases for additive rules. The proof for the rest of changed rules are straightforward with the help of Lemma~\ref{lem:eqe_subs}.\qed
\end{proof}

\subsection{Soundness of $\iilsbbi$}
\label{app:sound_iilsbbi}


\begin{theorem}
If there is a derivation $\Pi$ for a sequent $\Gcal||\Gamma\vdash \Delta$ in
$\iilsbbi$, then there is a derivation $\Pi'$ for the sequent
$\Gcal;\Gamma\vdash \Delta$ in
$\ilsbbi$.
\end{theorem}

\begin{proof}
The soundness proof for this system is rather straightforward. To prove this, we show that each rule in $\iilsbbi$ can be simulated in $\ilsbbi$. To do this, one just need to unfold the structural rule applications into the derivation. For instance, we can simulate the $id$ rule in $\iilsbbi$ by using the following rules in $\ilsbbi$:
\begin{center}
\AxiomC{$\mathcal{S}(\mathcal{G},\sigma)\vdash_E (w_1 = w_2)$}
\RightLabel{\tiny $id$}
\alwaysSingleLine
\UnaryInfC{$\mathcal{S}(\mathcal{G},\sigma);\Gamma;w_1:P \vdash w_2:P;\Delta$}
\alwaysNoLine
\UnaryInfC{$\vdots \sigma$}
\UnaryInfC{$\mathcal{G};\Gamma;w_1:P \vdash w_2:P;\Delta$}
\DisplayProof
\end{center}

The above works because the $id$ rule in $\iilsbbi$ requires
$\mathcal{G}\vdash_R (w_1 = w_2)$, which by definition ensures that
$\mathcal{S}(\mathcal{G},\sigma)\vdash_E (w_1 = w_2)$ holds. The case
for $\top^* R$ works similarly. One thing to notice is that structural
rules only add relational atoms into the current set, so except for
$\mathcal{G}$ is becoming a bigger set, all the other structures in
the sequent remain the same after the sequence $\sigma$ of
applications. Let us examine the simulation of $\mand R$ in $\ilsbbi$. 

\begin{center}
\footnotesize
\AxiomC{$\mathcal{S}(\mathcal{G},\sigma);\Gamma \vdash x':A;w:A\mand B;\Delta$}
\AxiomC{$\mathcal{S}(\mathcal{G},\sigma);\Gamma \vdash y':B;w:A\mand B;\Delta$}
\RightLabel{\tiny $\mand R$}
\BinaryInfC{$\mathcal{S}(\mathcal{G},\sigma);\Gamma \vdash w:A\mand B;\Delta$}
\alwaysNoLine
\UnaryInfC{$\vdots \sigma$}
\UnaryInfC{$\mathcal{G};\Gamma \vdash w:A\mand B;\Delta$}
\DisplayProof
\end{center}

The condition of the $\mand R$ rule is $\mathcal{S}(\mathcal{G},\sigma) \vdash_E (w = w')$. Since the $\iilsbbi$ rule requires $\mathcal{G} \vdash_R (x,y\simp w)$, which by definition ensures that there is a solution $(\mathcal{G},\sigma)$ such that $(x',y'\simp w')\in \mathcal{S}(\mathcal{G},\sigma)$, and the following holds:
\begin{center}
$\mathcal{S}(\mathcal{G},\sigma) \vdash_E (x = x')$\\
$\mathcal{S}(\mathcal{G},\sigma) \vdash_E (y = y')$\\
$\mathcal{S}(\mathcal{G},\sigma) \vdash_E (w = w')$\\
\end{center}

The last relation entailment is enough to guarantee that the $\mand R$
rule is applicable. To restore each branch, we need the
Lemma~\ref{lem:subs_ilsbbi} (Substitution lemma for $\ilsbbi$). Here
we use double line to indicate the premise and the conclusion are
equivalent. Let us look at the left branch. By the first relation
entailment, there is a sequence $\sigma'$ of $Eq_1,Eq_2$ applications so that $x\theta = x'\theta$. Therefore we can construct a proof for the left branch as follows.
\begin{center}
\AxiomC{$\mathcal{S}(\mathcal{G},\sigma);\Gamma \vdash x:A;w:A\mand B;\Delta$}
\dashedLine
\RightLabel{\tiny Lemma~\ref{lem:subs_ilsbbi}}
\UnaryInfC{$\mathcal{S}(\mathcal{G},\sigma)\theta;\Gamma\theta \vdash x\theta:A;w\theta:A\mand B;\Delta\theta$}
\doubleLine
\UnaryInfC{$\mathcal{S}(\mathcal{G},\sigma)\theta;\Gamma\theta \vdash x'\theta:A;w\theta:A\mand B;\Delta\theta$}
\alwaysNoLine
\UnaryInfC{$\vdots \sigma'$}
\UnaryInfC{$\mathcal{S}(\mathcal{G},\sigma);\Gamma \vdash x':A;w:A\mand B;\Delta$}
\DisplayProof
\end{center}

The case for $\mimp L$ is analogous. The rest rules are the same as in
$\ilsbbi$, thus we conclude that the rules in $\iilsbbi$ are
sound.\qed
\end{proof}

\subsection{Completeness of $\iilsbbi$}
\label{app:comp_iilsbbi}


The completeness proof runs the same as in $\ilsbbi$: if we add the structural rules $E$, $U$, $A$, $A_C$ in $\iilsbbi$, then it becomes a superset of $\ilsbbi$. Then we prove that these rules are admissible in $\iilsbbi$ by showing they can permute through $\mand R$, $\mimp L$, $id$, and $\top^* R$.

First of all, let us show that when we add $E$, $U$, $A$, $A_C$ (from $\ilsbbi$) to $\iilsbbi$, its rules can simulate those ones in $\ilsbbi$. As most of the rules are identical, the key part is the show the relation entailment is as powerful as the equality entailment. This is ``built-in'' the definition, so there is no surprise.

\begin{lemma}
If $\mathcal{G}\vdash_E (w_1 = w_2)$, then $\mathcal{G}\vdash_R (w_1 = w_2)$.
\end{lemma}

\begin{proof}
Let $\sigma$ be an empty list of rule applications, then $\mathcal{S}(\mathcal{G},\emptyset) = \mathcal{G}$. Therefore by definition $\mathcal{G}\vdash_R (w_1 = w_2)$.\qed  
\end{proof}

If we change $\vdash_R$ to $\vdash_E$ in $\iilsbbi$, every rule is the same as the one in $\ilsbbi$. Therefore $\iilsbbi + E +U + A + A_C$ is at least as powerful as $\ilsbbi$.

\begin{lemma}
\label{lem:str_adms}
The rules $E$, $U$, $A$, and $A_C$ are admissible in $\iilsbbi$. 
\end{lemma}

\begin{proof}
We show that the said rules can permute upwards through $id$,
$\top^* R$, $\mand R$ and $\mimp L$, the other cases are cover by
Lemma~\ref{lem:str_permute}. We only give some examples here, the others are similar. The heart of the argument is that the application of structural rules are hidden inside the relation entailment, so we do not have to apply them explicitly.

Permute $E$ through $id$, the suppose the original derivation runs as follows.
\begin{center}
\AxiomC{$\mathcal{G};(y,x\simp z);(x,y\simp z) \vdash_R (w_1 = w_2)$}
\RightLabel{\tiny $id$}
\UnaryInfC{$\mathcal{G};(y,x\simp z);(x,y\simp z)||\Gamma;w_1:P \vdash w_2:P;\Delta$}
\RightLabel{\tiny $E$}
\UnaryInfC{$\mathcal{G};(x,y\simp z)||\Gamma;w_1:P \vdash w_2:P;\Delta$}
\DisplayProof
\end{center}

The permuted derivation is:
\begin{center}
\AxiomC{$\mathcal{G};(x,y\simp z) \vdash_R (w_1 = w_2)$}
\RightLabel{\tiny $id$}
\UnaryInfC{$\mathcal{G};(x,y\simp z)||\Gamma;w_1:P \vdash w_2:P;\Delta$}
\DisplayProof
\end{center}

Assume $\mathcal{G};(y,x\simp z);(x,y\simp z) \vdash_R (w_1 = w_2)$ is
derived by applying a sequence $\sigma$ of structural rules. Then
$\Scal((\mathcal{G};(x,y\simp z)),\sigma')$ can prove
$\mathcal{G};(x,y\simp z) \vdash_R (w_1 = w_2)$, where $\sigma'$ is $E(\{(x,y\simp z)\},\emptyset)$ followed by $\sigma$. That is, the application of $E$ is absorbed in $\vdash_R$. 

Permute $A$ through $id$, the argument is similar. The original derivation is:
\begin{center}
\footnotesize
\AxiomC{$\mathcal{G};(u,w\simp z);(y,v\simp w);(x,y\simp z);(u,v\simp x') \vdash_R (w_1 = w_2)$}
\RightLabel{\tiny $id$}
\UnaryInfC{$\mathcal{G};(u,w\simp z);(y,v\simp w);(x,y\simp z);(u,v\simp x')||\Gamma;w_1:P \vdash w_2:P;\Delta$}
\RightLabel{\tiny $A$}
\UnaryInfC{$\mathcal{G};(x,y\simp z);(u,v\simp x')||\Gamma;w_1:P \vdash w_2:P;\Delta$}
\DisplayProof
\end{center}

The condition on the rule $A$ is $\mathcal{G};(x,y\simp z);(u,v\simp x')\vdash_E (x = x')$. Then we can omit the application of $A$, since $\mathcal{G};(u,w\simp
z);(y,v\simp w);(x,y\simp z);(u,v\simp x') \vdash_R (w_1 = w_2)$
implies $\mathcal{G};(x,y\simp z);(u,v\simp x') \vdash_R (w_1 = w_2)$,
one just need to add the $A$ application ahead to the sequence of structural
rules that derives the former relation entailment to get a new
sequence of rules to derive the latter one.\qed
\end{proof}

\subsection{Soundness of $\fvlsbbi$}
\label{app:sound_fvlsbbi}

Proof of Theorem~\ref{thm:sound_fvlsbbi}.

\begin{proof}
By induction on the height $n$ of derivation $\Pi$.

\begin{enumerate}

\item Base case: $n = 1$. In this case, we can only use a zero-premise rule
to prove the sequent. Since the sequent is ground, there are no free
variables. Thus the constraint generated by the rule application 
is a simple constraint, of the form
$
\Gcal \vdash^?_R (a = b)
$
or 
$\Gcal \vdash^?_R (a = \epsilon).$
A solution of this constraint is simply a derivation $\sigma$ of
$\Gcal \vdash_R (a = b)$ (resp. $\Gcal \vdash (a = \epsilon).$
In either case, this translates straightforwardly into a derivation
in $\iilsbbi$ with the same rule.

\item 
Inductive case: $n > 1$. This can be done by a case analysis of the
last rule application in $\Pi$. We demonstrate the case for $\mand R$,
where a constraint is generated. The case for $\mimp L$ is analogous,
and the other cases are easy since we can use the induction hypothesis
directly. Suppose $\Pi$ runs as follows.

\begin{center}
\AxiomC{$\Pi_1$}
\alwaysNoLine
\UnaryInfC{$\mathcal{G}||\Gamma \vdash \fvx:A; w:A\mand B;\Delta$}
\AxiomC{$\Pi_2$}
\UnaryInfC{$\mathcal{G}||\Gamma \vdash \fvy:B; w:A\mand B;\Delta$}
\alwaysSingleLine
\RightLabel{\tiny $\mand R$}
\BinaryInfC{$\mathcal{G}||\Gamma \vdash w:A\mand B;\Delta$}
\DisplayProof
\end{center}

Suppose $\Cbb(\Pi) = (\{\cfr_1,\dots,\cfr_k\}, \preceq)$, for some $k \geq 1.$
Suppose that the constraint generated by this rule application is $\Gcal \vdash^?_R (\fvx,\fvy\simp w)$
and it corresponds to $\cfr_i$ for some $i \in \{1,\dots,k\}$. 
By the assumption, there is a solution $(\theta,\{\sigma_1,\cdots,\sigma_k\})$ for the constraint
system $\Cbb = (\Ccal(\Pi),\preceq^{\Pi})$. Now $\cfr_i$ must be a simple
constraint in $\Cbb$, as the end sequent is ground. Let
$(\theta_i,\sigma_i)$ be the solution to $\cfr_i$, where $\theta_i$ is a
restriction to $\theta$ containing $\fvx$ and $\fvy$, and
$\sigma_i\in \{\sigma_1,\cdots,\sigma_k\}$. By definition of the
solution to a simple constraint, $\sigma_i$ is a derivation of
$\Gcal \vdash_R (\fvx\theta_i,\fvy\theta_i\simp w)$. Therefore in
$\iilsbbi$, to derive the end sequent, we apply $\mand R$ backwards:

\begin{center}
\footnotesize
\AxiomC{$\Scal(\mathcal{G},\sigma_i)||\Gamma \vdash \fvx\theta_i:A; w:A\mand B;\Delta$}
\AxiomC{$\Scal(\mathcal{G},\sigma_i)||\Gamma \vdash \fvy\theta_i:B; w:A\mand B;\Delta$}
\RightLabel{\tiny $\mand R$}
\BinaryInfC{$\mathcal{G}||\Gamma \vdash w:A\mand B;\Delta$}
\DisplayProof
\end{center}

The condition on this rule is $\Gcal \vdash_R (\fvx\theta_i,\fvy\theta_i\simp w)$. Now we construct the derivation for both branches in the following
way. Firstly we substitute $\fvx$ and $\fvy$ with $\fvx\theta_i$ and
$\fvy\theta_i$ respectively in $\Pi_1$ and $\Pi_2$, making the end
sequents in the two derivations ground. Let us refer to the modified
derivations as $\Pi_1'$ and $\Pi_2'$ respectively. Then for each
sequent in $\Pi_1'$ and $\Pi_2'$ and each constraint in
$\Ccal(\Pi_1')\cup\Ccal(\Pi_2')$, we change the set of relational
atoms to be the union of $\Scal(\Gcal,\sigma_i)$ and the original
one. This is harmless because we can use weakening to obtain the same
sequents as in $\Pi_1'$ and $\Pi_2'$, and weakening is
height-preserving admissible. Let the resultant derivations be
$\Pi_1''$ and $\Pi_2''$ respectively.  Now the end sequents of
$\Pi_1''$ and $\Pi_2''$ are respectively just the same as the two
branches we created in the $\iilsbbi$ derivation. Moreover, each
constraint in $\Ccal(\Pi_1'')\cup\Ccal(\Pi_2'')$ is in the restricted
constraint system $\Cbb' =
(\Ccal(\Pi),\preceq^{\Pi})\uparrow(\cfr_i,\theta_i,\sigma_i)$, which has
a solution
$(\theta\setminus\theta_i,\{\sigma_1,\cdots,\sigma_k\}\setminus \sigma_i)$,
and obeys the partial order $\preceq'$. Further, as $\Pi''_1$
(resp. $\Pi''_2$) uses the same rule applications as in $\Pi_1$
(resp. $\Pi_2$), the order of constraints is preserved. That is, in
the constraints system $\Cbb_1 = (\Ccal(\Pi''_1),\preceq^{\Pi''_1})$
(resp. $\Cbb_2 = (\Ccal(\Pi''_2),\preceq^{\Pi''_2})$), if
$\cfr \preceq^{\Pi''_1} \cfr'$ (resp. $\cfr \preceq^{\Pi''_2} \cfr'$) then
$\cfr \preceq' \cfr'$ in $\Cbb'$. Therefore we can construct the solution
$(\theta_1'',\Sigma_1)$ to $\Cbb_1$ (and analogously to $\Cbb_2$) as
follows.
\begin{align*}
\theta_1'' & = (\theta\setminus\theta_i)\uparrow fv(\Ccal(\Pi''_1))\\
\Sigma_1 & = \{\sigma \mid \cfr \in \Ccal(\Pi''_1), \sigma \in \{\sigma_1,\cdots,\sigma_k\}\setminus\sigma_i, 
\ and \ \sigma = dev(\cfr) \}
\end{align*}

By the induction hypothesis, we can obtain a $\iilsbbi$ derivation for each branch.\qed

\end{enumerate}
\end{proof}

\subsection{Completeness of $\fvlsbbi$}
\label{app:comp_fvlsbbi}

Proof of Theorem~\ref{thm:comp_fvlsbbi}.

\begin{proof}
We describe the construction from 
a $\iilsbbi$ derivation $\Pi$ to a $\fvlsbbi$ derivation $\Pi'$. 
We need to prove a stronger invariant: for each sequent $\Gcal_E ; \Gcal_S || \Gamma \vdash \Delta$ in $\Pi$, if there exists a triple
consisting of: 
\begin{itemize}
\item a symbolic sequent 
$\Gcal_E' || \Gamma' \vdash \Delta'$, 
\item a well-formed constraint system $\Cbb = (\Ccal,\preceq)$, 
\item and a solution $S = (\theta, \{\overset{\rightharpoonup}{\sigma}\})$ to $\Cbb$
\end{itemize} 
such that
\begin{itemize}
\item $X$ is a thread of $\Cbb$ consisted of $fv(\Gcal_E' || \Gamma' \vdash \Delta')$, 
\item $\Gcal_E'\theta = \Gcal_E$, $\Gamma'\theta =\Gamma$, $\Delta'\theta = \Delta$ and 
\item $\Gcal_E\cup\Gcal_S = \Scal^*(\Cbb, S, X)$, 
\end{itemize}
then there is a symbolic derivation $\Psi$ of $\Gcal_E' || \Gamma' \vdash \Delta'$
such that $\Cbb \circ^X \Cbb(\Psi)$ is well-formed and solvable. 

First of all, by Lemma~\ref{lem:cons_merge}, since the end sequent in $\Psi$ only contains the free variables occur in $X$, the composition $\Cbb \circ^X \Cbb(\Psi)$ must be well-formed. Thus we only need to show that there is a solution to this constraint system. We prove this by case analysis on the last rule in $\Pi$, and show that in each
case, for each premise of the rule, one can find a triple satisfying
the above property, such that the symbolic sequent(s) in the premise(s),
together with the one in the conclusion form a valid inference in $\fvlsbbi.$
We illustrate it here with a case when $\Pi$ ends with $\mand R$:

Suppose $\Pi$ ends with $\mand R$, where the conclusion, the premises, and the relational entailment are
respectively:
\begin{itemize}
\item $\mathcal{G}_E;\mathcal{G}_S||\Gamma \vdash w:A\mand B;\Delta$
\item $\mathcal{S}((\mathcal{G}_E;\mathcal{G}_S),\sigma)||\Gamma \vdash w_1:A;w:A\mand B;\Delta$
\item $\mathcal{S}((\mathcal{G}_E;\mathcal{G}_S),\sigma)||\Gamma \vdash w_2:B;w:A\mand B;\Delta$
\item $\mathcal{G}_E;\mathcal{G}_S \vdash_R (w_1,w_2\simp w)$
\end{itemize}
and suppose that the relation in the last item is derived via $\sigma.$
Suppose that we can find a triple consisting of
\begin{itemize}
\item a symbolic sequent $\Gcal_E' || \Gamma' \vdash \fvw : A \mand B ; \Delta'$
\item a well-formed constraint system $\Cbb = (\Ccal,\preceq)$, and
\item a solution $S = (\theta, \{\sigma_1,\dots,\sigma_n\})$ to $\Cbb$
\end{itemize}
satisfying the following:
\begin{itemize}
\item $X$ is a thread of $\Cbb$ consisted of $fv(\Gcal_E' || \Gamma' \vdash \fvw: A\mand B;\Delta')$, 
\item $\Gcal_E'\theta = \Gcal_E$, $\Gamma'\theta =\Gamma$, $\Delta'\theta = \Delta$, $w = \fvw\theta$ and 
\item $\Gcal_E\cup\Gcal_S = \Scal^*(\Cbb, S, X)$. 
\end{itemize}
We need to show that we can find such triples for the premises, and more importantly,
the symbolic sequents in the premises are related to the symbolic sequent in the conclusion 
via $\mand R$.
In this case, the symbolic sequents are simply the following:
\begin{enumerate}
\item $\Gcal_E' || \Gamma' \vdash \fvx : A ; \fvw : A \mand B ; \Delta'$, for the left premise, 
\item $\Gcal_E' || \Gamma' \vdash \fvy : A ; \fvw : A \mand B ; \Delta'$, for the right premise. 
\end{enumerate}
The constraint systems are: 
$\Cbb' = (\Ccal \cup \{\cfr_j\},\preceq')$ for both premises, where $\cfr_j = \Gcal_E' \vdash_R^? (\fvx, \fvy \simp \fvw)$ and $\preceq'$ is $\preceq$ extended with $\cfr(end(X)) \preceq' \cfr_j$. 
The solutions, for both premises, are the tuple $S' = (\theta', \Sigma)$
where $\theta' = \theta \cup \{ \fvx \mapsto w_1, ~ \fvy \mapsto
w_2\}$ and $\Sigma = \{\sigma_1,\dots,\sigma_n, \sigma\}$. It is guaranteed that $\theta'$ is enough to make both premises grounded, as $\fvx$ and $\fvy$ are the only two new free variables. The threads of free variables $X_1$ and $X_2$ for the two premises are naturally $X@[\fvx]$ and $X@[\fvy]$ respectively. By Proposition~\ref{prop:scalstar}, in each premise, the following holds:
\begin{center}
$\Gcal_E\cup\Gcal_S' =
\Scal(\Gcal_E\cup\Gcal_S,\sigma) = 
\Scal(\Gcal_E\cup\Gcal_S\cup\Gcal_E'\theta,\sigma) = \Scal^*(\Cbb',S',X_1) = \Scal^*(\Cbb',S',X_2)$. 
\end{center}
So by the induction hypothesis
we have a symbolic derivation $\Pi_1'$ for sequent (1) and a symbolic derivation $\Pi_2'$ for
sequent (2), such that $\Cbb_{\beta1} = \Cbb' \circ^{X_1} \Cbb(\Pi_1')$ and $\Cbb_{\beta2} = \Cbb' \circ^{X_2} \Cbb(\Pi_2')$ are both
solvable. Suppose the solutions are respectively $(\theta'\cup\theta_1,\Sigma\cup\Sigma_1)$ and $(\theta'\cup\theta_2,\Sigma\cup\Sigma_2)$. Then construct $\Pi'$ by applying the $\mand R$ rule to $\Pi_1'$ and $\Pi_2'$.
Note that the variables created in $\Pi_1'$ are $\Pi_2'$ are distinct so their constraints
are independent of each other. So we can construct $\Cbb_p = \Cbb(\Pi_1')\circ^{\emptyset}\Cbb(\Pi_2') = (\Ccal_p,\preceq_p)$, along an empty thread $\emptyset$. Now $\Cbb(\Pi')$ is obtained as $(\Ccal_p\cup\{\cfr_j\},\preceq^{\Pi'})$, where $\preceq^{\Pi'}$ is derived as follows.
\begin{itemize}
\item If $\cfr \preceq_p \cfr'$ in $\Cbb_p$, then $\cfr \preceq^{\Pi'} \cfr'$ in $\Cbb(\Pi')$
\item For any minimum constraint $\cfr_m$ in $\Cbb_p$, $\cfr_j\preceq^{\Pi'} \cfr_m$ in $\Cbb(\Pi')$
\end{itemize}

The solution to $\Cbb_{\alpha} = \Cbb\circ^X\Cbb(\Pi')$ is constructed as the combination of the solutions to $\Cbb_{\beta1}$ and $\Cbb_{\beta2}$: $(\theta'\cup\theta_1\cup\theta_2,\Sigma\cup\Sigma_1\cup\Sigma_2)$. This construction of the solution is indeed valid, because the symbolic derivation that gives $\Cbb_{\alpha}$ also yields exactly $\Cbb_{\beta1}$ and $\Cbb_{\beta2}$ (respectively on its two branches created by the $\mand R$ rule).

\qed
\end{proof}

\subsection{The Proof of the Heuristic Method}
\label{app:heuristics}

In the following proofs we use the tree representation of a set of relational atoms. Given a labelled binary tree $tr$ as defined in Section~\ref{sec:heur-proof-search}, we say another labelled binary tree $tr'$ is a permutation of $tr$ if they have the same root and same multiset of leaves. A permutation on $tr$ is generally done by applying the rules $E,A$ on $\Rel(tr)$. Figure~\ref{fig:tree_permute} gives some examples on tree permutations. In Figure~\ref{fig:tree_permute}, (b) is permuted from (a) by using $E$ on $(d,e\simp b)$, whereas (c) is permuted from (a) by using $A$ on the two relational atoms in the original tree.
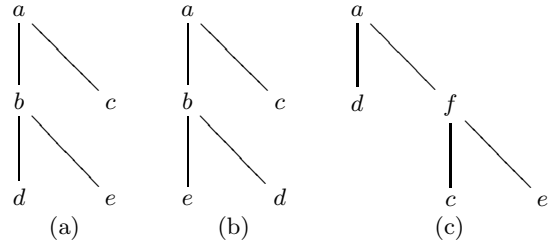
\begin{figure}[ht!]
\begin{center}
\footnotesize
\begin{tabular}{c@{\qquad}c@{\qquad}c@{\qquad}}
\xymatrix{
a \ar@{-}[d] \ar@{-}[dr]\\ 
b \ar@{-}[d] \ar@{-}[dr] & c\\
d & e
}
&
\xymatrix{
a \ar@{-}[d] \ar@{-}[dr]\\ 
b \ar@{-}[d] \ar@{-}[dr] & c\\
e & d
}
&
\xymatrix{
a \ar@{-}[d] \ar@{-}[dr]\\ 
d & f \ar@{-}[d] \ar@{-}[dr]\\
& c & e
}
\\
(a) & (b) & (c)
\end{tabular}
\end{center}
\caption{Examples of tree permutations.}
\label{fig:tree_permute}  
\end{figure}

\begin{lemma}
\label{lem:btree_permute}
Let $tr$ be a labelled binary tree with a root labelled with $r$ and a
multiset of labels $L$ for the leaves. If there is a labelled binary
tree $tr'$ with the same root and leaves labels respectively, then
there is a variant $tr''$ of $tr'$ and a sequence $\sigma$ of $E,A$ rule applications such that
$\Rel(tr'')\subseteq \Scal(\Rel(tr),\sigma).$ 
\end{lemma}

\begin{proof}
Prove by induction on the width of the tree $tr$. 
We show that any distinct permutation(i.e., they
are not variants of each other) of a tree can be achieved by using the rules $E$ and $A$.
Base case is when
there are only two leaves in $tr$. In this case, there is only one relational atom in $\Rel(tr)$, thus clearly there is only one distinct permutation of $tr$, which can be obtained by applying $E$ on $\Rel(tr)$.

The next case is when there are 3 leaves in the tree, meaning $\Rel(tr)$ contains two relational atoms. In this case, it can be easily 
checked that there are 12 distinct permutations of $tr$, all of which can be derived by using $E$ and $A$.

Inductive case, suppose the lemma holds for all trees with width less
than $n$, consider a tree $tr$ with width $n$. Suppose further that
the root label of $tr$ is $r$, it's two children are in the
relational atom $(w_1,w_2\simp r)$, and the multisets of leaves labels
for the subtrees of $w_1$ and $w_2$ are $L_1$, $L_2$ respectively. Let $tr'$ be a
permutation of $tr$ with the same root label and leaves labels, and in
$tr'$ the two children of the root label are in the relational atom
$(w_3,w_4\simp r)$. Suppose the multisets of leave labels for the
subtrees of $w_3,w_4$ are $L_3,L_4$ respectively. Apparently, since
$L_1\cup L_2 = L_3\cup L_4 = L$, every label in $L_3$ is either in
$L_1$ or in $L_2$. Let $L' = L_1\cap L_3$ and $L'' = L_2\cap L_3$,
then $L'\cup L'' = L_3$ and $(L_1\setminus L')\cup (L_2\setminus L'')
= L_4$. By
the induction hypothesis on the subtrees of $w_1$ and $w_2$, there
exist $w_5,w_6,w_7,w_8$ s.t. $(w_5,w_6\simp w_1),(w_7,w_8\simp w_2)$
hold, and the subtrees of $w_5,w_6,w_7,w_8$ give the multisets of
leaves $L', (L_1\setminus L'), L'',(L_2\setminus L'')$
respectively. Then we use the following derivation to permute the tree:
\begin{center}
\AxiomC{$(w'',w'''\simp r);(w_6,w_8\simp w''');(w_5,w_7\simp w'');\cdots$}
\RightLabel{\tiny $A$}
\UnaryInfC{$(w',w_6\simp r);(w'',w_8\simp w');(w_5,w_7\simp w'');\cdots$}
\RightLabel{\tiny $E \times 2$}
\UnaryInfC{$(w_6,w'\simp r);(w_8,w''\simp w');(w_5,w_7\simp w'');\cdots$}
\RightLabel{\tiny $A$}
\UnaryInfC{$(w_6,w'\simp r);(w_2,w_5\simp w');(w_8,w_7\simp w_2);\cdots$}
\RightLabel{\tiny $E$}
\UnaryInfC{$(w_6,w'\simp r);(w_2,w_5\simp w');(w_7,w_8\simp w_2);\cdots$}
\RightLabel{\tiny $A$}
\UnaryInfC{$(w_6,w_5\simp w_1);(w_7,w_8\simp w_2);(w_1,w_2\simp r);\cdots$}
\RightLabel{\tiny $E$}
\UnaryInfC{$(w_5,w_6\simp w_1);(w_7,w_8\simp w_2);(w_1,w_2\simp r);\cdots$}
\DisplayProof
\end{center}
Now the subtrees of $w''$ and $w'''$ has the same multisets of leaves
as $w_3$ and $w_4$ respectively. Again by the induction hypothesis on
the subtrees of $w''$ and $w'''$, we obtain a tree $tr''$ which is
a variant of $tr'$.\qed
\end{proof}

Proof of Lemma~\ref{lem:heuristics}.
\begin{proof}
The lemma restricts the labels of internal nodes to be free variables
that are created after all the labels on the left hand side. Additionally, each free variable is only allowed to occur once in a tree. Therefore
given a set $\Gcal$ of relational atoms as the left hand side of those
constraints, and any sequence $\sigma$ of structural rule
applications, the free variable labels for internal nodes can be
assigned to any labels occur in $\Scal(\Gcal,\sigma)$. By
Lemma~\ref{lem:btree_permute}, there exists a sequence $\sigma$ of
$E,A$ applications which converts
the tree on the left hand side to a tree which is a variant of the one
on the right hand side,
thus those constraints can be solved by assigning the free variables
in the internal nodes to the corresponding labels.\qed
\end{proof}

\subsection{Proof of Formulae in the Conclusion}
\label{app:formulae_conc}
In this section we show the proofs of the four formulae in the
conclusion. We extend $\lsbbi$ in the obvious way to handle the
additive connectives $\lnot$ and $\lor$, where $\lnot p = p\limp \bot$ and $p\lor q = \lnot(\lnot p\land \lnot q)$. Thus we obtain the left and right rules for $\lnot,\lor$ as in the classical setting. To save space, we shall write $r^n$ to mean the rule $r$ is applied $n$ times, and write $r_1;r_2$ to mean apply $r_1$ then apply $r_2$ on a sequent, when the order of rule applications does not matter.
\begin{enumerate}
\item To prove the formula $(F * F) \limp F$, where $F =
\lnot(\top \mimp \lnot \top^*)$, we use the following derivation in
$\lsbbi$: 
\begin{center}
\small
\AxiomC{$(w',w''\simp \epsilon);(b',c'\simp w'');(b,c\simp w');(b,c\simp a);\cdots$}
\RightLabel{$A$}
\UnaryInfC{$(w',c'\simp w);(w,b'\simp \epsilon);\cdots$}
\RightLabel{$E^2$}
\UnaryInfC{$(c',w'\simp w);(b,c\simp w');(b',w\simp \epsilon);\cdots$}
\RightLabel{$A$}
\UnaryInfC{$(b',w\simp \epsilon);(\epsilon,b\simp w);(c',c\simp \epsilon);\cdots$}
\RightLabel{$A$}
\UnaryInfC{$(b,c\simp a);(b',b\simp \epsilon);(c',c\simp \epsilon);(\epsilon,\epsilon\simp \epsilon);\cdots$}
\RightLabel{$U$}
\UnaryInfC{$(b,c\simp a);(b',b\simp \epsilon);(c',c\simp \epsilon);a:\top\mimp\lnot\top^*;b':\top,c':\top\vdash$}
\RightLabel{$\top^* L^2$}
\UnaryInfC{$(b,c\simp a);(b',b\simp b'');(c',c\simp c'');a:\top\mimp\lnot\top^*;b':\top,c':\top;b":\top^*;c'':\top^*\vdash$}
\RightLabel{$\lnot R^2$}
\UnaryInfC{$(b,c\simp a);(b',b\simp b'');(c',c\simp c'');a:\top\mimp\lnot\top^*;b':\top,c':\top\vdash\b'':\lnot\top^*;c'':\lnot\top^*$}
\RightLabel{$\mimp R^2$}
\UnaryInfC{$(b,c\simp a); a:\top\mimp\lnot\top^*\vdash b:\top\mimp\lnot\top^*;c:\top\mimp\lnot\top^*$}
\RightLabel{$\lnot L^2;\lnot R$}
\UnaryInfC{$(b,c\simp a);b:\lnot(\top\mimp\lnot\top^*);c:\lnot(\top\mimp\lnot\top^*)\vdash a:\lnot(\top\mimp\lnot\top^*)$}
\RightLabel{$\mand L$}
\UnaryInfC{$a:F\mand F\vdash a:F$}
\RightLabel{$\limp R$}
\UnaryInfC{$\vdash a:(F\mand F)\limp F$}
\DisplayProof
\end{center}

The correct relational atom that is required to split $a:\top\mimp
\lnot \top^*$ is $(w'',a\simp \epsilon)$. However, in the labelled
sequent calculus we can only obtain $w'',w'\simp \epsilon$. Although
$w'$ and $a$ both have exactly the same children, but the
non-deterministic monoid allows the composition $b\circ c$ to be
multiple elements, or even $\emptyset$ in $\mathcal{M}$. Thus we
cannot conclude that $w' = a$. This can be solved by using $P$ to
replace $w'$ by $a$, then use $E$ to obtain $(w'',a\simp \epsilon)$ on the left hand side of the sequent, then the derivation can go through:
\begin{center}
\small
\AxiomC{}
\RightLabel{$\top^* R$}
\UnaryInfC{$(w'',a\simp \epsilon);\cdots;\vdash \epsilon:\top^*$}
\RightLabel{$\lnot L$}
\UnaryInfC{$(w'',a\simp \epsilon);\cdots;\epsilon:\lnot\top^*\vdash$}
\AxiomC{}
\RightLabel{$\top R$}
\UnaryInfC{$(w'',a\simp \epsilon);\cdots\vdash w'':\top$}
\RightLabel{$\mimp L$}
\BinaryInfC{$(w'',a\simp \epsilon);\cdots;a:\top\mimp\lnot\top^*;b':\top,c':\top\vdash$}
\DisplayProof
\end{center}

\item The trick to prove $(\lnot \top^* \mimp \bot) \limp \top^*$ is to
create a relational atom $(w,w\simp w')$, as shown below.
\begin{center}
\small
\AxiomC{$$}
\RightLabel{$\top^* R$}
\UnaryInfC{$(\epsilon,\epsilon\simp w');\cdots\vdash \epsilon:\top^*$}
\RightLabel{$\top^* L$}
\UnaryInfC{$(w,w\simp w');\cdots;w:\top^*\vdash w:\top^*$}
\RightLabel{$\lnot R$}
\UnaryInfC{$(w,w\simp w');\cdots\vdash w:\lnot
  \top^*;w:\top^*$}
\AxiomC{$$}
\RightLabel{$\bot L$}
\UnaryInfC{$(w,w\simp w');\cdots;w':\bot\vdash w:\top^*$}
\RightLabel{$\mimp L$}
\BinaryInfC{$(w,w\simp w');w:\lnot \top^* \mimp \bot\vdash w:\top^*$}
\RightLabel{$T$}
\UnaryInfC{$w:\lnot \top^* \mimp \bot\vdash w:\top^*$}
\RightLabel{$\limp R$}
\UnaryInfC{$\vdash w: (\lnot \top^* \mimp \bot) \limp \top^*$}
\DisplayProof
\end{center}

\item The proof for $(\top^* \land ((p \mand q)
\mimp \bot)) \limp ((p \mimp \bot) \lor (q \mimp \bot))$ is as
follows.
\begin{center}
\small
\AxiomC{}
\RightLabel{$id$}
\UnaryInfC{$\cdots;c:q\vdash c:q;\cdots$}
\AxiomC{}
\RightLabel{$id$}
\UnaryInfC{$\cdots;a:p\vdash a:p;\cdots$}
\RightLabel{$\mand R$}
\BinaryInfC{$(a,c\simp e);\cdots;a:p;c:q \vdash e:p\mand q;\cdots$}
\AxiomC{}
\RightLabel{$\bot L$}
\UnaryInfC{$\cdots e:\bot\vdash\cdots$}
\RightLabel{$\mimp L$}
\BinaryInfC{$(e,\epsilon\simp e);(a,c\simp e);(a,\epsilon\simp b);(c,\epsilon\simp d);\epsilon:(p \mand q)
\mimp \bot;a:p;c:q \vdash b:\bot;d:\bot$}
\RightLabel{$U$}
\UnaryInfC{$(a,c\simp e);(a,\epsilon\simp b);(c,\epsilon\simp d);\epsilon:(p \mand q)
\mimp \bot;a:p;c:q \vdash b:\bot;d:\bot$}
\RightLabel{$T$}
\UnaryInfC{$(a,\epsilon\simp b);(c,\epsilon\simp d);\epsilon:(p \mand q)
\mimp \bot;a:p;c:q \vdash b:\bot;d:\bot$}
\RightLabel{$\mimp R^2$}
\UnaryInfC{$\epsilon:(p \mand q)
\mimp \bot \vdash \epsilon:p \mimp \bot; \epsilon:q \mimp \bot$}
\RightLabel{$\top^* L$}
\UnaryInfC{$w:\top^*; w:(p \mand q)
\mimp \bot \vdash w:p \mimp \bot; w:q \mimp \bot$}
\RightLabel{$\land L;\lor R$}
\UnaryInfC{$w:\top^* \land ((p \mand q)
\mimp \bot) \vdash w:(p \mimp \bot) \lor (q \mimp \bot)$}
\RightLabel{$\limp R$}
\UnaryInfC{$\vdash w:(\top^* \land ((p \mand q)
\mimp \bot)) \limp ((p \mimp \bot) \lor (q \mimp \bot))$}
\DisplayProof
\end{center}

\item The proof for $\lnot (\top^* \land A \land (B \mand \lnot (C \mimp
(\top^* \limp A))))$ in $\lsbbi$ is as follows.
\begin{center}
\small
\AxiomC{}
\RightLabel{$id$}
\UnaryInfC{$(c,b\simp \epsilon);(a,b\simp \epsilon);\epsilon:A;a:B;c:C\vdash \epsilon:A$}
\RightLabel{$\top^* L$}
\UnaryInfC{$(c,b\simp d);(a,b\simp \epsilon);\epsilon:A;a:B;c:C;d:\top^*\vdash d:A$}
\RightLabel{$\limp R$}
\UnaryInfC{$(c,b\simp d);(a,b\simp \epsilon);\epsilon:A;a:B;c:C\vdash d:
\top^* \limp A$}
\RightLabel{$\mimp R$}
\UnaryInfC{$(a,b\simp \epsilon);\epsilon:A;a:B\vdash b:C \mimp
(\top^* \limp A)$}
\RightLabel{$\lnot L$}
\UnaryInfC{$(a,b\simp \epsilon);\epsilon:A;a:B; b:\lnot (C \mimp
(\top^* \limp A))\vdash$}
\RightLabel{$\mand L$}
\UnaryInfC{$\epsilon:A;\epsilon:B \mand \lnot (C \mimp
(\top^* \limp A))\vdash$}
\RightLabel{$\top^* L$}
\UnaryInfC{$w:\top^*;w:A;w:B \mand \lnot (C \mimp
(\top^* \limp A))\vdash$}
\RightLabel{$\land L$}
\UnaryInfC{$w:\top^* \land A \land (B \mand \lnot (C \mimp
(\top^* \limp A)))\vdash$}
\RightLabel{$\lnot R$}
\UnaryInfC{$\vdash w:\lnot (\top^* \land A \land (B \mand \lnot (C \mimp
(\top^* \limp A))))$}
\DisplayProof
\end{center}
\end{enumerate}

\end{document}